\documentclass[11pt]{article}
\usepackage{amsmath,amsfonts,amssymb,amsthm}
\usepackage{fullpage}
\usepackage[ruled]{algorithm2e}

\usepackage{subfigure} \usepackage{epstopdf}
\usepackage{graphicx}
\usepackage{color}
\usepackage[sans]{dsfont}
\usepackage{mathtools}
\usepackage{hyperref}
\usepackage{cleveref}
\usepackage{tcolorbox}
\usepackage{multirow}
\usepackage[shortlabels]{enumitem}
\usepackage{xspace}
\tcbuselibrary{skins,breakable}
\tcbset{enhanced jigsaw}
\usepackage[font=small,labelfont=bf]{caption}

\newtheorem{theorem}{Theorem}

\newtheorem{lemma}[theorem]{Lemma}
\newtheorem{observation}[theorem]{Observation}

\theoremstyle{definition}
\newtheorem{claim}[theorem]{Claim}

\newtheorem{remark}[theorem]{Remark}

\newcommand{\set}[1]{\left\{ #1 \right\}}

\newcommand{\pset}{{\mathcal{P}}}

\newcommand{\dset}{{\mathcal{D}}}

\newcommand{\fset}{{\mathcal{F}}}

\newcommand{\eps}{{\varepsilon}}
\newcommand{\poly}{\mathsf{poly}}

\renewcommand{\cong}{\mathsf{cong}}

\newcommand{\dist}{\textnormal{\textsf{dist}}}

\newcommand\vol{\mathsf{cost}}

\def\card#1{\left| #1 \right|}
\def\deg{\textrm{deg}}

\def\pr#1{\mathrm{Pr}\left[ #1 \right]}

\def\ex#1{{\mathbb{E}}\left[ #1 \right]}

\newcommand{\TS}{\mathsf{TS}}

\newcounter{note}

\newcommand{\zesn}{\mathsf{0EwSN}}

\begin{document}

\begin{titlepage}
	
	\title{On $(1+\eps)$-Approximate Flow Sparsifiers}

	\author{Yu Chen\thanks{EPFL, Lausanne, Switzerland. Email: {\tt yu.chen@epfl.ch}. Supported by ERC Starting Grant 759471.} \and Zihan Tan\thanks{Rutgers University, NJ, USA. Email: {\tt zihantan1993@gmail.com}. Supported by a grant to DIMACS from the Simons Foundation (820931).}} 
	
	\maketitle

	\thispagestyle{empty}
	\begin{abstract}
		
Given a large graph $G$ with a subset $|T|=k$ of its vertices called terminals, a \emph{quality-$q$ flow sparsifier} is a small graph $G'$ that contains $T$ and preserves all multicommodity flows that can be routed between terminals in $T$, to within factor $q$. The problem of constructing flow sparsifiers with good (small) quality and (small) size has been a central problem in graph compression for decades.

A natural approach of constructing $O(1)$-quality flow sparsifiers, which was adopted in most previous constructions, is contraction.
Andoni, Krauthgamer, and Gupta constructed a sketch of size $f(k,\eps)$ that stores all feasible multicommodity flows up to a factor of $(1+\eps)$, raised the question of constructing quality-$(1+\eps)$ flow sparsifiers whose size only depends on $k,\eps$ (but not the number of vertices in the input graph $G$), and proposed a contraction-based framework towards it using their sketch result. 

In this paper, we settle their question for contraction-based flow sparsifiers, by showing that quality-$(1+\eps)$ contraction-based flow sparsifiers with size $f(\eps)$ exist for all $5$-terminal graphs, but not for all $6$-terminal graphs.
Our hardness result on $6$-terminal graphs improves upon a recent hardness result by Krauthgamer and Mosenzon on exact (quality-$1$) flow sparsifiers, for contraction-based constructions. Our construction and proof utilize the notion of \emph{tight spans} in metric geometry, which we believe is a powerful tool for future work.

	\end{abstract}
\end{titlepage}

\renewcommand{\baselinestretch}{0.75}\normalsize
\tableofcontents
\renewcommand{\baselinestretch}{1.0}\normalsize

\newpage

\section{Introduction}

Graph compression is a paradigm of converting large graphs into smaller ones that faithfully preserve crucial features, such as flow/cut values and distances. It involves reducing the size of graphs prior to subsequent computation, and thereby significantly saving computational resources. 
This paradigm has proved powerful in designing faster and better approximation algorithms on graphs.

We study a sub-paradigm of graph compression called \emph{vertex sparsification}, and more specifically, we study flow-approximating\footnote{A closely related notion is cut-approximating vertex sparsifiers. We review the previous work on cut sparsifiers and discuss the connection between cut and flow sparsifiers in \Cref{sec: related}.} vertex sparsifiers first introduced in \cite{hagerup1998characterizing,moitra2009approximation,leighton2010extensions}.
In this setting, we are given a large graph $G$ together with a set $T$ of $k$ designated vertices called \emph{terminals}, and the goal is to compute a small graph $G'$ that contains $T$ and preserves all multicommodity flows that can be routed between terminals in $T$.
Informally\footnote{A formal definition is provided in \Cref{sec: prelim}.}, we say that $G'$ is a \emph{flow sparsifier} of $G$ with respect to $T$, with \emph{quality $q$} for some real number $q>1$, iff every multicommodity flow on $T$ that is routable in $G$ can be routed in $G'$, and every multicommodity flow on $T$ that is routable in $G'$ can be routed in $G$ if the capacities of edges in $G$ are increased by factor $q$.
Ideally, we would like to construct flow sparsifiers $G'$ with good (small) quality and small size (measured by the number of vertices in $G'$).

In a restricted case where $V(G')=T$ is required (that is, the sparsifier may only contain terminals), 
it was shown by Leighton and Moitra \cite{leighton2010extensions} that the quality-$O(\log k/\log\log k)$ flow sparsifiers exist, and Charikar, Leighton, Li and Moitra \cite{charikar2010vertex} showed that they can be computed efficiently.
On the negative side, a lower bound of $\Omega(\log\log k)$ on quality was proved in \cite{leighton2010extensions}, and this bound was later improved to $\Omega(\sqrt{\log k/\log\log k})$ by Makarychev and Makarychev \cite{makarychev2010metric}.
Thus, the next question, which has also been a central question on flow/cut sparsifiers over the past years, is:

\vspace{-10pt}
\[\emph{Can better quality sparsifiers be achieved by allowing a small number of Steiner vertices}? 
\]
\vspace{-10pt}

A natural way of constructing flow sparsifiers is by contraction. That is, we compute a partition $\fset$ of the vertices in $V(G)$ into disjoint sets, and then contract each set $F\in \fset$ into a supernode to obtain $G'$.
We call such a graph $G'$ a \emph{contraction-based flow sparsifier}.
To the best of our knowledge, most previous constructions of flow sparsifiers with Steiner nodes are contraction-based.
Chuzhoy \cite{chuzhoy2012vertex} showed that there exist $O(1)$-quality contraction-based flow sparsifiers with size $C^{O(\log\log C)}$, where $C$ is the total capacity of all terminal-incident edges (assuming each edge has capacity at least $1$).
Andoni, Gupta, and Krauthgamer \cite{andoni2014towards} showed the construction of quality-$(1+\eps)$ flow sparsifiers for quasi-bipartite graphs and exact (quality-$1$) contraction-based sparsifiers for planar graphs where all terminals lie on the same face (where they used the results of \cite{krauthgamer2013mimicking}). 
For general graphs, they constructed a sketch of size $f(k,\eps)$ that stores all feasible multicommodity flows up to factor $(1+\eps)$, raising the hope for quality-$(1+\eps)$ flow sparsifiers of size $f(k,\eps)$ for general graphs. 
On the negative side, the only lower bound, due to Krauthgamer and Mosenzon \cite{krauthgamer2023exact}, showed that there exist $6$-terminal graphs whose quality-$1$ flow sparsifiers must have an arbitrarily large size. Their construction can be also modified to show that quality-$(1+\eps)$ flow sparsifiers for $k$-terminal networks must contain at least $f(k,\eps)=\Omega(k/\poly(\eps))$ vertices for $k\ge 6$, and they raised proving any upper bound of $f(k,\eps)$ as a challenging open question.

\subsection{Our Results}

In this paper, we make progress on the size of quality-$(1+\eps)$ contraction-based flow sparsifiers for general graphs. We show that, every $5$-terminal network admits a quality-$(1+\eps)$ contraction-based flow sparsifier whose size depends only on $\eps$, while this is not true for all $6$-terminal networks.
Our main results for $5$-terminal and $6$-terminal networks are in sharp contrast with each other, and are formally stated in the following theorems.

\begin{theorem}
\label{main: upper}
Let $\eps>0$ be any real number. Then every graph with $5$ terminals admits a quality-$(1+\eps)$ contraction-based flow sparsifier on $2^{\poly(1/\eps)}$ vertices.
\end{theorem}

Previously, it was shown that every $4$-terminal network admits an exact sparsifier with $O(1)$ vertices (see e.g., \cite{andoni2014towards}), while there is no known upper or lower bound for exact or $(1+\eps)$ flow sparsifiers $5$-terminal networks.

\begin{theorem}
\label{main: lower}
For every positive integer $N$, there exists a graph $G$ with $6$ terminals, such that any contraction-based flow sparsifier of $G$ on at most $N$ vertices has quality at least $1+10^{-18}$.
\end{theorem}

Our results settle the open question on the size upper bound of $(1+\eps)$ flow sparsifiers by \cite{krauthgamer2023exact}, for contraction-based constructions. Our lower bound in \Cref{main: lower} can also be viewed as improving upon the previous lower bound of \cite{krauthgamer2023exact} (where they showed an arbitrarily large size lower bound for exact flow sparsifiers of $6$-terminal networks, and we proved the same result for quality-$(1+\eps)$ flow sparsifiers), again for contraction-based constructions. 
Compared with the results in \cite{andoni2014towards} (that constructed a sketch of size $f(k,\eps)$ for $(1+\eps)$-approximately storing all multicommodity flows), our lower bound in \Cref{main: lower} illustrates that contraction-based constructions, which covered most previous algorithmic results, are not optimal data structures for preserving the flow structure of graphs to within factor $(1+\eps)$.

\subsection{Technical Overview}
\label{sec: tech_overview}

The size lower bound (for exact flow sparsifiers) in \cite{krauthgamer2023exact} was proved by analyzing  the slope of the ``feasible demand polytope'' of certain graphs. However, it seems hard for their approach to give strong lower bounds for quality-$(1+\eps)$ flow sparsifiers. We employ a different approach, more similar to the ones in \cite{moitra2009approximation,leighton2010extensions} and \cite{andoni2014towards}. 

In \cite{moitra2009approximation} and \cite{leighton2010extensions}, the connection between the \emph{$0$-Extension} problem and the construction of flow/cut sparsifier without Steiner nodes was established. In a $0$-Extension instance, we are given an undirected edge-capacitated graph $G=(V,E,c)$, a set $T\subseteq V$ of its vertices called \emph{terminals}, and a metric $D$ on terminals, and the goal is to find a mapping $f: V\to T$ that maps each vertex to a terminal in $T$, such that each terminal is mapped to itself (i.e., $f(t)=t$ for all $t\in T$), and the sum $\sum_{(u,v)\in E}c{(u,v)}\cdot D(f(u),f(v))$ is minimized. Typically, the integrality gap of its \emph{semi-metric relaxation LP} was shown to be an upper bound of the best quality achievable by flow sparsifiers (without Steiner nodes). In order to investigate flow sparsifiers with Steiner vertices, we study the following variant of $0$-Extension, called \emph{$0$-Extension with Steiner Nodes}\footnote{We remark that the version we provide in this section is not the most standard ``Steiner node" generalization of the $0$-Extension problem. We present this version here because it is easier to form a connection between this version and flow sparsifiers. We provide the most standard generalization and some discussions in \Cref{apd: comparison}}.

\paragraph{$0$-Extension with Steiner Nodes ($\zesn$).}
In an instance of the $\zesn$ problem, the input consists of
\begin{itemize}
	\item an edge-capacitated graph $G=(V,E,c)$, with length $\set{\ell_e}_{e\in E}$ on its edges; and
	\item a set $T\subseteq V$ of $k$ terminals.
\end{itemize}
A solution consists of
\begin{itemize}
	\item a partition $\fset$ of $V$, such that distinct terminals of $T$ belong to different sets in $\fset$; for each vertex $u\in V$, we denote by $F(u)$ the cluster in $\fset$ that contains it;
	\item a semi-metric $\delta$ on the clusters in $\fset$, such that for all pairs $t,t'\in T$, $\delta(F(t),F(t'))= \dist_{\ell}(t,t')$, where $\dist_{\ell}(\cdot,\cdot)$ is the shortest-path distance (in $G$) metric induced by edge length $\set{\ell_e}_{e\in E(G)}$.
\end{itemize}
We define the \emph{cost} of a solution $(\fset,\delta)$ as $\vol(\fset,\delta)=\sum_{(u,v)\in E}c(u,v)\cdot\delta(F(u),F(v))$, and its \emph{size} as $|\fset|$.
The goal is to compute a solution $(\fset,\delta)$ with small size and cost. 
Typically, the following ratio is a central measure to be minimized, called \emph{average stretch}:
\[\rho=\frac{\sum_{(u,v)\in E}c(u,v)\cdot\delta(F(u),F(v))}{\sum_{(u,v)\in E}c(u,v)\cdot\ell_{(u,v)}}.\]

In $\zesn$, instead of forcing all vertices to be mapped to terminals, we allow them to be mapped to non-terminals (or Steiner nodes), which are clusters in $\fset$ that do not contain terminals. 
We are also allowed to manipulate the distances between these non-terminals, conditioned on not destroying the shortest-path distance metric (in $G$) on terminals. We remark that a similar variant was proposed in \cite{andoni2014towards}, and we provide a detailed comparison and discussion between them in \Cref{sec: reduction} and \Cref{apd: comparison}.

Similar to  \cite{moitra2009approximation,leighton2010extensions}, using the framework of \cite{andoni2014towards}, we can show the following connection between $\zesn$ and the best quality achievable by contraction-based flow sparsifiers (with Steiner nodes):
Let $G$ be a graph, $\eps>0$ be any real number and $T$ be a set of $k$ terminals. If for every set of edge lengths $\set{\ell_e}$, the corresponding $\zesn$ instance admits a solution of size $f(k)$ and average stretch $\rho$, then $G$ has a contraction-based flow sparsifier with size $g(k,\eps)$ and quality $(1+\eps)\cdot \rho$, where $f$ and $g$ are functions that do not depend on the size of $G$. (See \Cref{sec: reduction}.)

\begin{figure}[h]
	\centering
	\scalebox{0.1}{\includegraphics{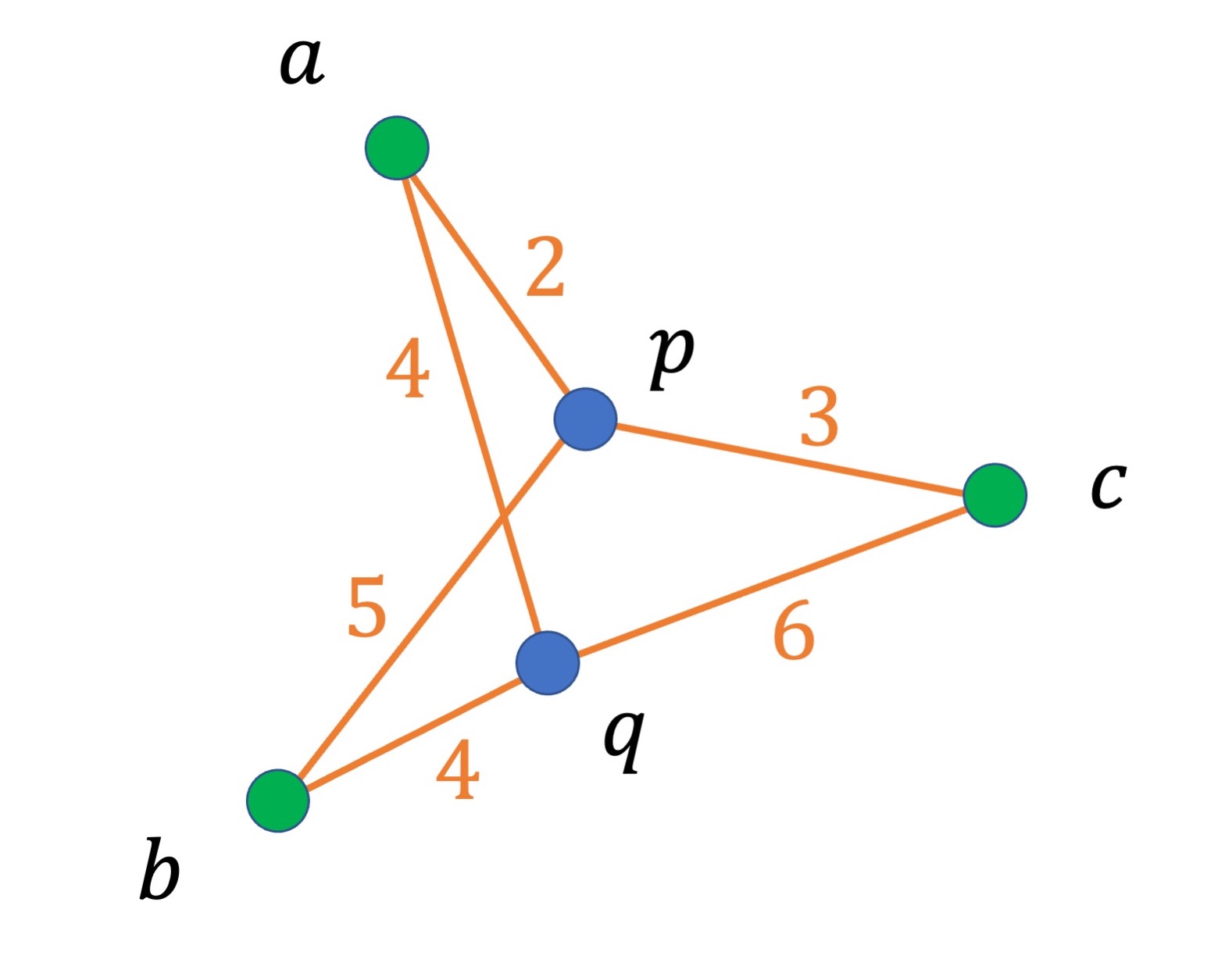}}
	\caption{A $\zesn$ instance with terminals $a,b,c$. Every edge has capacity $1$ and length marked in yellow.\label{fig: projection}}
\end{figure}

We then turn to study the $\zesn$ problem. For some intuition, let us consider the simple instance in \Cref{fig: projection}. Recall that the goal is to manipulate the edge length such that the sum of all edge length is minimized (as in this case all capacities are the same) while the distance between terminals are preserved.
Observe that in the current graph, all terminal shortest paths are supported by vertex $p$ (that is, $a$-$b$ shortest path is $(a,p,b)$, $a$-$c$ shortest path is $(a,p,c)$, and $b$-$c$ shortest path is $(b,p,c)$). So we should not modify the lengths of edges $(a,p),(b,p),(c,p)$, as shortening one by $\delta$ would force us to lengthening the other two by $\delta$ each, causing the total sum to increase. However, the lengths of $q$-incident edges can be manipulated without distorting the terminal distances. For example, we can shorten $(b,q)$ from $4$ to $3$ and $(c,q)$ from $6$ to $5$, so $(a,q,b)$ becomes another $a$-$b$ shortest path and $(b,q,c)$ becomes another $b$-$c$ shortest path (and then we cannot shorten any single edge without increasing others). At this moment, $q$ can in fact be identified with the point on edge $(p,b)$ that is at distance $3$ from $b$ and $2$ from $p$.

This simple example actually illustrates the first step (out of two) of our algorithm for the $\zesn$ problem, which we call the \emph{projection} step. The main idea is to repeatedly reduce the distances from non-terminals to terminals in some simple way until we cannot do so anymore without destroying the terminal-induced metric. Clearly, at the end of this process, for every non-terminal $q$, and every terminal $t$, there should be another terminal $t'$ such that $\dist(q,t)+\dist(q,t')=\dist(t,t')$ holds, and it is this tight constraint that prevents  us from further shortening the $q$-$t$ distance. In fact, in the area of metric geometry, there is indeed a notion called \emph{tight span} (first proposed and studied in \cite{dress1984trees}), characterizing the resulting distances that we may get from the projection step.

\paragraph{Tight Span.}
Let $D$ be a metric on a set $T$ of points. The \emph{tight span} of $D$, denoted as $\TS(D)$, is defined as
\[
\TS(D)=\bigg\{x=(x_t)_{t\in T} \in (\mathbb{R}^+)^{T} \text{ }\bigg|\text{ } x_t=\max_{t'\in T}\set{D(t,t')-x_{t'}}\bigg\}.
\]
So elements in $\TS(D)$ are $|T|$-dimensional vectors with coordinates indexed by points in $T$.
Intuitively, we can think of an element $x\in \TS(D)$ as an ``imaginary point'' in the metric space $(T,D)$, that is at distance $x_t$ to each $t\in T$. 
The distances $\set{x_t}_{t\in T}$ need to satisfy triangle inequalities with the distances in $\set{D(t,t')}_{t,t'}$. That is,
for all pairs $t,t'$ of points in $T$, $x_t+x_{t'}\ge D(t,t')$ must hold. Equivalently, for all $t\in T$, $x_t\ge \max_{t'\in T}\set{D(t,t')-x_{t'}}$.
Moreover, for $x$ to be in the tight span, it is additionally required that, \emph{for each point $t\in T$, at least one inequality in $\set{x_t+x_{t'}\ge D(t,t')}_{t'}$ is tight} (and therefore the name), giving that 
$x_t=\max_{t'\in T}\set{D(t,t')-x_{t'}}$.

We then show that, we can efficiently project all non-terminals into the tight span of the terminal metric, such that the length of every edge does not increase (where the tight span is equipped with the geodesic distance, or equivalently the $\ell_{\infty}$ norm). This means that we only need to solve $\zesn$ instance that are ``embedded into the tight span''. (See \Cref{sec: tight span}.)

Now constructing a solution to $\zesn$ (with a bounded size) is essentially partitioning the tight span into a finite number of components (and then each component is contracted as a Steiner node) that does not cut too many edges. A desirable property of a metric space for such a good partitioning to exist is the following property called \emph{separability}. Specifically, we say that a metric space $(X,D)$ (where $|X|=+\infty$) is \emph{separable}, iff there exists a finite subset $\bar X\subseteq X$ and a (randomized) mapping $f: X\to \bar X$, such that
for every pair $x,x'\in X$, 
$$\mathbb{E}[D(f(x),f(x'))]=D(x,x').$$
For example, a line metric (where $X=[a,b]$ and $D(x,x')=|x-x'|$) is separable, as we can simply define $\bar X=\set{a,b}$ and choose a random threshold value $c\in [a,b]$ and set $f(x)=a$ iff $x\le c$ and $f(x)=b$ iff $x> c$.
Similarly, a rectangle $\overline{abcd}$ with an axis-aligned $\ell_1$ metric is also separable (with $X=\set{a,b,c,d}$ and $f$ determined by two random threshold values).

The crucial distinction between our results for $5$-terminal and $6$-terminal networks is due to the following fact that we prove in \Cref{sec: upper} and \Cref{sec: lower}:
\emph{The tight span of all $5$-point metrics are separable, but there are $6$-point metrics whose tight spans are not separable.}

To see why the tight spans of all $5$-point metrics are separable, a vague explanation is that $5$-point tight spans are ``at most $2$-dimensional'', and $2$-dimensional rectangles with $\ell_{\infty}$ norm, after properly changing the coordinate systems, become $2$-dimensional $\ell_1$ spaces, and are therefore separable from the above discussion. (See \Cref{sec: upper}.)
The intuition for $5$-point tight spans to be $2$-dimensional is as follows. Recall that the elements in $5$-point tight spans are $5$-dimensional vectors, and for such a vector $x$ to be in the tight span, it is additionally required that, for each coordinate $t\in T$, at least one triangle inequality in $\set{x_t+x_{t'}\ge D(t,t')}_{t'}$ is tight. Therefore, in a proper region of the tight span, at least three triangle inequalities need to be tight, reducing the ``degree of freedom'' of the $5$-dimensional vector to be $2$ (and therefore making the region essentially $2$-dimensional). Vaguely, if there are only two tight inequalities, say $T=\set{a,b,c,d,e}$ $x_a+x_b=D(a,b)$, $x_c+x_d=D(c,d)$, then nothing prevents us from further reducing $x_e$, the distance between $x$ and $e$.

On the other hand, the tight span for $6$-point metrics can be $3$-dimensional, as in a proper region, all $6$-dimensional vectors may satisfy the tight constraints $x_a+x_b=D(a,b)$, $x_c+x_d=D(c,d)$, and $x_e+x_f=D(e,f)$, and so no coordinate can be further reduced. It appears that the geodesic metric in $3$-dimensonal tight span regions differs fundamentally from $2$-dimensonal regions (which are $\ell_1$ metrics), and are therefore not separable. We manage to leverage this non-separability to construct hard instances for $\zesn$ and eventually generalize them to provide a similar lower bound for the version in \cite{andoni2014towards}, which in turn leads to lower bounds for contraction-based flow sparsifiers. The construction and the analysis of the hard instances are the most technical components of the paper.
(See \Cref{sec: lower}.)

\subsection{Related Work}
\label{sec: related}

\paragraph{Cut sparsifiers.}
Cut sparsifiers are closely related to flow sparsifiers.
Given a graph $G$ and a set $T\subseteq V(G)$ of terminals, a cut sparsifier of $G$ with respect to $T$ is a graph $G'$ with $T\subseteq V(G')$, such that for every partition $(T_1,T_2)$ of $T$, the size of the minimum cut separating $T_1$ from $T_2$ in $G$ and the size of the minimum cut separating $T_1$ from $T_2$ in $G'$, are within some small multiplicative factor $q$, which is also called the \emph{quality} of the sparsifier.
Quality-$q$ flow sparsifiers are also quality-$q$ cut sparsifiers, but the converse is not true.

In a special case where $V(G')=T$, Moitra \cite{moitra2009approximation} showed that every graph with $k$ terminals admits a cut sparsifier with quality $O(\log k/\log\log k)$, and the strongest lower bound is $\Omega(\sqrt{\log k}/\log\log k)$ \cite{makarychev2010metric,charikar2010vertex}. 
In another special case where $q=1$, it was shown that every $k$-terminal graph admits an exact cut sparsifier of size at most $2^{2^k}$ \cite{hagerup1998characterizing,khan2014mimicking}, and the strongest lower bound is $2^{\Omega(k)}$ \cite{khan2014mimicking,krauthgamer2013mimicking,karpov2017exponential}. Closing this gap remains a very interesting open problem.
If we further assume that each terminal has degree $1$, then Chuzhoy \cite{chuzhoy2012vertex} has shown the construction of $O(1)$-quality cut sparsifier of size $O(k^3)$, and Kratsch and Wahlstrom \cite{kratsch2012representative} have constructed quality-$1$ cut sparsifiers of size $O(k^3)$ via a matroid-based approach.

There are also some other recent work on (i) constructing better cut sparsifiers for special types of graphs, for example trees \cite{goranci2017vertex}, planar graphs \cite{krauthgamer2017refined,karpov2017exponential,goranci2017improved}; (ii) preserving terminal min-cut values up to some threshold \cite{chalermsook2021vertex,liu2020vertex}; and
(iii) dynamic cut/flow sparsifiers and their use in dynamic graph algorithms \cite{durfee2019fully,chen2020fast,goranci2021expander}.

\paragraph{Distance sparsifiers.} There are also rich lines of work for constructing vertex/edge sparsifiers for preserving distances (e.g. spanners, emulators, distance-preserving minor, distance oracles, etc). We refer the readers to the comprehensive survey \cite{ahmed2020graph}.

\subsection{Organization}
The rest of the paper is organized as follows. We start with some preliminaries and formal definitions in \Cref{sec: prelim}.
We describe a high-level framework in \Cref{sec: reduction} for proving our results \Cref{main: upper} and \Cref{main: lower}, reducing them to proving upper and lower bounds of the $\zesn$ problem.
Then in \Cref{sec: tight span}, we recall the notion tight span, which is crucial for our algorithm and lower bound result on the $\zesn$ problem, presented in \Cref{sec: upper} and \Cref{sec: lower}, respectively.

\section{Preliminaries}
\label{sec: prelim}

By default, all logarithms are to the base of $2$. For a real number $x$, we denote $(x)^+=\max\set{x,0}$.

Let $G=(V,E,\ell)$ be an edge-weighted graph, where each edge $e\in E$ has weight (or \emph{length}) $\ell_e$. 
For a vertex $v\in V$, we denote by $\deg_G(v)$ the degree of $v$ in $G$.
For a pair $v,v'$ of vertices in $G$, we denote by $\dist_{G}(v,v')$ (or $\dist_{\ell}(v,v')$) the shortest-path distance between $v$ and $v'$ in $G$.
We may omit the subscript $G$ in the above notations when the graph is clear from the context.

\paragraph{Demands, congestion, and quality of a sparsifier.}
Let $G$ be a graph and let $T$ be a subset of its vertices called \emph{terminals}. A demand $\dset$ on $T$ is a function that assigns to each (unordered) pair $t,t'\in T$ a real number $\dset(t,t')\ge 0$. Let $F$ be a multi-commodity flow on $G$. We say that $F$ \emph{routes} $\dset$, iff for every pair $t,t'\in T$, $F$ sends $\dset(t,t')$ units of flow from $t$ to $t'$ (or from $t'$ to $t$) in $G$. The \emph{congestion} of flow $F$ in $G$, denoted by $\cong_G(F)$, is defined to be the maximum amount of flow sent via a single edge in $G$.
The \emph{volume} of flow $F$, denoted by $\mathsf{vol}(F)$, is defined to be the sum, over all edges in $G$, the total amount of flow sent via it, so clearly $\cong_G(F)\ge \mathsf{vol}(F)/|E(G)|$.
The \emph{congestion} of $\dset$ in $G$, denoted by $\cong_G(\dset)$, is defined to be the minimum congestion of any flow that routes $\dset$ in $G$.

Let $H$ be a graph with $T\subseteq V(H)$. We say that $H$ is a \emph{flow sparsifier} of $G$ with respect to $T$ with \emph{quality} $q\ge 1$, iff for any demand $\dset$ on $T$, 
\[\cong_H(\dset)\le \cong_G(\dset)\le q\cdot\cong_H(\dset).\]

A graph $H$ is a \emph{contraction-based flow sparsifier of $G$ with respect to $T$}, iff there exists a partition $\fset$ of vertices in $G$ into subsets where different terminals in $T$ lie in different sets in $\fset$, and $H$ is obtained from $G$ by contracting vertices in each set $F$ in $\fset$ into a single node $u_F$, keeping parallel edges and discard self-loops. For each $t\in T$, if $F(t)$ is the (unique) cluster in $\fset$ that contains  $t$, then the node $u_{F(t)}$ in $H$ is identified with $t$.

\section{Reduction to Variants of $0$-Extension}
\label{sec: reduction}

Recall that we are given a graph $G$ and a set $T$ of its vertices called terminals. Let $\dset$ be a demand on the $T$.
Let $H$ be a flow sparsifier of $G$ with respect to $T$.
We use the following LP for computing $\lambda_H(\dset):=(\cong_H(\dset))^{-1}$. For every pair $t,t'\in T$, we denote by $\pset_{t,t'}$ the collection of all $t$-$t'$ paths in $H$ connecting $t$ to $t'$ in $H$, and $\pset=\bigcup_{t,t'\in T}\pset_{t,t'}$.
\begin{eqnarray*}
	\mbox{(LP-Primal)}\quad	 \lambda_H(\dset)= \quad & \text{maximize} &\lambda \\
	&\sum_{P\in \pset_{t,t'}}x_P\geq \lambda\cdot \dset_{t,t'} &\forall t,t'\in T\\
	&\sum_{P\in \pset: e\in P}x_P\leq c(e) &\forall e\in E(H)\\
	&x_P\geq 0&\forall P\in \pset
\end{eqnarray*}
Taking its dual, we obtain the following LP.
\begin{eqnarray*}
	\mbox{(LP-Dual)}\quad	 \quad & \text{minimize} & \sum_{e\in E(H)}c(e)\cdot \ell_e \\
	&\sum_{t,t'\in T}\delta_{t,t'}\cdot \dset_{t,t'}\geq 1 &\\ 
	&\sum_{e\in P}\ell_e\geq \delta_{t,t'} &\forall t,t'\in T, \forall P\in \pset_{t,t'}\\
	& \delta_{t,t'} \ge 0 &\forall t,t'\in T\\
	&\ell_e\geq 0&\forall e\in E(H)
\end{eqnarray*}

We formulate the following two graph-theoretic problems out of (LP-Dual) (one of them has been implicitly defined in \cite{andoni2014towards}), that can be viewed variants of the classic $0$-Extension problem, and then we reduce \Cref{main: upper} and \ref{main: lower} to proving upper and lower bounds of these variants.
We provide a comparison between these variants and the classic $0$-Extension problem in \Cref{apd: comparison}.

\subsection{Two Variants of $0$-Extension}
\label{sec: variants}

\paragraph{$0$-Extension with Steiner Nodes.}
In an instance of the \emph{$0$-Extension with Steiner Nodes} problem ($\zesn$), the input consists of
\begin{itemize}
	\item an edge-capacitated graph $G=(V,E,c)$, with length $\set{\ell_e}_{e\in E}$ on its edges; and
	\item a set $T\subseteq V$ of $k$ terminals.
\end{itemize}
A solution consists of
\begin{itemize}
	\item a partition $\fset$ of $V$, such that distinct terminals of $T$ belong to different sets in $\fset$; for each vertex $u\in V$, we denote by $F(u)$ the cluster in $\fset$ that contains it;
	\item a semi-metric $\delta$ on the clusters in $\fset$, such that \underline{for all pairs $t,t'\in T$, $\delta(F(t),F(t'))= \dist_{\ell}(t,t')$}, where $\dist_{\ell}(\cdot,\cdot)$ is the shortest-path distance (in $G$) metric induced by edge length $\set{\ell_e}_{e\in E(G)}$.
\end{itemize}

We define the \emph{cost} of a solution $(\fset,\delta)$ as $\vol(\fset,\delta)=\sum_{(u,v)\in E}c(u,v)\cdot\delta(F(u),F(v))$, and define its \emph{size} to be $|\fset|$.
The goal is to compute a solution $(\fset,\delta)$ with small size and cost.

\paragraph{The ``average'' version of $\zesn$.}
In \cite{andoni2014towards}, the following variant of $0$-Extension (referred to as $\zesn_{\textsf{ave}}$) was proposed.
Its input consists of
\begin{itemize}
	\item an edge-capacitated graph $G=(V,E,c)$,  with length $\set{\ell_e}_{e\in E}$ on its edges;
	\item a set $T\subseteq V$ of $k$ terminals; and
	\item \underline{a \emph{demand} $\dset: T\times T\to \mathbb{R}^+$ on terminals}.
\end{itemize}
A solution consists of 
\begin{itemize}
	\item a partition $\fset$ of $V$, such that distinct terminals of $T$ belong to different sets in $\fset$; for each vertex $u\in V$, we denote by $F(u)$ the cluster in $\fset$ that contains it;
	\item a semi-metric $\delta$ on the clusters in $\fset$, such that:\\ \underline{$\sum_{t,t'}\dset(t,t')\cdot\delta(F(t),F(t'))\ge \sum_{t,t'} \dset(t,t')\cdot\dist_{\ell}(t,t')$},\\
	where $\dist_{\ell}(\cdot,\cdot)$ is the shortest-path distance (in $G$) metric induced by edge length $\set{\ell_e}_{e\in E(G)}$.
\end{itemize}
The cost of a solution $(\fset,\delta)$ is $\vol(\fset,\delta)=\sum_{(u,v)\in E}c(u,v)\cdot\delta(F(u),F(v))$, and its \emph{size} is $|\fset|$.

The difference between two variants are underlined. In $\zesn_{\textsf{ave}}$, it is required that some ``average'' terminal distance does not decrease, while in $\zesn$ it is required that all pairwise distances between terminals are preserved.
Clearly, the requirement in $\zesn$ is stronger, which implies that a valid solution to $\zesn$ is also a valid solution to the same $\zesn_{\textsf{ave}}$ instance (with an arbitrary $\dset$).

\subsection{Proof of \Cref{main: upper}}

In \cite{andoni2014towards}, the following theorem was proved ((LP1) and Proposition 4.2).
\begin{theorem}
\label{thm: AGK}
Given a graph $G$ with a set $T$ of $k$ terminals, if for every length $\set{\ell_e}_{e\in E(G)}$ and every demand $\dset$, the instance $(G,T,\ell,\dset)$ of $\zesn_{\textnormal{\textsf{ave}}}$ has a solution $(\fset,\delta)$ with size $|\fset|\le f(k)$ and cost 
\[\sum_{(u,v)\in E}c(u,v)\cdot\delta(F(u),F(v))\le q\cdot \sum_{(u,v)\in E}c(u,v)\cdot\dist_{\ell}(u,v),\] 
then there is a quality-$(1+\eps) q$ flow sparsifier $H$ for $G$ w.r.t $T$ with $|V(H)|\le (f(k))^{(O(\log k/\eps))^{k^2}}$.
\end{theorem}

Our first main result is the following theorem, whose proof is deferred to \Cref{sec: upper}.

\begin{theorem}
	\label{thm: upper 5}
	For each instance $(G,T,\ell)$ of $\zesn$ with $|T|\le 5$, there exists a solution $(\fset,\delta)$, such that $|\fset|\le 30$, and $\vol(\fset,\delta)=\sum_{(u,v)\in E}c(u,v)\cdot\dist_{\ell}(u,v)$.
\end{theorem}

We now use \Cref{thm: upper 5} to complete the proof of  \Cref{main: upper}.
Let $G$ be any graph with a set $T$ of $|T|\le 5$ terminals.
Take any instance $(G,T,\ell,\dset)$ of $\zesn_{\textnormal{\textsf{ave}}}$ and consider the instance $(G,T,\ell)$ of $\zesn$. From \Cref{thm: upper 5}, there is a solution $(\fset,\delta)$ to instance $(G,T,\ell)$, such that $|\fset|\le 30$, and $\vol(\fset,\delta)\le 1\cdot \sum_{(u,v)\in E}c(u,v)\cdot\dist_{\ell}(u,v)$. From the above discussions, we know that $(\fset,\delta)$ is also a solution to instance $(G,T,\ell,\dset)$.
Therefore, from \Cref{thm: AGK}, there is a quality-$(1+\eps)$ flow sparsifier for $G$ with respect to $T$ with at most $30^{O(1/\eps)^{25}}=2^{\poly(1/\eps)}$ vertices.

\subsection{Proof of \Cref{main: lower}}

Our second main result is the following theorem, whose proof is deferred to \Cref{sec: lower}.

\begin{theorem}
\label{thm: lower 6}
For any $N$, there exists an instance $(G,T,\ell,\dset)$ of $\zesn_{\textnormal{\textsf{ave}}}$ with $|T|= 6$, such that any solution $(\fset,\delta)$ with $|\fset|\le N$ satisfies that $\vol(\fset,\delta)\ge (1+10^{-18})\cdot \sum_{(u,v)\in E}c(u,v)\cdot\dist_{\ell}(u,v)$.
\end{theorem}

We now use \Cref{thm: lower 6} to complete the proof of \Cref{main: lower}.
For the given $N$, let $(G,T,\ell,\dset)$ be the instance of $\zesn_{\textnormal{\textsf{ave}}}$ given by \Cref{thm: lower 6}.
Denote $L=\sum_{t,t'}\dset(t,t')\cdot \dist_{\ell}(t,t')$. Let $\dset'=\dset/L$.
Let $H$ be any flow sparsifier of $G$ with $|V(H)|\le N$. 
We will show $\cong_G(\dset')\ge (1+10^{-18})\cdot \cong_H(\dset')$, which by definition of flow sparsifiers implies \Cref{main: lower}.

Consider (LP-Dual) for computing $(\cong_G(\dset'))^{-1}$ and $(\cong_H(\dset'))^{-1}$.
\begin{itemize}
\item For $G$, if we set $\delta_{t,t'}=\dist_{\ell}(t,t')$, then it is easy to verify that $(\ell,\delta)$ is a feasible solution to (LP-Dual($G$)), and therefore $(\cong_G(\dset'))^{-1}\le \sum_{(u,v)\in E}c(u,v)\cdot\dist_{\ell}(u,v)$.
\item For $H$, take any solution $(\ell',\delta')$ to (LP-Dual($H$)), assuming without loss of generality $\delta'_{t,t'}=\dist_{\ell}(t,t')$, let $\fset$ be the partition of $V(G)$ that forms $H$ from $G$ (so $|\fset|\le N$), and consider the solution $(\fset,\delta')$. It is easy to verify that $(\fset,\delta')$ is a feasible solution to $(G,T,\ell,\dset)$, as
\[\sum_{t,t'}\dset'(t,t')\cdot\delta'_{t,t'}\ge 1=\sum_{t,t'} \dset'(t,t')\cdot\dist_{\ell}(t,t').\]
From \Cref{thm: lower 6}, $\sum_{(u,v)\in E}c(u,v)\cdot\delta'(F(u),F(v))\ge (1+10^{-18})\cdot\sum_{(u,v)\in E}c(u,v)\cdot\dist_{\ell}(u,v)$. This means that the value of  (LP-Dual($H$)) on $(\ell',\delta')$ is at least $(1+10^{-18})\cdot (\cong_G(\dset'))^{-1}$. As $(\ell',\delta')$ is arbitrary, this implies that $\cong_G(\dset')\ge (1+10^{-18})\cdot \cong_H(\dset')$.
\end{itemize} 



\newcommand{\Poly}{\mathcal{P}}
\newcommand{\proj}{\mathsf{proj}}
\def\norm#1{\left\| #1 \right\|}

\section{Tight Span and its Projection}
\label{sec: tight span}

A key ingredient to our algorithm and lower bound  is the theory of \emph{tight span} in metric geometry. The notion of tight span was first proposed and studied in \cite{dress1984trees}. As a central notion in T-theory \cite{dress1996t}, it has proved powerful in the study of phylogenetic analysis and optimal reconstruction of metrics
\cite{dress2001hereditarily,dress2006hereditarily,huber2008characterizing,koolen2009optimal,dress2012basic,huber2021optimal}.


Let $D$ be a metric on a set $T$ of points. The \emph{tight span} of $D$, denoted as $\TS(D)$, is defined as
\[
\TS(D)=\bigg\{x=(x_t)_{t\in T} \in (\mathbb{R}^+)^{T} \text{ }\bigg|\text{ } x_t=\max_{t'\in T}\set{D(t,t')-x_{t'}}\bigg\}.
\]
So elements in $\TS(D)$ are $|T|$-dimensional vectors with coordinates indexed by points in $T$.
Intuitively, we can think of an element $x\in \TS(D)$ as an ``imaginary point'' in the metric space $(T,D)$, that is at distance $x_t$ to each $t\in T$. 
The distances $\set{x_t}_{t\in T}$ need to satisfy triangle inequalities with the distances in $\set{D(t,t')}_{t,t'}$. That is,
for all pairs $t,t'$ of points in $T$, $x_t+x_{t'}\ge D(t,t')$ must hold. Equivalently, for all $t\in T$, $x_t\ge \max_{t'\in T}\set{D(t,t')-x_{t'}}$.
Moreover, for $x$ to be in the tight span, it is additionally required that, \emph{for each point $t\in T$, at least one inequality in $\set{x_t+x_{t'}\ge D(t,t')}_{t'}$ is tight} (and therefore the name), giving that 
$x_t=\max_{t'\in T}\set{D(t,t')-x_{t'}}$.
We say that $t$ and $t'$, as coordinates of the vector $x$, are \emph{involved} in the triangle inequality $x_t+x_{t'}\ge D(t,t')$. So in other words, $x\in \TS(D)$ means every coordinate of $x$ is involved in some tight inequality.

\subsection{Examples}

For a metric $D$ on two points $\set{a,b}$, it is easy to see that the tight span is the one-dimensional set $\TS(D)=\set{(r,D(a,b)-r)\mid 0\le r\le D(a,b)}$. We now give examples for metrics on three or four points.

\paragraph{Example 1.} A metric $D$ on three points $\set{a,b,c}$ and its tight span $\TS(D)$ is shown in \Cref{fig: TS_3}. Specifically, $\TS(D)$ is the union of three $1$-dimensional sets, $\overline{oa}$, $\overline{ob}$, and $\overline{oc}$, where
\begin{itemize}
	\item 	 $\overline{oa}=\set{(r,7-r,5-r)\mid 0\le r\le 2}$; and the tight constraints are $x_a+x_b=7$ and $x_a+x_c=5$;
	\item $\overline{ob}=\set{(7-r,r,8-r)\mid 0\le r\le 5}$; and
	the tight constraints are $x_a+x_b=7$ and $x_b+x_c=8$;
	\item $\overline{oc}=\set{(5-r,8-r,r)\mid 0\le r\le 3}$; and
	the tight constraints are $x_a+x_c=5$ and $x_b+x_c=8$.
\end{itemize}

\begin{figure}[h]
	\centering
	\subfigure
	{\scalebox{0.1}{\includegraphics{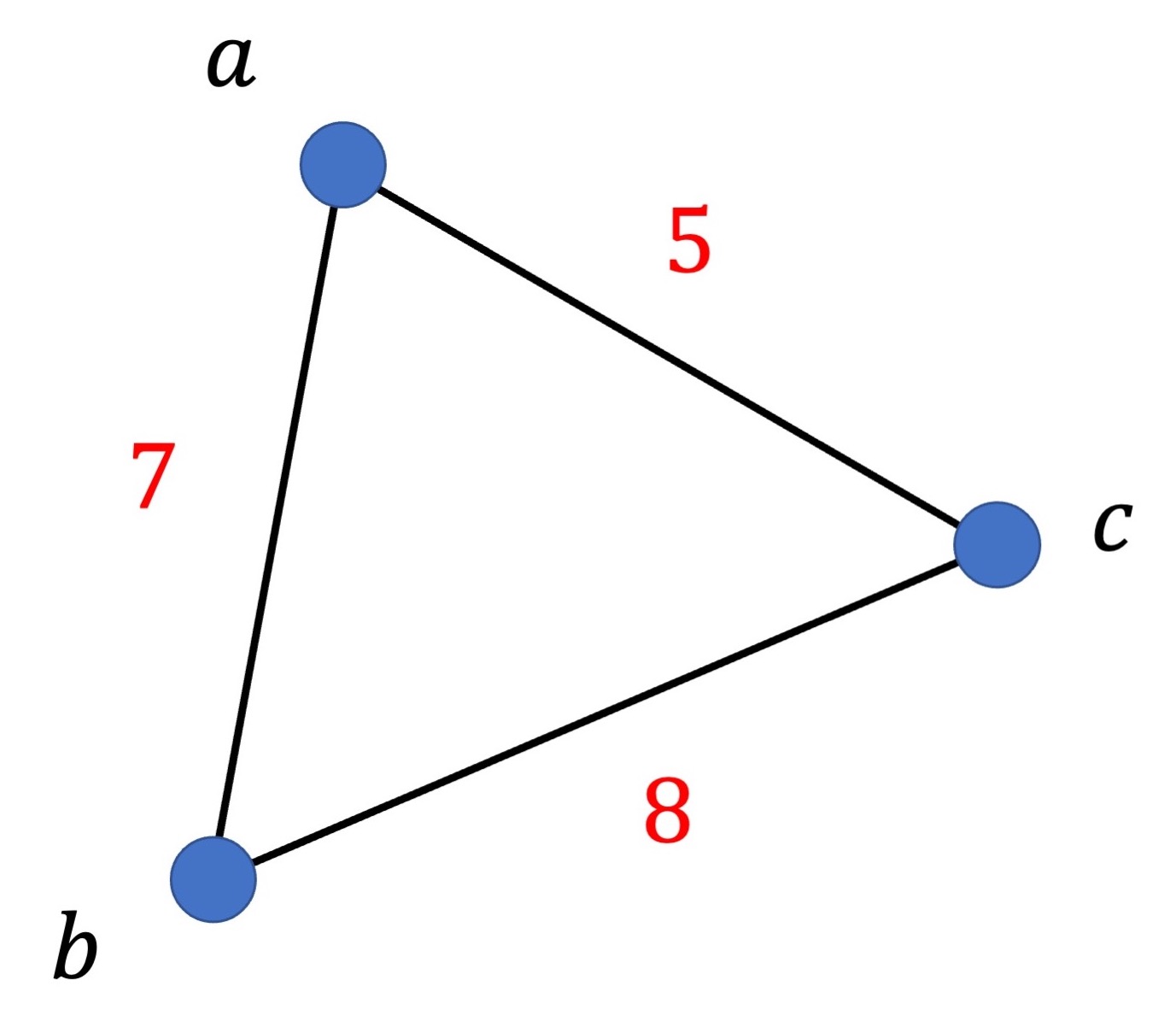}}\label{fig: case11}}
	\hspace{0.7cm}
	\subfigure
	{
		\scalebox{0.1}{\includegraphics{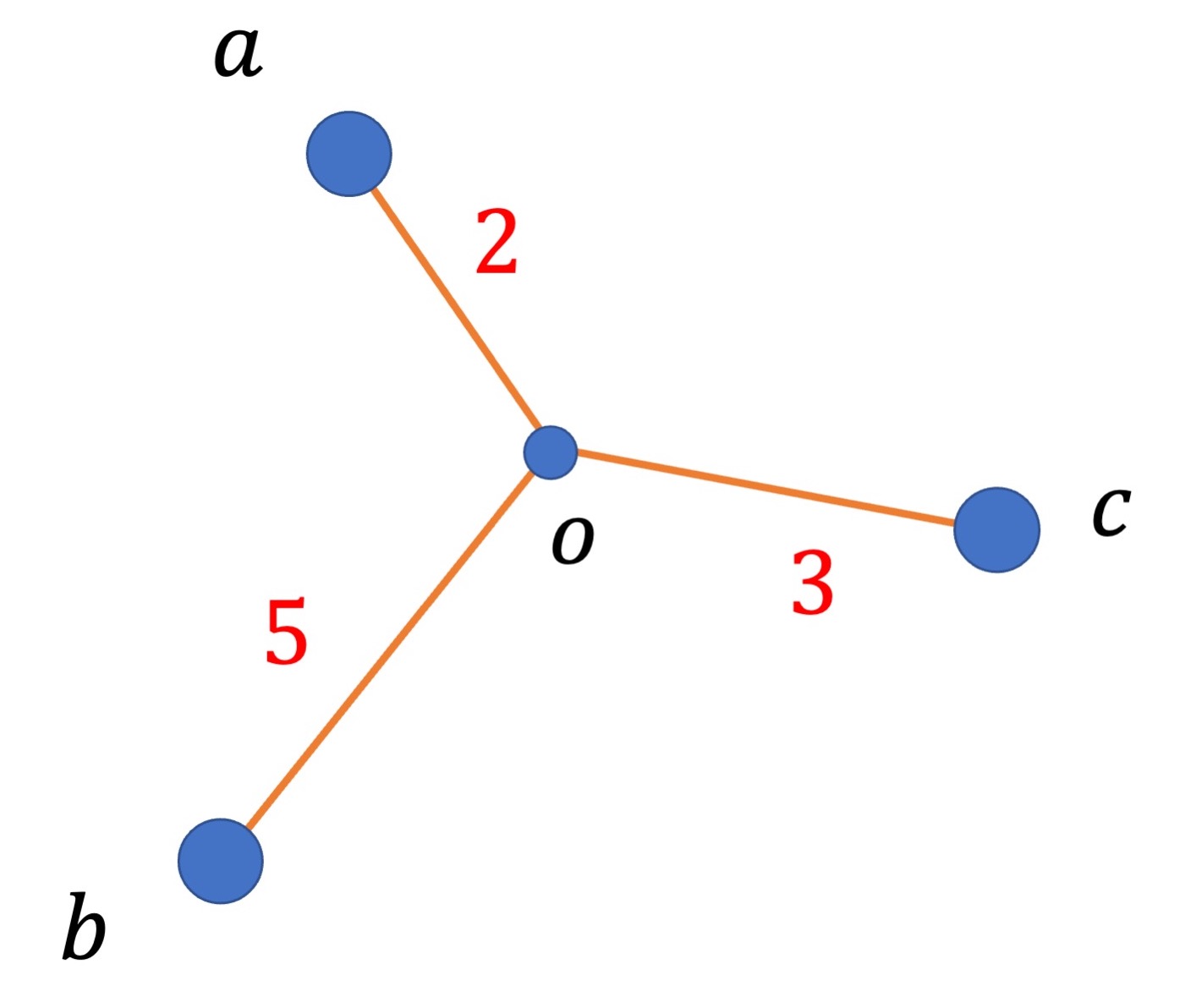}}\label{fig: case12}}
	\caption{An illustration of a metric $D$ on three points (left) and its tight span $\TS(D)$ (right).\label{fig: TS_3}}
\end{figure}

\paragraph{Example 2.} A metric on four points and its tight span is shown in \Cref{fig: TS_4}. Specifically, $\TS(D)$ is the union of four $1$-dimensional sets, $\overline{aa'},\overline{bb'},\overline{cc'},\overline{dd'}$, where
\begin{itemize}
	\item 	 $\overline{aa'}=\set{(r,7-r,8-r,4-r)\mid 0\le r\le 1.5}$; with $x_a+x_b=7$, $x_a+x_c=8$, and $x_a+x_d=4$;
	\item 	 $\overline{bb'}=\set{(7-r,r,6-r,8-r)\mid 0\le r\le 2.5}$; with $x_a+x_b=7$, $x_b+x_c=6$, and $x_b+x_d=8$;
	\item 	 $\overline{cc'}=\set{(8-r,6-r,r,5-r)\mid 0\le r\le 1.5}$; with $x_a+x_c=8$, $x_b+x_c=6$, and $x_c+x_d=5$;
	\item 	 $\overline{dd'}=\set{(4-r,8-r,5-r,r)\mid 0\le r\le 0.5}$; with $x_a+x_d=4$, $x_b+x_d=8$, and $x_c+x_d=5$;
\end{itemize}
and a $2$-dimensional set $\overline{a'b'c'd'}$, which is defined to be
\[\bigg\{\big(1.5+t+s,2.5+(3-t)+s, 1.5+(3-t)+(2-s), 0.5+t+(2-s)\big)\text{ }\bigg|\text{ } 0\le t\le 3, 0\le s\le 2\bigg\},\]
and the tight constraints are $x_a+x_c=8$ and $x_b+x_d=8$.

\begin{figure}[h]
	\centering
	\subfigure
	{\scalebox{0.1}{\includegraphics{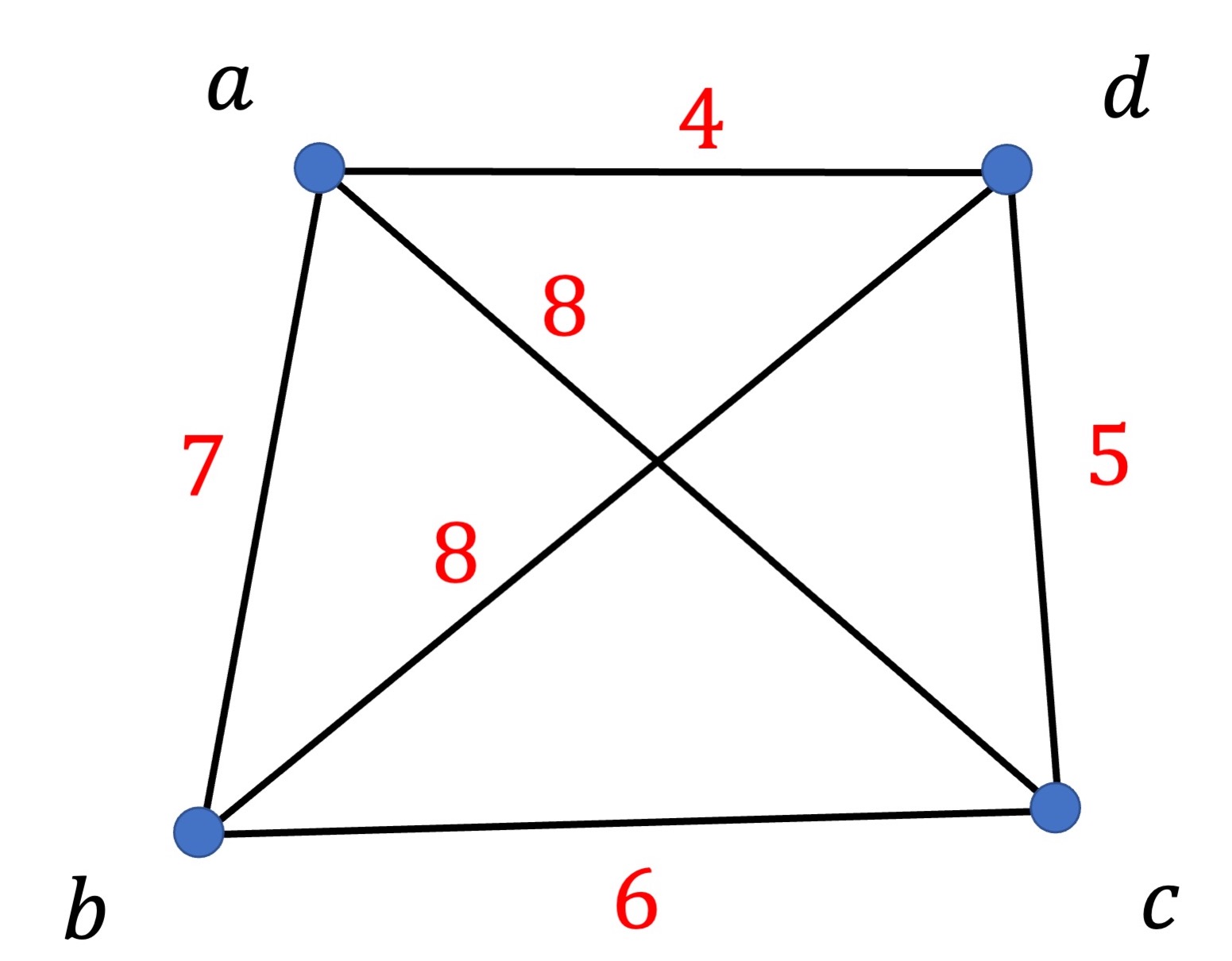}}}
	\hspace{0.7cm}
	\subfigure
	{
		\scalebox{0.12}{\includegraphics{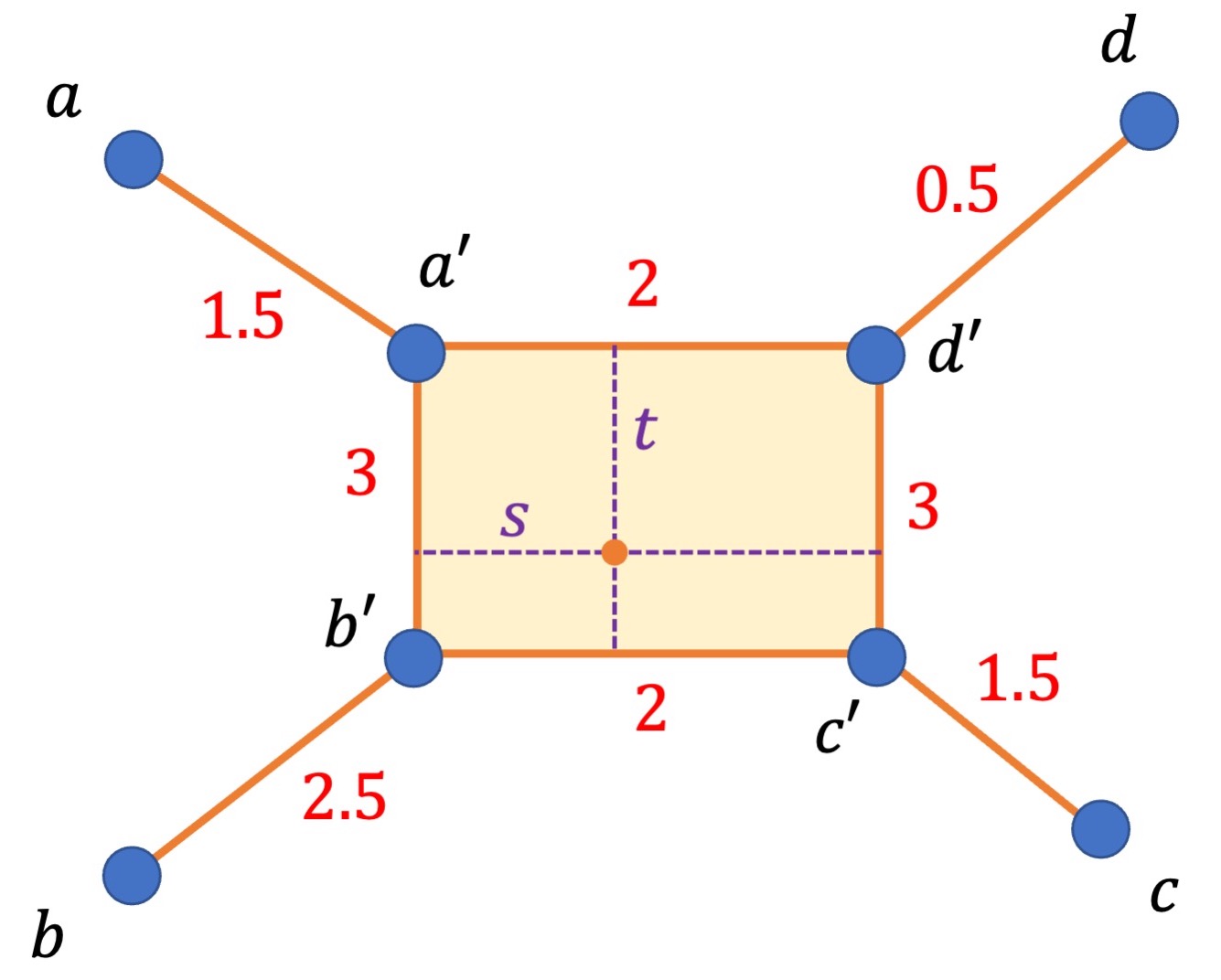}}}
	\caption{An illustration of a metric $D$ on four points (left) and its tight span $\TS(D)$ (right).\label{fig: TS_4}}
\end{figure}

\subsection{Projection onto the tight span}
\label{subsec: proj}

We say that a vector $x\in \mathbb{R}^T$ is \emph{valid} iff it satisfies all triangle inequalities in $\set{x_t+x_{t'}\ge D(t,t')}_{t,t'}$. 
We now describe a mapping that maps all valid vectors to vectors in $\TS(D)$, that is non-expanding under $\ell_{\infty}$ norm.
The existence of such a mapping was proved in Section 5.2 of \cite{dress2012basic} (there the mapping is called a ``contraction''). Here we describe an algorithm that computes the mapping.

For a valid vector that does not lie in $\TS(D)$, we describe a process that \emph{projects $x\in \mathbb{R}^T$ onto $\TS(D)$}, at a vector which we denote by $\proj(x)$. Intuitively, $\proj(x)$ is obtained from $x$ by decreasing the value of its coordinates at the same rate until some triangle inequalities become tight, freezing the coordinates that are involved in such inequalities and continuing on the remaining coordinates. 

We now describe the process in detail. Throughout, we maintain a vector $\hat x\in \mathbb{R}^T$, that is initialized to be $x$. Recall that coordinates of $x$ and $\hat x$ are indexed by points in $T$. Throughout the algorithm, all coordinates are either \emph{active} or \emph{inactive}. Initially, they are all active. The algorithm performs iterations until all coordinates become inactive.

In an iteration, we compute 
\begin{itemize}
\item for each pair $t,t'$ of active coordinates, $\Delta_{t,t'}=\frac{1}{2}\cdot (\hat x_t+\hat x_{t'}-D(t,t'))$; and
\item for each active $t$ and inactive $t'$, $\Delta_{t,t'}=\hat x_t+\hat x_{t'}-D(t,t')$;
\end{itemize}
and we compute $\Delta_t=\min_{t' \in T}\set{\Delta_{t,t'}}$ and $\Delta=\min_{t \in T}\set{\Delta_t}$. Then, for each active coordinate $t$,
\begin{itemize}
\item update $\hat x_t\leftarrow \hat x_t-\Delta$; and
\item if $\Delta_t=\Delta$ holds, mark $t$ inactive.
\end{itemize}

When all coordiates are inactive, we stop and return $\hat x$ as $\proj(x)$.
This completes the description of the algorithm. 

We now prove some of its properties. We start by showing in the following claim that the projection $\proj(x)$ computed by the algorithm indeed lies in $\TS(D)$. The proof is straightforward and is deferred to \Cref{apd: Proof of clm: feasible}. Intuitively, the moment that a coordinate $t$ becomes inactive, some inequality in $\set{\hat x_t+\hat x_{t'}-D(t,t')}_{t'}$ becomes tight and stays tight until the end.

\begin{claim}
\label{clm: feasible}
$\proj(x)$ lies in $\TS(D)$.
\end{claim}

We next prove the following crucial property of the projection function $\proj$, that states that $\proj$ is non-expanding with respect to $\ell_{\infty}$ norm.

\begin{lemma} \label{lem:proj_dis}
For any valid vectors $x,x'$, $\norm{\proj(x)-\proj(x')}_{\infty} \le \norm{x-x'}_{\infty}$.
\end{lemma}
\begin{proof}
We start by proving the following claim.
\begin{claim} \label{clm:project}
For any coordinate $t\in T$, there is another coordinate $t'\in T$, such that 
$$\big(\proj(x)\big)_t+\big(\proj(x)\big)_{t'} = D(t,t'), \text{ and } x_t-\big(\proj(x)\big)_t \ge x_{t'}-\big(\proj(x)\big)_{t'}.$$
\end{claim}

\begin{proof}
Consider the iteration that coordinate $t$ becomes inactive. From the algorithm, there exists another $t'$ such that in that iteration $\Delta=\Delta_{t,t'}(=\Delta_{t}=\Delta_{t'})$. From similar arguments in the proof of \Cref{clm: feasible}, we know that after this iteration, $\hat x_t+\hat x_{t'}=D(t,t')$ holds and coordinate $t'$ is inactive, and they will stay so until the end of the algorithm.
This means that coordinate $t'$ becomes inactive no later than $t'$. From the algorithm, in each iteration, every active coordinate is decreased by the same amount ($\Delta$ in that iteration) and every inactive coordinate stays the same. Therefore, the total decrease in $\hat x_t$ is at least the total decrease in $\hat x_{t'}$, implying that $x_t-\big(\proj(x)\big)_t \ge x_{t'}-\big(\proj(x)\big)_{t'}$. 
\end{proof}

For convenience, we denote $\alpha=\proj(x)$ and $\beta=\proj(x')$, and for each coordinate $t\in T$, we denote
$\hat \alpha_t=x_t-\alpha_t$ and 
$\hat \beta_t=x'_t-\beta_t$. 
We choose $t$ such that (it is easy to see that such $t$ exists)
\begin{itemize}
\item maximizes $|\alpha_t-\beta_t|$; and
\item among all  $t'$ that maximize $|\alpha_t-\beta_t|$, there is no other $t'$ with $\hat\alpha_{t} \ge \hat\alpha_{t'}$ and $\hat\beta_{t} > \hat\beta_{t'}$, or $\hat\alpha_{t} > \hat\alpha_{t'}$ and $\hat\beta_{t} \ge \hat\beta_{t'}$.
\end{itemize}
Assume without loss of generality that $\alpha_t\ge \beta_t$. It suffices to show that, for every $t\in T$, $|\alpha_t-\beta_t| \le \norm{x-x'}_{\infty}$. 
First, by definition, 
\[
\alpha_t-\beta_t = (x_t-\hat \alpha_t)- (x'_t-\hat \beta_t) \le  |x_t-x'_t| - \hat \alpha_t + \hat \beta_t\le \norm{x-x'}_{\infty} - \hat \alpha_t + \hat \beta_t.
\]
By \Cref{clm:project}, there exists another coordinate $t'\in T$ such that $\alpha_t+\alpha_{t'} = D(t,t')$ and $\hat \alpha_t \ge \hat \alpha_{t'}$. 
From triangle inequality $\beta_t+\beta_{t'}\ge D(t,t')$,
	$$\alpha_t-\beta_t =(D(t,t') - \alpha_{t'})-\beta_t \le (D(t,t') - \alpha_{t'}) - (D(t,t') - \beta_{t'}) = \beta_{t'} - \alpha_{t'}.$$
	and using similar arguments and the fact that $\hat \alpha_t \ge \hat \alpha_{t'}$, we can show that
	$$
	\beta_{t'} - \alpha_{t'} \le \norm{x-x'}_{\infty} - \hat\beta_{t'} + \hat\alpha_{t'} \le \norm{x-x'}_{\infty} - \hat\beta_{t'} + \hat\alpha_{t}.
	$$
Altogether, we get that $2\cdot (\alpha_t-\beta_t) \le 2\norm{x-x'}_{\infty} - \hat\beta_{t'} + \hat\beta_{t}$.
Now 
\begin{itemize}
\item if $\hat\beta_{t} \le \hat\beta_{t'}$, then $2\cdot (\alpha_t-\beta_t) \le 2\norm{x-x'}_{\infty}$ and we are done;
\item otherwise, $\hat\beta_{t} > \hat\beta_{t'}$, so the coordinate $t$ satisfies $\hat\alpha_{t} \ge \hat\alpha_{t'}$, $\hat\beta_{t} > \hat\beta_{t'}$ and $\alpha_t-\beta_t \le \beta_{t'}-\alpha_{t'}$, and this contradicts our choices of coordinate $t$.
\end{itemize}  
\end{proof}

\begin{remark}
Although $\proj(x)$ is called the \emph{projection} of $x$ into $\TS(D)$, it is in general \underline{not true} that $\proj(x)\in \arg_{y\in \TS(D)}\min\set{||y-x||_{\infty}}$.
For example, let $D$ be a metric on three points $\set{a,b,c}$, where $D(a,b)=D(b,c)=D(c,a)=4$. Then $x=(1,3,5)$ is a valid vector, 
and its projection onto $\TS(D)$ (calculated from the algorithm above) is $\proj(x)=(1,3,3)$, so $||x-\proj(x)||_{\infty}=2$. However, the point $a=(0,4,4)$ satisfies that $||x-a||_{\infty}=1$.
\end{remark}

\subsection{Metric $(\cdot,\cdot)_{\TS}$} 

At the end of this section, we define a metric $(\cdot,\cdot)_{\TS}$ over all points in $\TS(D)$ as follows. For every pair $x,y\in \TS(D)$, we define $$(x,y)_{\TS}=||x-y||_{\infty}.$$ 
We introduce this definition not merely for simplifying the notations, but for essentially relating $(x,y)_{\TS}$ to the ``graph structure'' of the tight span.

For example, from \Cref{fig: TS_3}, the tight span for a $3$-point metric exhibits a $3$-leg star, and the metric $(\cdot,\cdot)_{\TS}$ is indeed the ``shortest-path distance'' in this ``graph''. Specifically, for a point $x=(1,6,4)$ on line $\overline{oa}$ and another point $y=(3,6,2)$ on line $\overline{oc}$, their shortest path is $x$-$o$-$y$ whose length is $(x,o)_{\TS}+(o,y)_{\TS}=2$. Therefore, in this case $\TS(D)$ is  the union of line-metrics $\overline{oa},\overline{ob},\overline{oc}$ with a common line-endpoint $o$.

As another example (illustrated in \Cref{fig: TS_4}), the tight span of a $4$-point metric is essentially the union of a $2$-dimensional $\ell_1$-metric space (the rectangle $\overline{a'b'c'd'}$) with four $1$-dimensional line-metric pendants. Intuitively, if we replace the rectangle with a ``tightly knotted grid'' (aligned with axis $a'b'$ and $a'd'$) and replace the line pendants with paths (with vertices densely packed in it), then the metric $(\cdot,\cdot)_{\TS}$ is indeed the shortest-path distance in this ``graph''.

\newcommand{\opt}{\mathsf{OPT}}

\section{Graphs with $5$ Terminals: Proof of \Cref{thm: upper 5}}
\label{sec: upper}

In this section we provide the proof of \Cref{thm: upper 5}. 
Recall that we are given an instance $(G,T,\ell)$ of $\zesn$.
Denote $\opt=\sum_{(u,v)\in E}c(u,v)\cdot\dist_{\ell}(u,v)$. The goal is to show that there exists a solution $(\fset,\delta)$ with $|\fset|\le 30$ and $\vol(\fset,\delta)\le \opt$. 
We denote by $D$ the metric on $T$ induced by the shortest-path distance $\dist_{\ell}(\cdot,\cdot)$.
The proof consists of two steps. In the first step, we project vertices of $G$ onto the tight span $\TS(D)$ using the algorithm in \Cref{subsec: proj}. In the second step, we construct the solution $(\fset,\delta)$ by properly decomposing the tight span $\TS(D)$.


\subsection{Step 1. Project the vertices onto $\TS(D)$}

For each $v\in V(G)$, we define vector $x^v=(\dist_{\ell}(v,t))_{t\in T}$. That is, $x^v$ is a $|T|$-dimensional vector, such that each coordinate is indexed by a terminal $t\in T$, and the value of this coordinate is the shortest-path distance between $v$ and $t$. Clearly, $x^v$ is valid, as all triangle inequalities $\dist_{\ell}(v,t)+\dist_{\ell}(v,t')\ge \dist_{\ell}(t,t')=D(t,t')$ hold for all $t,t'\in T$. Moreover, for every pair $u,v$ of vertices in $G$, 
$$\norm{x^u-x^v}_{\infty}=\max_{t\in T}\set{|\dist_{\ell}(u,t)-\dist_{\ell}(v,t)|}\le \dist_{\ell}(u,v).$$

We then compute the tight span $\TS(D)$ of $D$, and use the algorithm in \Cref{subsec: proj} to project each vector $x^v$ onto $\TS(D)$. Denote $p^v=\proj(x^v)$. From \Cref{lem:proj_dis}, for all edge $(u,v)\in E(G)$,
$$(p^u,p^v)_{\TS}=\norm{p^u-p^v}_{\infty}\le \norm{x^u-x^v}_{\infty}\le \dist_{\ell}(u,v).$$

\subsection{Step 2. Construct a solution by partitioning $\TS(D)$}

From the previous results in \cite{koolen2009optimal}, the tight span of metrics on five points can be classified into three types. We will not discuss the classification in details, but will focus on the ``shape/structure'' of each type, and show the construction of a low-cost solution for them.
For readability, we provide full details for one type here, and defer the analysis of other two types to \Cref{apd: type 2 and 3}

\paragraph{Type 1.}
The first type of tight spans on $5$ vertices $\set{a,b,c,d,e}$, as illustrated in \Cref{fig: TS_5_1}, consist of five $1$-dimensional sets $L_a,L_b,L_c,L_d,L_e$, and five $2$-dimensional sets $R_a,R_b,R_c,R_d,R_e$.  
The tight span is determined by parameters $\set{l_a,l_b,l_c,l_d,l_e,l_{ab},l_{bc},l_{cd},l_{de},l_{ea}}$, which are uniquely determined by the distances in $D$. We do not discuss the calculation of these parameters here as it is irrelevant from our construction, and we refer the interested readers to \cite{koolen2009optimal}.

The set $L_a$ is the $a$-$a'$ line metric whose length is $l_a$. That is, the point $x$ in $L_a$ at distance $r$ from $a$ is at distance $D(a,t)-r$ from $t$ (in metric $(\cdot,\cdot)_{\TS}$), for every $t\in \set{b,c,d,e}$.
The tight inequalities for $x$ in $L_a$ are
$\set{x_a+x_t=D(a,t)}_{t\in \set{b,c,d,e}}$.
The other 1-dimensional sets can be defined similarly.

The set $R_a$ is the $o$-$ab$-$a'$-$ea$ rectangle with $\ell_1$ metric. That is, the point $x$ in $R_a$ at distance $r$ from the $a'$-$ab$ line and at distance $r'$ from the $a'$-$ea$ line is, in metric $(\cdot,\cdot)_{\TS}$, at distance 
\begin{itemize}
\item $r+r'+l_a$ from $a$;
\item $r+(l_{ea}-r')+l_{bc}+l_b$ from $b$;
\item $(l_{ab}-r)+(l_{ea}-r')+l_{bc}+l_{cd}+l_c$ from $c$;
\item $(l_{ab}-r)+(l_{ea}-r')+l_{cd}+l_{de}+l_d$ from $d$;
\item $(l_{ab}-r)+r'+l_{de}+l_e$ from $e$.
\end{itemize}
The tight inequalities for $x\in R_a$ are
$x_a+x_c=D(a,c)$, $x_a+x_d=D(a,d)$, and $x_b+x_e=D(b,e)$.
The other 2-dimensional sets can be defined similarly.

\begin{figure}[h]
	\centering
	\subfigure[A type-1 tight span on the set $\set{a,b,c,d,e}$, with its five $2$-dimensional sets $R_a$,$R_b$,$R_c$,$R_d$,$R_e$ shown in yellow, gray, pink, blue, and green, respectively.]
	{\scalebox{0.12}{\includegraphics{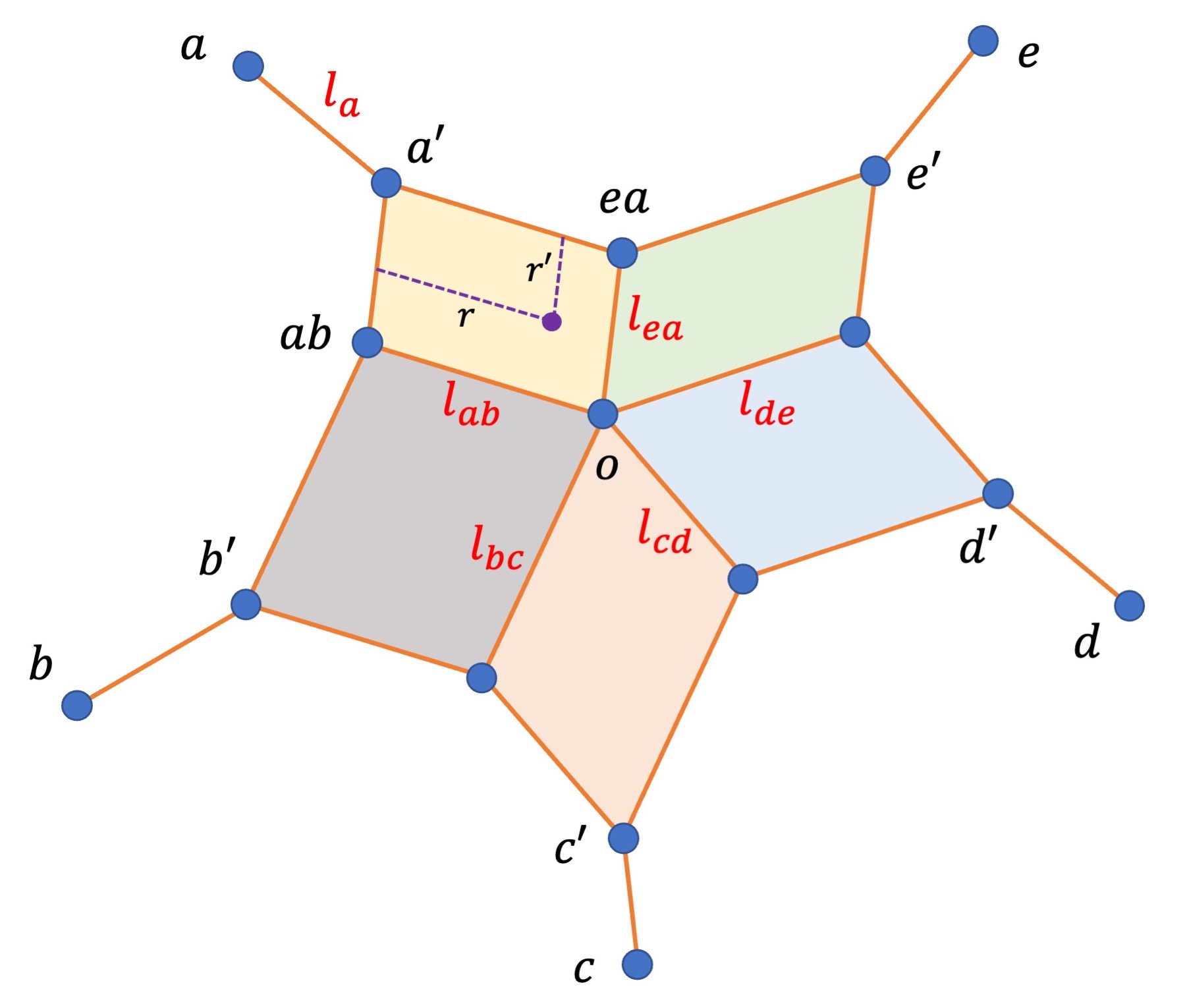}}}
	\hspace{0.7cm}
	\subfigure[The construction of $\fset$ based on a random partition of $\TS(D)$, where the region corresponding to the sets $F_o,F_{bc},F_{d'}$ are shaded.]
	{
		\scalebox{0.12}{\includegraphics{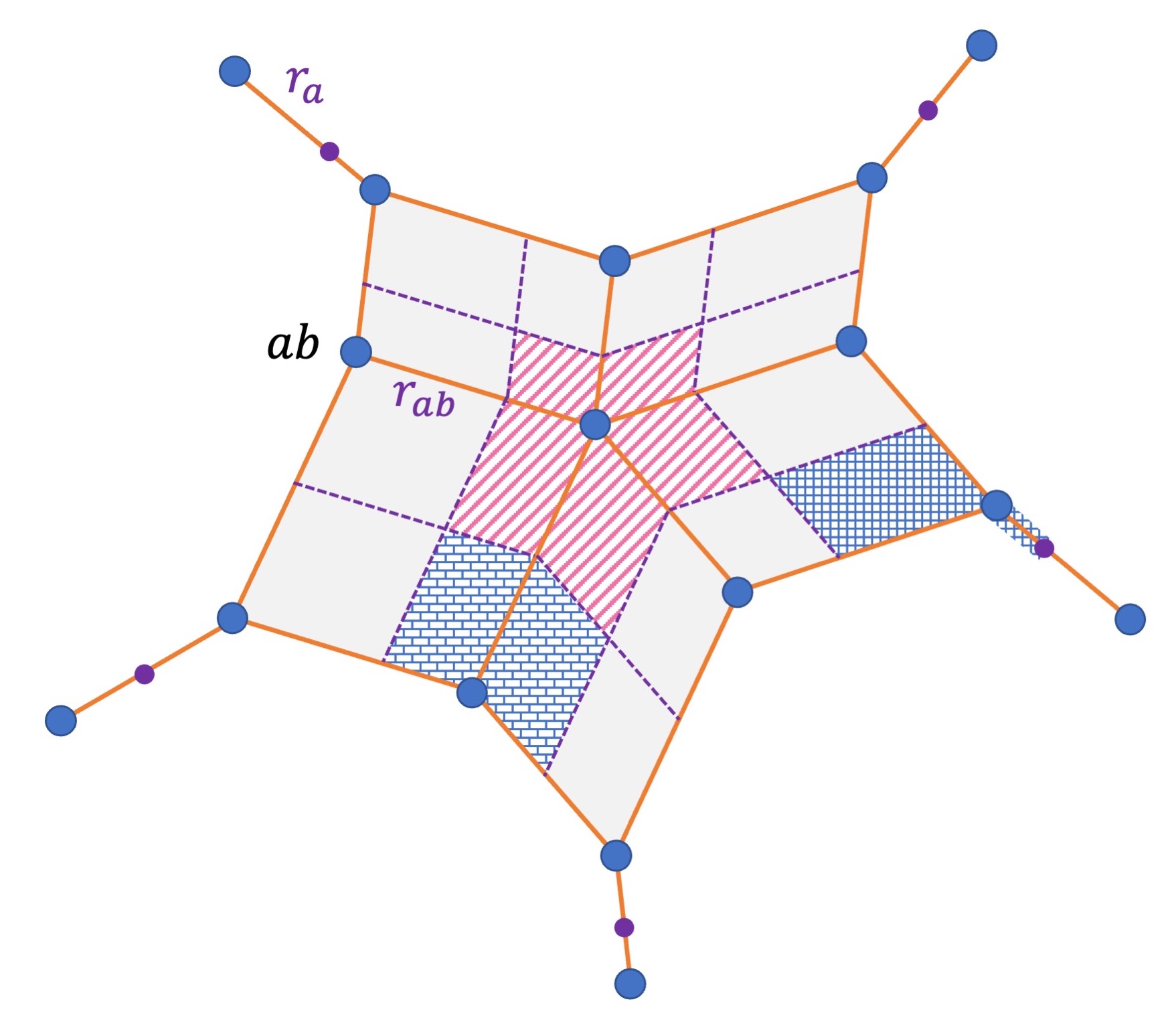}\label{fig: TS 5}}}
	\caption{An illustration of a type-1 tight span on $5$ points (left) and its decomposition (right).\label{fig: TS_5_1}}
\end{figure}

Rectangles $R_a,R_b,R_c,R_d,R_e$ appear in this order. Consecutive rectangles $R_a,R_b$ share its boundary $o$-$ab$ line (and similarly, $R_b,R_c$ share $o$-$bc$ line, etc). All rectangles share an endpoint $o$.
Intuitively, $(x,y)_{\TS}$ is the ``shortest distance one has to travel in canonical directions from $x$ to $y$''. For example,
\begin{itemize}
\item if a pair $x,y\in \TS(D)$ belong to consecutive rectangles say $x\in R_a, y\in R_b$, then $x$ needs to first reach line $o$-$ab$ via a rectilinear path in $R_a$, and then reach $y$ via a rectilinear path in $R_b$; in fact $(x,y)_{\TS}$ is essentially their $\ell_1$ distance in the big rectangle $R_a\cup R_b$ (that is, rectangle $a'$-$b'$-$bc$-$ea$), with axis $a'$-$b'$ and $a'$-$ea$;
\item if a pair $x,y\in \TS(D)$ belong to non-consecutive rectangles say $x\in R_a, y\in R_c$, then $(x,y)_{\TS}$ is the $x$-$o$ $\ell_1$-distance in  $R_a$ plus the $o$-$y$ $\ell_1$-distance in $R_c$.
\end{itemize}

We now proceed to construct a solution $(\fset,\delta)$ whose cost is bounded by $\opt$, when the tight span $\TS(D)$ is of type 1.

The collection $\fset$ contains a set for each node marked blue in \Cref{fig: TS_5_1}. That is, 
\[\fset=\set{F_a,F_b,F_c,F_d,F_e,
	F_{a'},F_{b'},F_{c'},F_{d'},F_{e'}, F_{ab},F_{bc},F_{cd},F_{de},F_{ea},F_o},\]
and so $|\fset|=16$. Here sets $F_a,F_b,F_c,F_d,F_e$ contain terminals, and others do not.

For the points in set $L_a$, we pick a number $r_a$ uniform at random from $[0,l_a]$, and then define
\[
F_a=\bigg\{u\text{ }\bigg|\text{ } p^u \text{ lies in }L_a\text{ and is at distance}<r_a\text{ from }a\bigg\};\text{ and }F^{L_a}_{a'}=L_a\setminus F_a,
\]
where by ``distance'' we mean the distance in $(\cdot,\cdot)_{\TS}$.
The other $1$-dimensional sets $L_b,\ldots,L_e$ are partitioned in a similar way (into $(F_b,F^{L_b}_{b'}),\ldots,(F_e,F^{L_e}_{e'})$, respectively), based on random numbers $r_b,\ldots,r_e$ from intervals $[0,l_b],\ldots,[0,l_e]$ respectively.

We then partition the rectangles. We first independently pick random numbers $r_{ab}\in [0,l_{ab}],r_{bc}\in [0,l_{bc}],r_{cd}\in [0,l_{cd}],r_{de}\in [0,l_{de}],r_{ea}\in [0,l_{ea}]$, respectively.
For  $R_a$, we define
\begin{itemize}
\item $F^{R_a}_{a'}=\bigg\{u\text{ }\bigg|\text{ } p^u \text{ lies in }R_a\text{ and is at distance}<r_{ab}\text{ from }ab\text{ and}< r_{ea}\text{ from }ea\bigg\}$;
\item $F^{R_a}_{ab}=\bigg\{u\text{ }\bigg|\text{ } p^u \text{ lies in }R_a\text{ and is at distance}<r_{ab}\text{ from }ab\text{ and}\ge r_{ea}\text{ from }ea\bigg\}$;
\item $F^{R_a}_{ea}=\bigg\{u\text{ }\bigg|\text{ } p^u \text{ lies in }R_a\text{ and is at distance}\ge r_{ab}\text{ from }ab\text{ and}< r_{ea}\text{ from }ea\bigg\}$; and
\item $F^{R_a}_{o}=\bigg\{u\text{ }\bigg|\text{ } p^u \text{ lies in }R_a\text{ and is at distance}\ge r_{ab}\text{ from }ab\text{ and}\ge r_{ea}\text{ from }ea\bigg\}$.
\end{itemize}

The other $2$-dimensional sets are partitioned in a similar way.
Finally, we aggregate the partitioning constructed above, by setting
\begin{itemize}
\item $F_{a'}=F^{L_a}_{a'}\cup F^{R_a}_{a'}$ (and similarly for $F_{b'},\ldots,F_{e'}$);
\item $F_{ab}=F^{R_a}_{ab}\cup F^{R_b}_{ab}$ (and similarly for $F_{bc},\ldots,F_{ea}$); and 
\item $F_o=F^{R_a}_{o}\cup F^{R_b}_{o}\cup F^{R_c}_{o}\cup F^{R_d}_{o}\cup F^{R_e}_{o}$.
\end{itemize}
This completes the construction of $\fset$. See \Cref{fig: TS 5} for an illustration. As we have projected all vertices in $G$ to $\TS(D)$ in the first step, clearly $\fset$ is a partition of $V(G)$.
We then define the metric $\delta$ on $\fset$ as such that, for every pair $F_{t},F_{t'}$ (where $t$ is some blue node in $\TS(D)$), $\delta(F_t,F_{t'})=(t,t')_{\TS}$.

We now show that the (random) solution $(\fset,\delta)$ satisfies that $\mathbb{E}[\vol(\fset,\delta)]\le \opt$. In fact, we will prove the following lemma in the next subsection.

\begin{lemma}
\label{lem: non-expanding}
For each edge $(u,v)\in E(G)$, if $u\in F_t$ and $v\in F_{t'}$, then $\mathbb{E}[(t,t')_{\TS}]\le (p^u,p^v)_{\TS}$.
\end{lemma}

Note that this lemma immediately implies that
\[\mathbb{E}[\vol(\fset,\delta)]\le\sum_{(u,v)\in E(G)}c(u,v)\cdot(p^u,p^v)_{\TS}\le\sum_{(u,v)\in E(G)}c(u,v)\cdot\dist_{\ell}(u,v)= \opt,\]
completing the proof of \Cref{thm: upper 5} in the case where the metric $D$ has a type-1 tight span $\TS(D)$.

\subsection{Proof of \Cref{lem: non-expanding}}

Essentially, \Cref{lem: non-expanding} is true because the metric space $(\TS(D),(\cdot,\cdot)_{\TS})$ is the union of $1$-dimensional and $2$-dimensional $\ell_1$ metrics, and $\ell_1$ metrics admit simple non-expanding randomized decompositions.
We now provide the proof, tailored to the structure of $\TS(D)$.
We denote by $F_t$ ($F_{t'}$, resp.) the set in $\fset$ that contains vertex $u$ ($v$, resp.).
We distinguish between the following cases.

\subsubsection*{Case 1. $p^u$ and $p^v$ lie in the same $1$-dimensional or $2$-dimensional set}

Assume first that $p^u$ and $p^v$ lie in the same $1$-dimensional set. Assume without loss of generality that $p^u,p^v\in L_a$. Denote $a_u=(p^u,a)_{\TS}$ and $a_v=(p^v,a)_{\TS}$, and assume without loss of generality that $a_u\le a_v$. Then
\begin{itemize}
\item with probability $\frac{a_u}{l_a}$, $u,v$ both go to $F^{L_a}_{a'}$ (and so $F_{a'}$);
\item with probability $\frac{a_v-a_u}{l_a}$, $u$ goes to $F_a$ and $v$ goes to $F^{L_a}_{a'}$ (and so $F_{a'}$); and
\item with probability $1-\frac{a_v}{l_a}$, $u,v$ both go to $F^{L_a}_{a}$ (and so $F_{a}$).
\end{itemize}
Therefore, $\mathbb{E}[(t,t')_{\TS}]=\frac{a_v-a_u}{l_a}\cdot (a,a')_{\TS}=\frac{a_v-a_u}{l_a}\cdot l_a=a_v-a_u=(p^u,p^v)_{\TS}$.

Assume now that $p^u$ and $p^v$ lie in the same $2$-dimensional set, and assume without loss of generality that $p^u,p^v\in R_a$. By definition of $R_a$, let $p^u=(x^u,y^u)$ and $p^v=(x^v,y^v)$ (see \Cref{fig: TS_sep} for an illustration), and assume without loss of generality that $x^u\le x^v$ and $y^u\le y^v$. 
Then
\begin{itemize}
	\item with probability $\frac{x^u}{l_{ab}}$, $u,v$ both go to $F_{o}\cup F_{ea}$;
	\item with probability $\frac{x^v-x^u}{l_{ab}}$, $u$ goes to $F_{a'}\cup F_{ab}$ and $v$ goes to $F_{o}\cup F_{ea}$; and
	\item with probability $1-\frac{x^v}{l_{ab}}$, $u,v$ both go to $F_{a'}\cup F_{ab}$.
\end{itemize}
\begin{figure}[h]
	\centering
	\scalebox{0.1}{\includegraphics{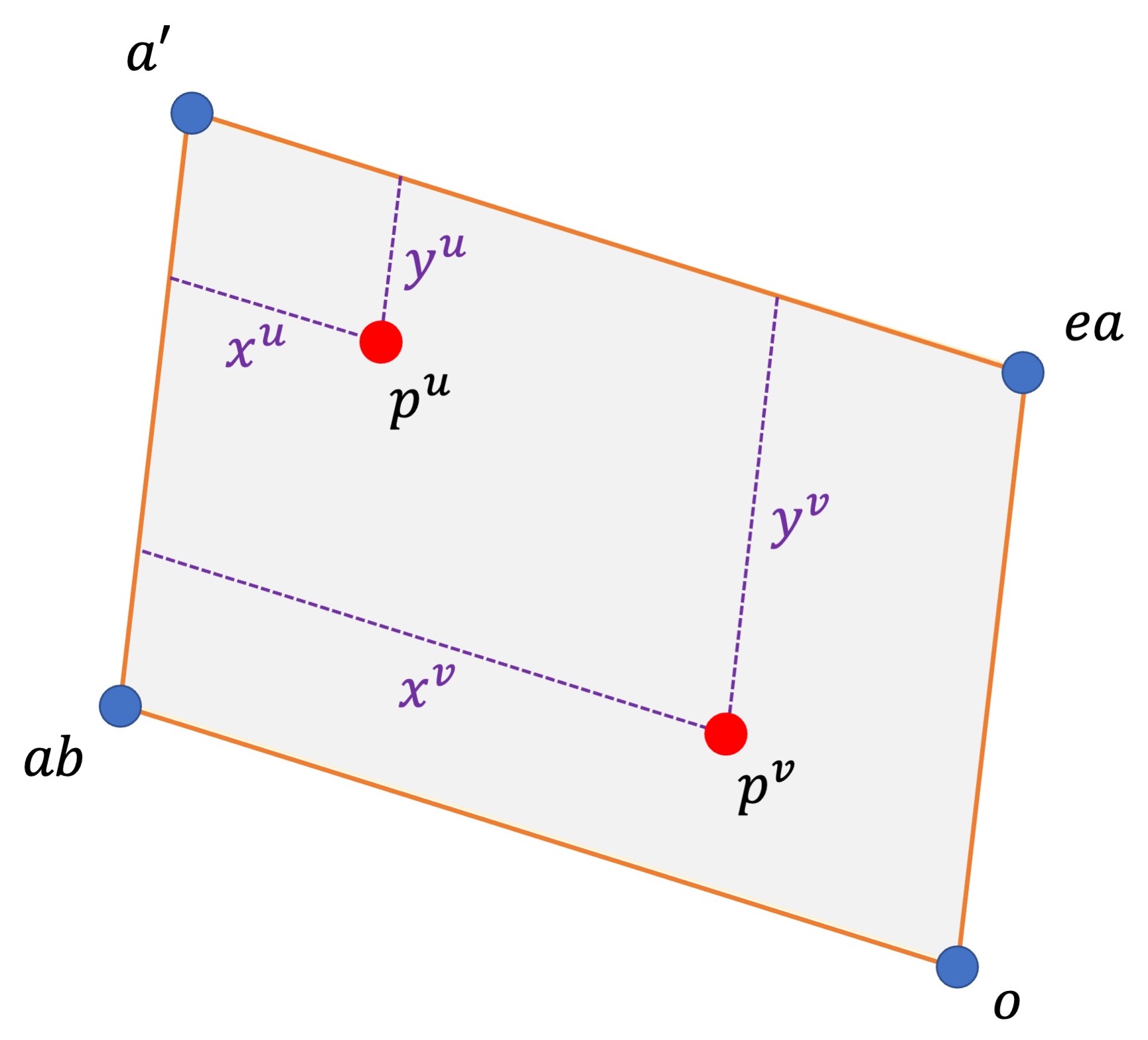}}
	\caption{An illustration of points $p^u$ and $p^v$ lying in the same $2$-dimensional set.\label{fig: TS_sep}}
\end{figure}
We can similarly calculate the probability of $u$ and $v$ going to $F_{a'}\cup F_{ea}$ and $F_{o}\cup F_{ab}$, using $y^u,y^v$ and $l_{ea}$.
Altogether, we get that \[\mathbb{E}[(t,t')_{\TS}]=\frac{x^v-x^u}{l_{ab}}\cdot (o,ab)_{\TS}+\frac{y^v-y^u}{l_{ea}}\cdot (o,ea)_{\TS}=(x^v-x^u)+(y^v-y^u)=(p^u,p^v)_{\TS}.\]



\subsubsection*{Case 2. $p^u$ and $p^v$ lie in consecutive $2$-dimensional sets}

Assume without loss of generality that $p^u\in R_a$ and $p^v\in R_e$. Let $p^u=(x^u,y^u)$ and $p^v=(x^v,y^v)$ (see \Cref{fig: TS_sep_2} for an illustration). Assume without loss of generality that $x^u\le x^v$ and $y^u\le y^v$. Let $w$ be any node on $o$-$ea$ line with $y^u\le y^w:=(w,ea)_{\TS}\le y^v$.
\begin{figure}[h]
	\centering
	\scalebox{0.1}{\includegraphics{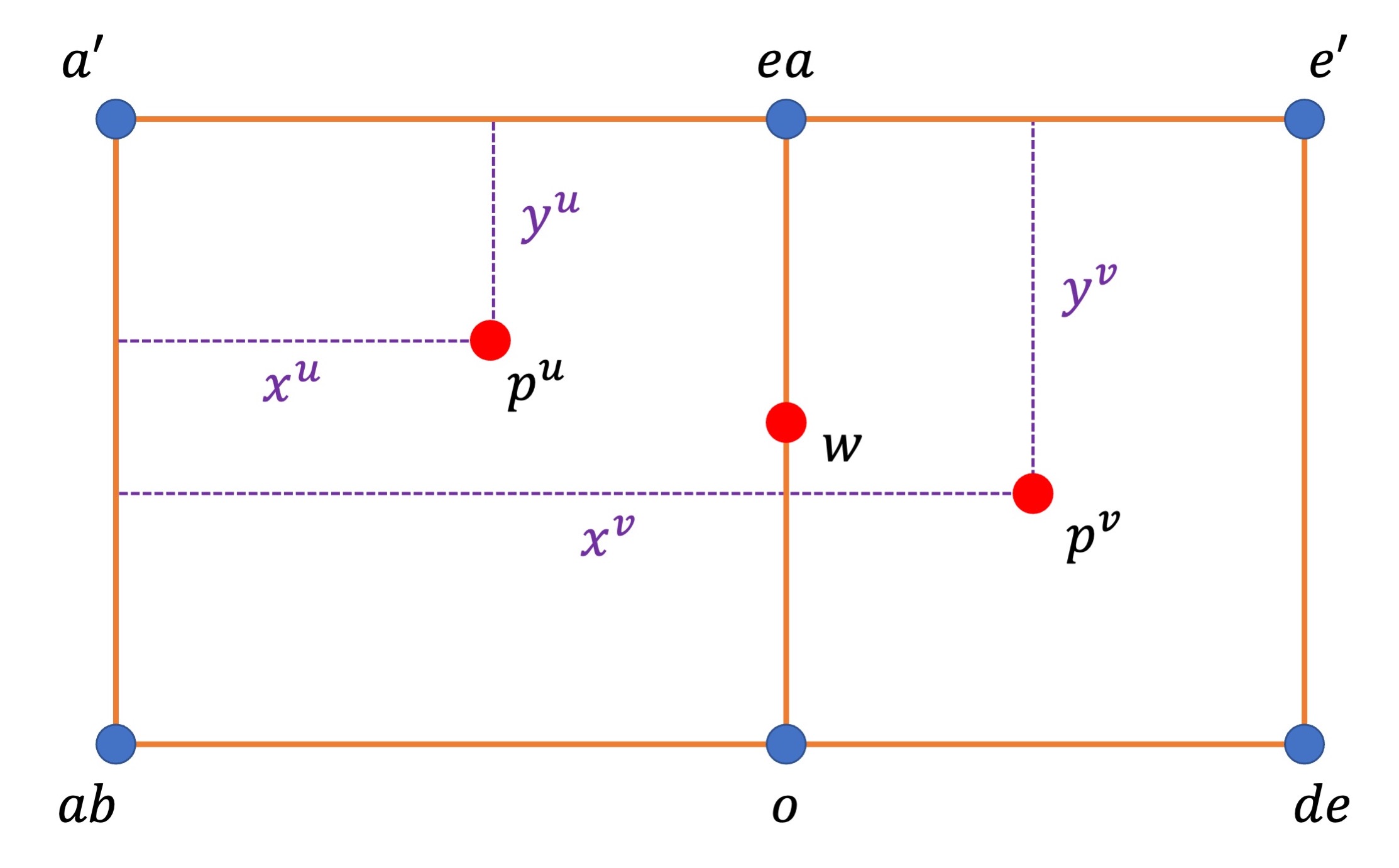}}
	\caption{An illustration of a points $p^u$ and $p^v$ lying in consecutive $2$-dimensional sets.\label{fig: TS_sep_2}}
\end{figure}

Then via similar arguments in Case 1, we can show that, 
(denoting $F_{t''}$ as the set that contains $w$)
\[
\begin{split}
\mathbb{E}[(t,t')_{\TS}]& \le \mathbb{E}[(t,t'')_{\TS}+(t'',t')_{\TS}]\\ & \le \bigg((l_{ab}-x^u)+(y^w-y^u)\bigg)+\bigg((x^v-l_{ab})+(y^v-y^w)\bigg)\\
& = (x^v-x^u)+(y^v-y^u)= (p^u,p^v)_{\TS}.
\end{split}
\]

\subsubsection*{Case 3. $p^u$ and $p^v$ lie in non-consecutive $2$-dimensional sets}

Assume without loss of generality that $p^u\in R_a$ and $p^v\in R_c$. Let $p^u=(x^u,y^u)$ and $p^v=(x^v,y^v)$ (see \Cref{fig: TS_sep_3} for an illustration). 
\begin{figure}[h]
	\centering
	\scalebox{0.1}{\includegraphics{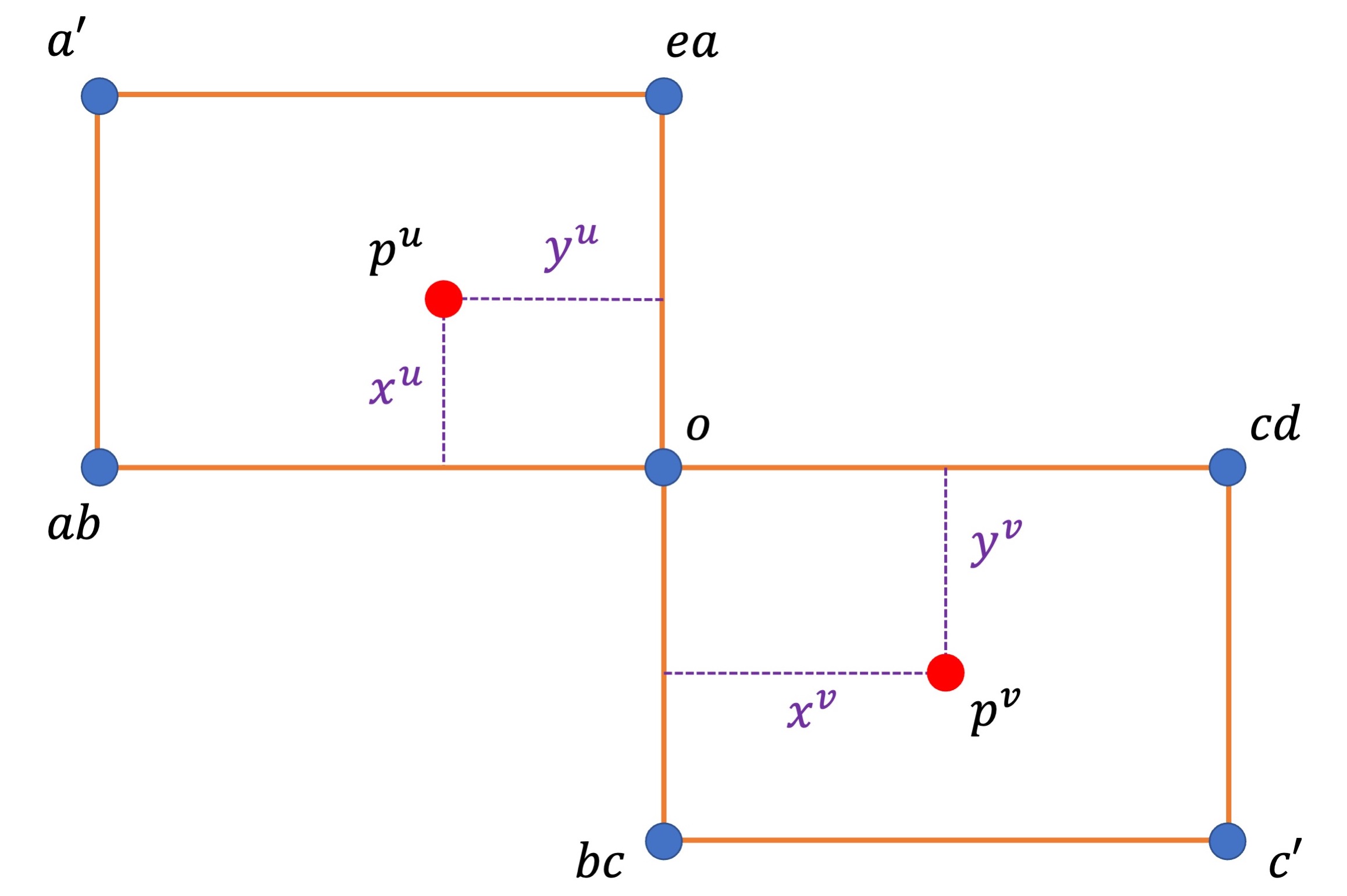}}
	\caption{An illustration of a points $p^u$ and $p^v$ lying in non-consecutive $2$-dimensional sets.\label{fig: TS_sep_3}}
\end{figure}
Then via similar arguments in Case 1, we can show that, 
%
%
\[
\mathbb{E}[(t,t')_{\TS}] \le \mathbb{E}[(t,o)_{\TS}+(o,t')_{\TS}] = (x^v+x^u)+(y^v+y^u)= (p^u,p^v)_{\TS}.
\]

The remaining case where $u$ lies in a $1$-dimensional set (say $L_a$) and $v$ lies in a $2$-dimensional set can be reduced to one of the above cases, depending on which rectangle contains $v$, as in this case $(u,v)_\TS=(u,a)_{\TS}+(a,v)_{\TS}$ always holds.
The other remaining case where $u,v$ belong to different $1$-dimensional sets can be proved in a similar way.

\section{Graphs with $6$ Terminals: Proof of \Cref{thm: lower 6}}

\label{sec: lower}

In this section, we provide the proof of \Cref{thm: lower 6}. We will first prove \Cref{thm: lower 6} for the problem $\zesn$ (that is, we will show that, for any integer $N$, there exists an instance $(G,T,\ell)$ of $\zesn$ with $|T|= 6$, such that any solution $(\fset,\delta)$ with $|\fset|\le N$ must satisfy that $\vol(\fset,\delta)\ge (1+\Omega(1))\cdot \sum_{(u,v)\in E}c(u,v)\cdot\dist_{\ell}(u,v)$), and then generalize it to problem $\zesn_{\textsf{ave}}$ in \Cref{sec: ave}.

Let $N$ be any integer. We will first construct a hard instance $(G,T,\ell)$ in \Cref{subsec: instance} with $|T|=6$, and then show in \Cref{subsec: onto plane,subsec: basic,subsec: plane} that any solution $(\fset,\delta)$ to this instance with $|\fset|\le N$ must satisfy that $\vol(\fset,\delta)\ge (1+\Omega(1))\cdot \sum_{(u,v)\in E}c(u,v)\cdot\dist_{\ell}(u,v)$.


\subsection{The hard instance}
\label{subsec: instance}

Recall that in an instance $(G,T,\ell)$ of $\zesn$, $G$ is a graph and $\ell$ is its edge weight function. Denote by $D$ the shortest-path distance metric on $T$ induced by $\dist_{\ell}(\cdot,\cdot)$.
We will first define the metric $(D,T)$, and then define $G$ (and $\ell$) based on it.

For convenience, we denote $T=\set{a,b,c,d,e,f}$. The metric $D$ on $T$ is given by table~\ref{fig: table}.
Before we define $G$, we first describe the tight span $\TS(D)$ (see \Cref{fig: TS_6_1} for an illustration). 
It consists of a 2-dimensional set, which is a rectangle with endpoints $b,c,d,f$, and a $3$-dimensional set, which is a triangular prism with the top-face $\Delta_{\text{acg}}$ and the bottom face $\Delta_{\text{efh}}$.
For all points $v$ in the prism, the tight constraints are
$v_a+v_d=D(a,d)=3$, $v_b+v_e=D(b,e)=3$, and $v_c+v_f=D(c,f)=3$.
We provide in \Cref{apd: structure of TS} a detailed explanation on why the tight span $\TS(D)$ is in this shape.

\begin{figure}[h]
	\centering
	\subfigure
	{\scalebox{0.1}{\includegraphics{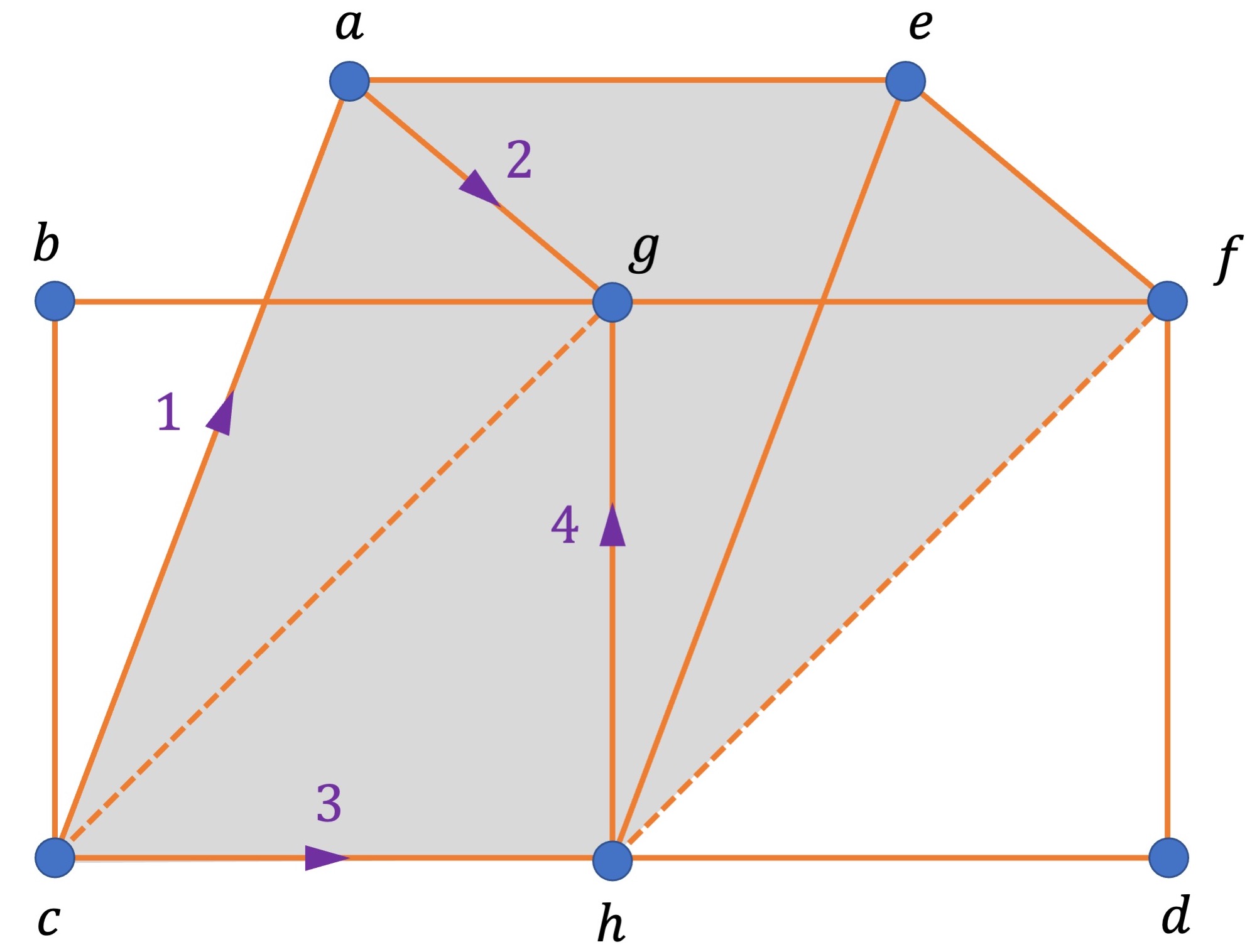}}\label{fig: TS_6_1}}
	\hspace{0.7cm}
	\subfigure
	{
		\scalebox{0.27}{\includegraphics{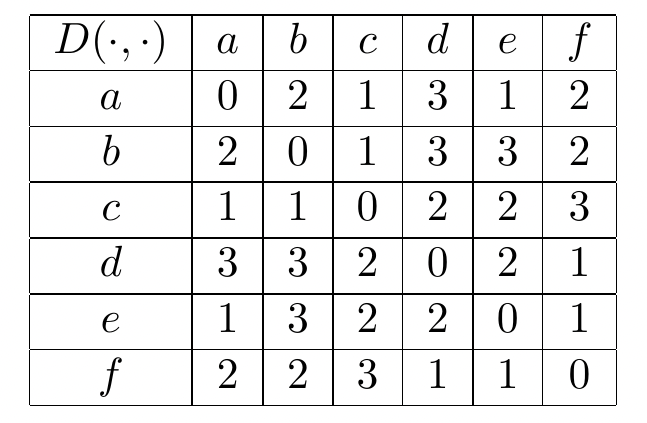}\label{fig: table}}}
	\caption{Metric $D$ (right) and its tight span $\TS(D)$ (left). $D$ can be viewed as the shortest-path distance on the graph induced by all solid lines (with length $1$ each). The dashed lines are of length $2$. The ``prism'' part is shown in gray. The four critical directions are marked by $1$-$4$.\label{fig: TS_6}}
\end{figure}



\paragraph{The structure of $\TS(D)$.}

We now discuss how to determine the distance (in $(\cdot,\cdot)_{\TS}$) between points in  $\TS(D)$. There are four directions in $\TS(D)$ that we call \emph{critical directions}: $\vec{ca}$ (direction 1), $\vec{ag}$ (direction 2), $\vec{ch}$ (direction 3), and $\vec{hg}$ (direction 4). See \Cref{fig: TS_6} for an illustration. Under $(\cdot,\cdot)_{\TS}$, the distances between pairs $(c,h),(c,a),(a,g),(h,g)$ are all $1$. 
Let $v,v'$ be a pair in $\TS(D)$, the distance between $v,v'$ under $(\cdot,\cdot)_{\TS}$ is in fact the \emph{shortest distance one has to travel only in the critical directions to go from $v$ to $v'$}.
For example, consider the pair $(a,d)$. It can either goes
\begin{itemize}
\item $a\xrightarrow{\vec{ch},1} e \xrightarrow{\vec{ag},1} f \xrightarrow{\vec{hg},-1} d$, with the total travel distance $3$; or
\item $a\xrightarrow{\vec{ca},-1} c \xrightarrow{\vec{ch},2} d$, with the total travel distance $3$; or
\item $a\xrightarrow{\vec{ag},1/2} \textsf{Mid}(a,g) \xrightarrow{\vec{ca},-1/2} \textsf{Mid}(c,g)
\xrightarrow{\vec{ch},1} \textsf{Mid}(h,f)
\xrightarrow{\vec{hg},-1/2} \textsf{Mid}(h,g) 
\xrightarrow{\vec{ch},1/2} d$, \\
where $\textsf{Mid}(a,g)$ is the midpoint between $a$ and $g$, and similar for others,
with the total travel distance $3$ (see \Cref{fig: TS6_direction} for an illustration).
\end{itemize}
All these ways of travelling are shortest from $a$ to $d$.


\begin{figure}[h]
	\centering
	\scalebox{0.1}{\includegraphics{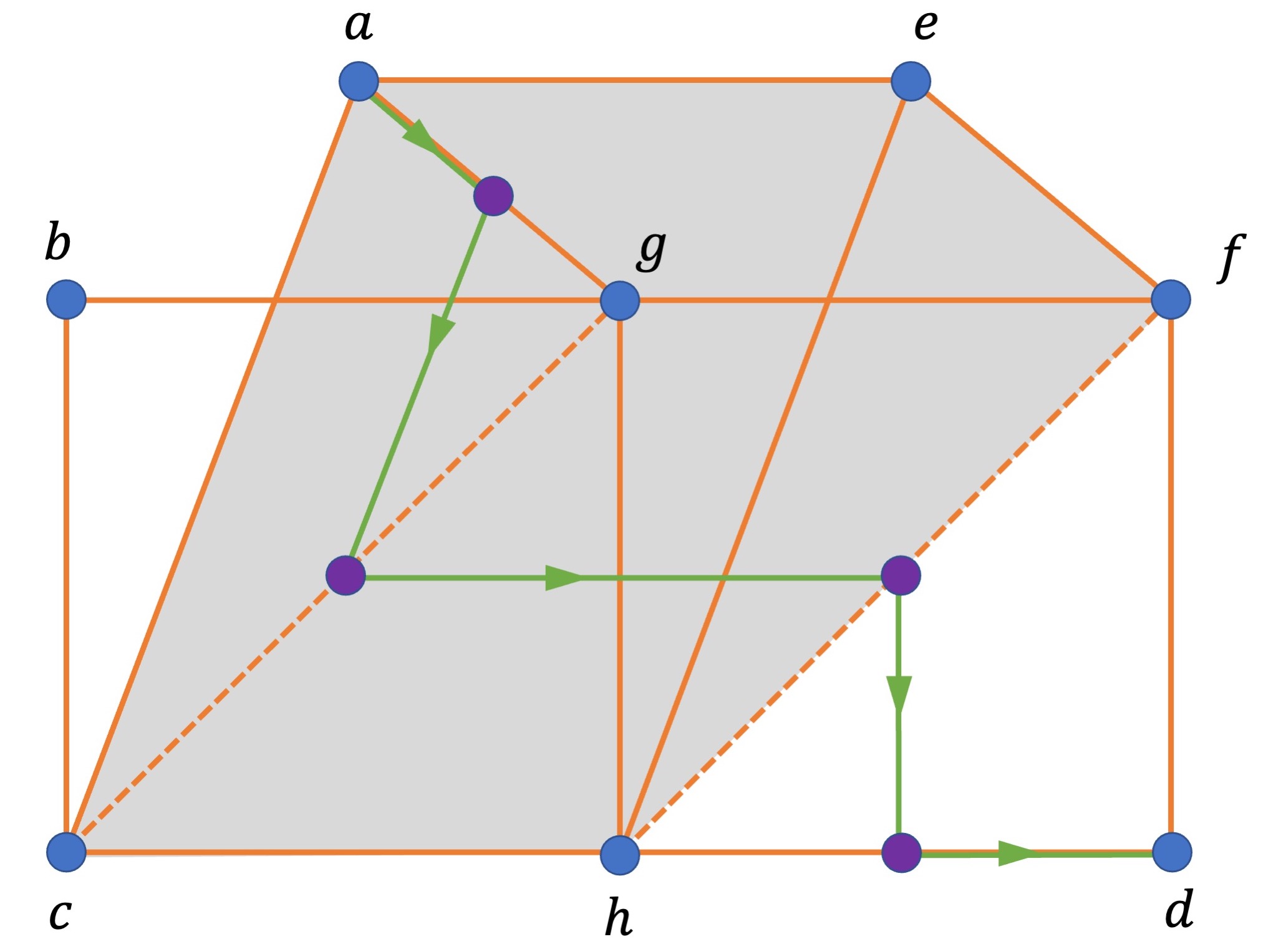}}
	\caption{A way of travelling from $a$ to $d$ only in critical directions.\label{fig: TS6_direction}}
\end{figure}

\paragraph{Associated vectors.}
In order to talk about points in $\TS(D)$ in a more convenient way, for each point $v\in \TS(D)$, instead of using the $6$-dimensional vector $(v_a,v_b,v_c,v_d,v_e,v_f)$, we will represent it in some other way tailored to the structure of $\TS(D)$.
Specifically, we form a (non-orthogonal) coordinate system as follows. Let $\vec{ch}$ be the $x$-direction, let $\vec{cg}$ be the $y$-direction, let $\vec{ca}$ be the $z$-direction, and let $c$ be the reference point, so every point is uniquely represented as a $3$-dimensional vector $(v[x],v[y],v[z])$ in this system, which we call its \emph{associated vector}.
For example, the associated vector for terminal $a$ is $(0,0,1)$, for $h$ is $(1,0,0)$ and for $g$ is $(0,2,0)$.
More generally, 
\begin{itemize}
\item if $v$ is in the rectangle, then $v[x] = v_b-1$, $v[x]+v[y]=v_c$ and $v[z]=0$;
\item if $v$ is in the prism, then $v[x]+1-v[z] = v_a$, $v[x]+v[z] = v_b$ and $v[x]+v[y]+v[z] = v_c$. 
\end{itemize}
Solving and combining them, we obtain the following observation.
\begin{observation} \label{obs:terminal}
For every $v\in \TS(D)$, 
\[v_a = \card{v[x]}+1-v[z], \quad \quad v_b=v[x]+1+v[z], \quad \quad v_c=v[x]+v[y]+v[z],\]
\[v_d=2-v[x]+v[z], \quad \quad v_e=\card{v[x]-1}+1-v[z],\quad \quad v_f=3-v[x]-v[y]-v[z].\]
\end{observation}

\begin{observation}
\label{obs: dist in TS}
If the associated vector of $v$ is $(x,y,z)$ and the associated vector of $v'$ is $(x',y',z')$, then 
    \begin{itemize}
        \item $(v,v')_{\TS} \ge |x-x'|+|z-z'|$;
        \item if $z=z'=0$, then $(v,v')_{\TS} = \frac{1}{2}\cdot\big(\card{(2x+y)-(2x'+y')} + \card{y-y'}\big)$.
    \end{itemize}
\end{observation}
\begin{proof}
    Remember that the shortest path between $v$ and $v'$ has to travel only in the critical directions. If we travel in direction $1$ or $2$, then the $x$ coordinate  will not change, and the change of $z$ coordinate is exactly the distance it travelled. If we travel in direction $3$ or $4$, the $z$ coordinate of a point will not change, and the change of $x$ coordinate is exactly the distance it travelled. Thus for any shortest path between $v$ and $v'$, we need to travel at least $|x-x'|$ on direction $1$ and $2$, and at least $|z-z'|$ on direction $3$ or $4$. Therefore, $(v,v')_{\TS} \ge |x-x'|+|z-z'|$.

    Note that $z=z'=0$, so $v$ and $v'$ both lie in the rectangle. On the one hand, from \Cref{obs:terminal}, $\card{v_a-v'_a}$ and $\card{v_e-v'_e}$ are both at most $\card{x-x'}$, which equals $\card{v_b-v'_b}$ and $\card{v_d-v'_d}$. On the other hand, $\card{v_c-v'_c} = \card{v_f-v'_f} = \card{(x+y)-(x'+y')}$. Therefore, $$(v,v')_{\TS} = \max\set{\card{x-x'},\card{(x+y)-(x'+y')}} = \frac{1}{2}\cdot\bigg(\card{(2x+y)-(2x'+y')} + \card{y-y'}\bigg).$$
\end{proof}

We are now ready to define the graph $G$ in the hard instance.

\paragraph{Constructing graph $G$.}
Let $L = 10^{3N}$. The vertex set of $G$ contains all points  $v\in \TS(D)$ whose associated vector $v = (v[x],v[y],v[z])$ satisfies that: $v[x]$ and $v[z]$ are integral multiples of $1/L$, and $v[y]$ is an integral multiple of $2/L$. 
We now define a collection of types of paths, and graph $G$ is simply the union of them.
Each path starts from some terminal (which we call its \emph{source}) and ends at another terminal (which we call its \emph{sink}), and consists of three parts: 
\begin{itemize}
\item an initial segment, which contains one edge connecting its source to some vice-source;
\item a main segment, which is a path connecting its vice-source to its vice-sink, such that all edges travel in the same direction (called the direction of the path), which is one of the four critical directions;
\item an ending segment, which contains one edge connecting its vice-sink to the sink.
\end{itemize}

\begin{table}[h]
\centering
	\begin{tabular}{| c | c | c | c | c | c | c | c |}
		\hline
		name $[i,j]$  & source  & vice-source ($\times \frac{1}{L}$) & direction & vice-sink ($\times \frac{1}{L}$) & sink &  capacity & length \\
		\hline
		$\mathsf{ad1}$ & $d$ & $(i,2j,0)$ & 1 & $(i,2j,L-j)$ & $a$ & 2 & 3 \\
		\hline
		$\mathsf{be1}$ & $b$ & $(i,2j,0)$ & 1 & $(i,2j,L-j)$ & $e$ & 2 & 3 \\
		\hline
		$\mathsf{ad2}$ & $a$ & $(i,0,j)$ & 2 & $(i,2j,0)$ & $d$ & 2 & 3\\
		\hline
		$\mathsf{be2}$ & $e$ &  $(i,0,j)$ & 2 & $(i,2j,0)$ & $b$ & 2 &3 \\
		\hline
		$\mathsf{ad3}$ & $a$ & $(0,2i,j)$ & 3 & $(L,2i,j)$ & $d$ & 1 & 3\\
		\hline
		$\mathsf{be3}$ & $b$ & $(0,2i,j)$ & 3 & $(L,2i,j)$ & $e$ & 1 & 3\\
		\hline
		$\mathsf{cf3}$ & $c$ & $(0,2i,j)$ & 3 & $(L,2i,j)$ & $f$ & 2 & 3\\
		\hline
		$\mathsf{ab}$ & $a$ & $(0,2i,j)$ & N/A & $(0,2i,j)$ & $b$ & 1 & 2\\
		\hline
		$\mathsf{de}$ & $e$ & $(0,2i,j)$ & N/A & $(0,2i,j)$ & $d$ & 1 & 2\\
		\hline
	\end{tabular}
\caption{Paths in direction $1,2,3$.}\label{table: path1}
\end{table}

The first group of paths are shown in \Cref{table: path1}.
As an example, the first row describes a collection that contains, for each pair $0 \le i,j \le L$, the path $(d, v_0, \dots, v_{L-j}, a)$ where $v_{s} = (i/L,2j/L, s/L)$ for any $0 \le s \le L-j$, and such path is named  $\mathsf{ad1}[i,j]$. 
Its length is $3$, and the weight of each of its edge is $2$.
See \Cref{fig: TS6_adbe1} for an illustration.
\Cref{table: path1} contains all paths in directions $1,2,3$.

\begin{figure}[h]
	\centering
	\scalebox{0.1}{\includegraphics{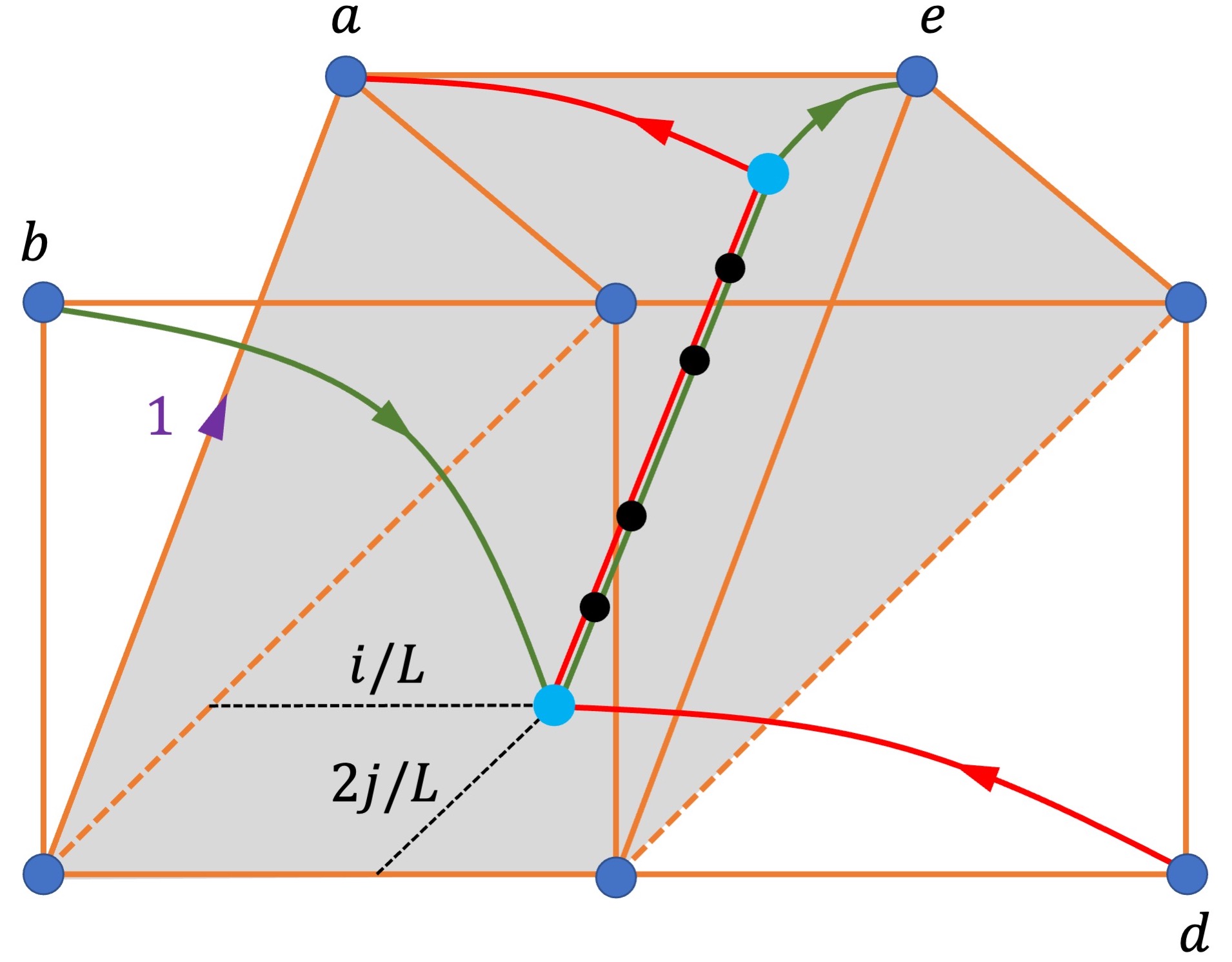}}
	\caption{Path $\mathsf{ad1}[i,j]$ (red), path $\mathsf{be1}[i,j]$ (green), and their shared vice-sink and vice-source (light blue).\label{fig: TS6_adbe1}}
\end{figure}

\begin{table}
	\begin{tabular}{| c | c | l | c | c | }
		\hline
		name  & src  & vice-source ($\times \frac{1}{L}$)  & vice-sink ($\times \frac{1}{L}$) & sink  \\
		\hline
		\multirow{3}{*}{$\mathsf{ad4}$} & $d$ & $(i,0,j)$ ($i+j\le L$) &  $(0,2i,j)$ & $a$ \\
		\cline{2-5}
		& $d$ & $(i,0,j)$ ($i+j\in [L,2L]$, $i,j\in [0,L]$) &  $(i+j-L,2(L-j),j)$ & $a$ \\
		\cline{2-5}
		& $d$ & $(L,2(i-L),j)$ ($i+j\in [L,2L]$, $j\in [0,L]$, $i\in [L,2L]$)  & $(i+j-L,2(L-j),j)$ & $a$ \\
		\hline
		\multirow{3}{*}{$\mathsf{be4}$} & $e$ & $(i,0,j)$ ($i+j\le L$) & $(0,2i,j)$ & $b$ \\
		\cline{2-5}
		& $e$ & $(i,0,j)$ ($i+j\in [L,2L]$, $i,j\in [0,L]$)  & $(i+j-L,2(L-j),j)$ & $b$ \\
		\cline{2-5}
		& $e$ & $(L,2(i-L),j)$ ($i+j\in [L,2L]$, $j\in [0,L]$, $i\in [L,2L]$)  & $(i+j-L,2(L-j),j)$ & $b$ \\
		\hline
		\multirow{3}{*}{$\mathsf{cf4}$} & $c$ & $(i,0,j)$ ($i+j\le L$)  & $(0,2i,j)$ & $f$ \\
		\cline{2-5}
		& $c$ & $(i,0,j)$ ($i+j\in [L,2L]$, $i,j\in [0,L]$) &  $(i+j-L,2(L-j),j)$ & $f$ \\
		\cline{2-5}
		& $c$ & $(L,2(i-L),j)$ ($i+j\in [L,2L]$, $j\in [0,L]$, $i\in [L,2L]$) &  $(i+j-L,2(L-j),j)$ & $f$ \\
		\hline
	\end{tabular}\caption{Paths in direction $4$.}\label{table: path2}
\end{table}

The next group of paths, all in direction $4$ with weight $1$ and length $3$, are shown in \Cref{table: path2}.
For example, the first row describes a type of paths that contains, for each pair $i,j$ with $0 \le i+j \le L$, a path $(d, v_0, \dots, v_i, a)$ where $v_{s} = ((i-s)/L,2s/L,j/L)$ for any $0 \le s \le L$, and such path is named  $\mathsf{ad4}[i,j]$.  See \Cref{fig: TS6_4} for an illustration.

\begin{figure}[h]
	\centering
	\scalebox{0.1}{\includegraphics{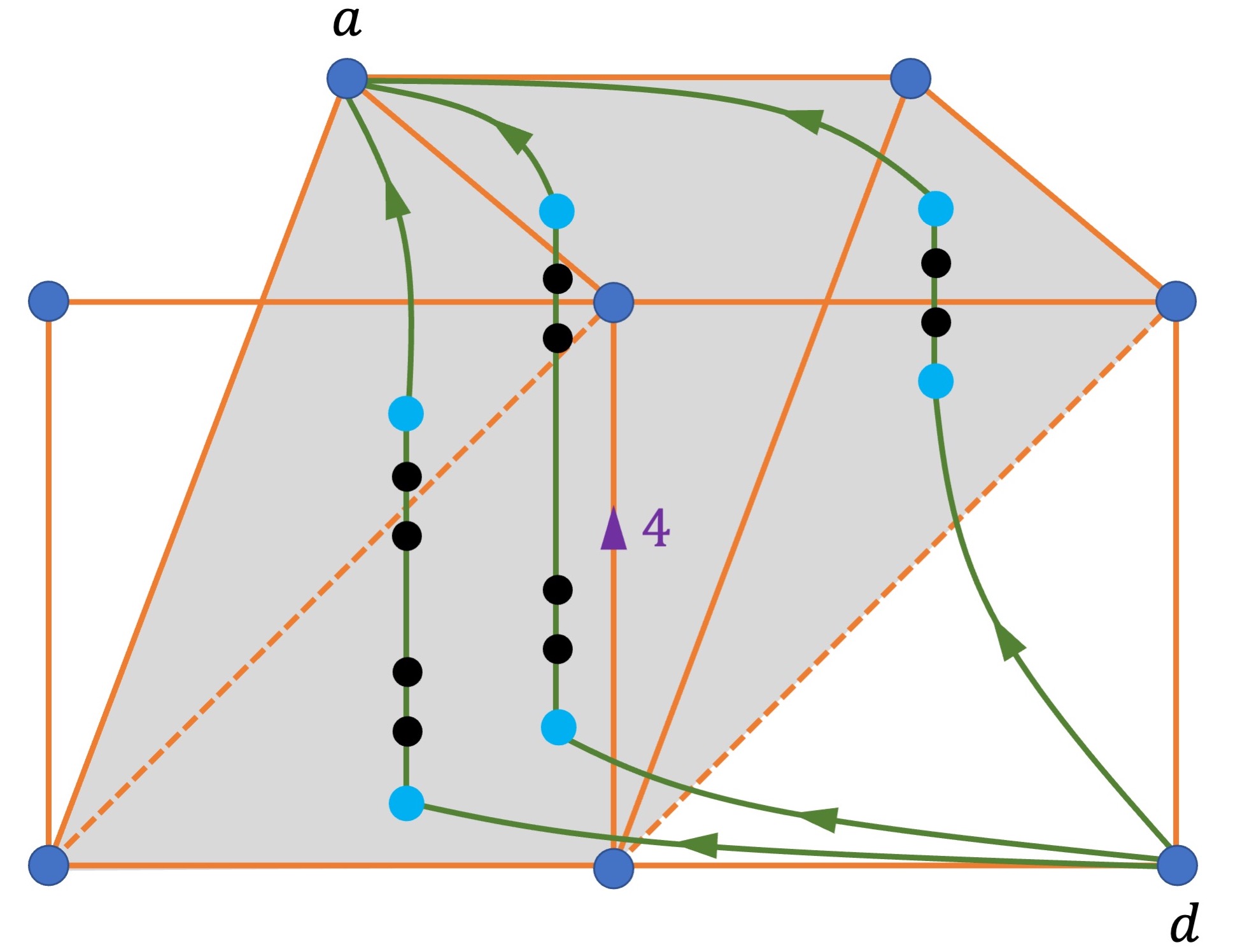}}
	\caption{One path of each kind of type $ad4$.\label{fig: TS6_4}}
\end{figure}


This completes the construction of graph $G$. From now on we will analyze this instance $(G,T,\ell)$ of $\zesn$ and prove \Cref{thm: lower 6} for it.
We use the following immediate property of graph $G$.

\begin{observation}
\label{obs: path net}
For every vertex $v$, for every critical direction, $v$ is incident to one (if it lies on the boundary) or two edges in this direction, and there is a path containing $v$ and its incident edges.
\end{observation}

\subsection{Analysis step 1. fine-grained analysis of $\vol(\fset,\delta)-\opt$}
\label{subsec: basic}

We first show that $\opt  :=\sum_{(u,v)\in E(G)}c(u,v)\cdot \dist_{\ell}(u,v)=O(L^2)$. 
Denote by $\pset$ the collection of all paths in $G$ defined above. 
For each path $P\in \pset$ connecting a pair $t_P,t'_P$ of terminals in $T$, by triangle inequality, its contributed cost is $\sum_{(u,v)\in E(P)}c(u,v)\cdot\dist_{\ell}(u,v)= c(u,v)\cdot\dist_{\ell}(t_P,t'_P)=c(u,v)\cdot D(t_P,t'_P)$.
Therefore,
\[
%
\opt  =\sum_{(u,v)\in E(G)}c(u,v)\cdot\dist_{\ell}(u,v)= \sum_{P\in \pset}c(u,v)\cdot D(t_P,t'_P),
\]
%
which implies that $\opt\le 90L^2$, as (i) there are $L^2$ pairs $(i,j)$; and (ii) for each pair $(i,j)$ there are at most $15$ paths with capacity at most $2$ and length at most $3$ each. 
Therefore, in order to prove \Cref{thm: lower 6} for the $\zesn$ instance, it suffices to show that, for any solution $(\fset,\delta)$ with $|\fset|\le N$, $\vol(\fset,\delta)-\opt\ge 10^{-15}\cdot L^2$.



First, from \Cref{lem:proj_dis}, we can assume that the solution $(\fset,\delta)$ satisfies that
\begin{itemize}
\item each set $F\in \fset$ corresponds to a point $f_F$ in $\TS(D)$; and
\item for each pair $F,F'\in \fset$, $\delta(F,F')=(f_{F},f_{F'})_{\TS}$.
\end{itemize}
For each vertex $v\in V(G)$, we denote by $f(v)$ the point in $\TS(D)$ that corresponds to the set in $\fset$ that contains $v$. For each path $P$, we $\vol(P)=\sum_{(u,v)\in E(P)}(f(u),f(v))_{\TS}$.


Note that $\vol(\fset,\delta)-\opt=\sum_{P\in \pset}\vol(P)-D(t_P,t'_P)$. We define the \emph{loss of path $P$} as $\ell(P)=\vol(P)-D(t_P,t'_P)$, and will aim to give a lower bound the sum of loss of all paths.
As $f(t)=t$ for all terminals, from triangle inequality, the loss of every path is non-negative. Moreover, for a path $P=(t,v_1,\ldots,v_k,t')$, we can further decompose $\ell(P)$ as (denoting $t=v_0$ and $t'=v_{k+1}$)
\[
\begin{split}
\ell(P)=\vol(P)-D(t,t') & =\bigg(\sum_{0\le j\le k}(f(v_j),f(v_{j+1}))_\TS\bigg)-D(t,t')\\
& =\sum_{0\le j\le k}\bigg((f(v_j),f(v_{j+1}))_\TS+(t,f(v_j))_\TS-(t,f(v_{j+1}))_\TS\bigg),
\end{split}
\]
where for each $j$, by triangle inequality, $(f(v_j),f(v_{j+1}))_\TS+(t,f(v_j))_\TS-(t,f(v_{j+1}))_\TS\ge 0$.

From \Cref{obs: path net}, for each vertex $v\in V(G)$ and each of the $4$ critical directions, there is some path in $\pset$ going through $v$ in this direction, creating a term in the above inequality. In order to best utilize the structure of $\TS(D)$, our plan is to estimate the $4$ terms related to each vertex together.
We first define vectors corresponding an edge in each of the four critical directions:
\[w_1 = (0,0,1/L), w_2 = (0,2/L,-1/L), w_3=(1/L,0,0), w_4=(-1/L,2/L,0).\]
For every $v\in V(G)$, every terminal $t$, and each $1 \le i \le 4$, we define 
\[\ell_i(v,t) = (f(v),f(v+w_i))_{\TS} + (f(v),t)_{\TS} - (f(v+w_i),t)_{\TS}.\] 
and
\[\ell'_i(v,t) = (f(v),f(v+w_i))_{\TS} - (f(v),t)_{\TS} + (f(v+w_i),t)_{\TS}.\] 
We use the following claim, that decompose the loss of paths into the sum of $\ell_i(v,t)$ terms.

\begin{lemma} \label{lem:cost} The following inequalities hold.
    \begin{itemize}
        \item $\sum_{i,j} \bigg(\ell(\mathsf{da1}[i,j])+\ell(\mathsf{be1}[i,j])\bigg) \ge \sum_v \bigg( \ell_1(v,d) +  \ell_1(v,b)\bigg)$;
        \item $\sum_{i,j} \bigg(\ell(\mathsf{ad2}[i,j])+\ell(\mathsf{eb2}[i,j])\bigg) \ge \sum_v \bigg(\ell'_2(v,d) +  \ell'_2(v,b)\bigg)$;
        \item $\sum_{i,j} \bigg(\ell(\mathsf{ad3}[i,j])+\ell(\mathsf{be3}[i,j])+2\ell(\mathsf{cf3}[i,j])\bigg) \ge \sum_v \bigg(\ell'_3(v,d)+\ell_3(v,b)+2\ell_3(v,c)\bigg)$; and
        \item $\sum_{i,j} \bigg(\ell(\mathsf{da4}[i,j])+\ell(\mathsf{eb4}[i,j])+2\ell(\mathsf{cf4}[i,j])\bigg) \ge \sum_v \bigg(\ell_4(v,d)+\ell'_4(v,b)+2\ell_4(v,c)\bigg)$.
    \end{itemize}
\end{lemma}

\begin{proof}
We only prove the first inequality, and the other three inequalities can be proved similarly. Consider the path $\mathsf{ad1}[i,j] = (d,v_0,\dots,v_{L-j},a)$. Note that 
\[\ell(\mathsf{ad1}[i,j]) \ge \sum_{s=0}^{L-j-1} \bigg((f(v_{s}),f(v_{s+1}))_\TS + (d,f(v_{s}))_\TS-(d,f(v_{s+1}))_\TS \bigg)\ge  \sum_{s=0}^{L-j-1} \ell_1(v_{s},d). \]
As every vertex $v\in V(G)$ appears in some path $\mathsf{da1}[i,j]$, $\sum_{i,j} \ell(\mathsf{da1}[i,j]) \ge \sum_v \ell_1(v,d)$ (by \Cref{obs: path net}). Similarly, $\sum_{i,j} \ell(\mathsf{be1}[i,j]) \ge \sum_v \ell_1(v,b)$, thus the first inequality is true.
\end{proof}


\begin{claim} \label{clm:x}
    For every vertex $v\in V(G)$, 
    \begin{itemize}
        \item $\ell_1(v,d)+\ell_1(v,b) \ge 2\card{f(v)[x]-f(v+w_1)[x]}$; 
        \item $\ell'_2(v,d)+\ell'_2(v,b) \ge 2\card{f(v)[x]-f(v+w_2)[x]}$;
        \item $(a,f(v))_{\TS}+(b,f(v))_\TS \ge 2 + 2\big(f(v)[x]\big)^+$; 
        \item $(d,f(v))_\TS+(e,f(v))_\TS \ge 2 + 2\big(1-f(v)[x]\big)^+$. 
    \end{itemize}
\end{claim}

\begin{proof}
Let $(x,y,z)$ be the associated vector of $f(v)$ and let $(x_1,y_1,z_1)$ be that of $f(v+w_1)$. 
From \Cref{obs: dist in TS}, $(f(v),f(v+w_1))_\TS \ge \card{x-x_1}+\card{z-z_1}$. From \Cref{obs:terminal}, $(f(v),b)_\TS+(f(v),d)_\TS=3+2z$ and $(f(v+w_1),b)_\TS+(f(v+w_1),d)_\TS=3+2z_1$. Therefore,
\[
  \begin{split}
        \ell_1(v,d)+\ell_1(v,b) & = \bigg((f(v),f(v+w_1))_\TS+(d,f(v))_\TS-(d,f(v+w_1))_\TS\bigg)\\
      & \quad\quad  +
\bigg((f(v),f(v+w_1))_\TS+(b,f(v))_\TS-(b,f(v+w_1))_\TS\bigg) \\ 
        & \ge 2\card{x-x_1} + 2\card{z-z_1} + 3 + 2z - 3 - 2z_1 \ge 2\card{x-x_1}.
  \end{split}
\]
The second inequality can be proved similarly.
For the third inequality, by \Cref{obs:terminal}, we get that $(a,f(v))_\TS+(b,f(v))_\TS = 2 + x + \card{x} = 2 + 2\max\{x,0\}$, and the last one can be proved similarly.
\end{proof}

\subsection{Analysis step 2. reduction to a geometric problem in $\mathbb{R}^2$}
\label{subsec: onto plane}

For each $v\in \TS(D)$, let $p(v)$ be the projection of $f(v)$ into the tight span of $D'$, which is the metric on $b,c,d,f$ induced by $D$, which is exactly the $2$-dimensional set (rectangle) $\overline{bcdf}$. 

\begin{observation}
\label{obs: proj_ down}
Let $(x,y,z)$ be the associated vector of $v$. Then $p(v) = f(v) + (z/2)\cdot(w_2-w_1)$, and so the associated vectors of $f(v)$ and $p(v)$ have the same $x$ coordinates.
\end{observation}
\begin{proof}
Since $v_c+v_f = 3$ holds for all $v\in \TS(D)$, 
when we project $v$ onto the rectangle, we only decrease $v_b$ and $v_d$. By \Cref{obs:terminal}, $v_b+v_d = 3+2z$, $v_b+v_c=2x+y+1+z \ge 1+z$ and $v_b+v_f = 4-y \ge 2+z$. Therefore, $\Delta_b=z$. Similarly, $\Delta_d$ is also $z$. So after projection, both $v_b$ and $v_d$ are decreased by $z$ , both $v_a$ and $v_e$ are increased by $z$ (note that $v_a+v_d = v_b+v_e =3$), and both $v_c$ and $v_f$ do not change. Thus, $x = v_a+v_b-2$ does not change, $y=v_c-v_b+1$ is increased by $z$ and $z$ becomes $0$. Thus, $p(v)=f(v)+z \cdot (0,1,-1) = f(v) + (z/2) \cdot (w_2-w_1)$.
\end{proof}

Notice that $\TS(D')\subsetneq \TS(D)$, and for any $u,v\in \TS(D')$, $(u,v)_{\TS(D')}=(u,v)_{\TS(D)}$ holds, so we still denote $\TS=\TS(D)$.
From \Cref{lem:proj_dis}, for every pair $u,v\in TS(D)$, $(p(u),p(v))_{\TS} \le (f(u),f(v))_{\TS}$. 


\begin{claim} \label{clm:proj_abc}
    For each point $v\in \TS(D)$, 
    \begin{itemize}
            \item $(c,f(v))_{\TS}=(c,p(v))_{\TS}$;
            \item $(b,f(v))_{\TS} - (d,f(v))_{\TS} = (b,p(v))_{\TS} - (d,p(v))_{\TS}$. 
    \end{itemize}
\end{claim}

\begin{proof}
    If $f(v)\in \TS(D')$, then $f(v)=p(v)$ and the inequalities clearly hold. Otherwise, $f(v)$ lies in the $3$-dimensional set (the triangular prism). We denote by $(x,y,z)$ the associated vector of $f(v)$. By \Cref{obs:terminal}, $(c,f(v))_{\TS} = x+y+z$. From \Cref{obs: proj_ down}, $p(v) = f(v) + (z/2)\cdot(w_2-w_1) = (x,y+z,0)$, so $(c,p(v))_\TS = x+y+z = (c,f(v))_\TS$. 
From \Cref{obs:terminal}, $(d,f(v))_\TS = 2-x+z$ and $(b,f(v))_\TS = 1+x+z$, so $(b,f(v))_\TS - (d,f(v))_\TS = 2x-1$, which equals $(b,p(v))_\TS-(d,p(v))_\TS$ as the associated vectors of $f(v)$ and $p(v)$ have the same $x$ coordinates (from \Cref{obs: proj_ down}).
\end{proof}

For each $v\in V(G)$, we define
    \begin{itemize}[leftmargin=*]
        \item $\ell_x(v) = \big|p(v)[y]-p(v+w_3)[y]\big| + 2\cdot\bigg((2p(v)[x]+p(v)[y])-(2p(v+w_3)[x]+p(v+w_3)[y])\bigg)^+$;
        \item $\ell_y(v) = 2 \big|p(v)[x]-p(v+w_3+w_4)[x]\big|$;
        \item $\ell^1_z(v) = \big|(2p(v)[x]+p(v)[y])-(2p(v+w_4)[x]+p(v+w_4)[y])\big|$; 
        \item $\ell^2_z(v) = \big|(2p(v)[x]-p(v)[y])-(2p(v + 2w_3 + w_4)[x]-p(v + 2w_3 + w_4)[y])\big|$.
    \end{itemize}

We prove the following two claims:

\begin{claim} \label{clm:Cx}
$\ell_x(v) \ge 2\big(p(v)[x]-p(v+w_3)[x]\big)^+$.
\end{claim}
\begin{proof}
   If $p(v)[x]-p(v+w_3)[x] \le 0$, then clearly the claim holds. 
Otherwise,
    \begin{align*}
        \ell_x(v) & \ge \card{p(v)[y]-p(v+w_3)[y]} + \bigg((2p(v)[x]+p(v)[y])-(2p(v+w_3)[x]+p(v+w_3)[y])\bigg)^+ \\
               & \ge 2(p(v)[x]-p(v+w_3)[x]).
    \end{align*}
\end{proof}

\begin{claim} \label{clm:transfer}
$\ell^2_z(v) \le \ell_x(v)+\ell_x(v+w_3+w_4)+\ell_y(v)+\ell_y(v+w_3)+\ell^1_z(v+w_3)$.
\end{claim}
\begin{proof}
Define associated vectors as $p(v)=(x,y,0)$, $p(v+w_3) = (x_1,y_1,0)$, $p(v+w_3+w_4) = (x_2,y_2,0)$ and $p(v+2w_3+w_4) = (x_3,y_3,0)$. By definition,
\[    \begin{split}
          &\ell_x(v)+\ell_x(v+w_3+w_4)+\ell_y(v)+\ell_y(v+w_3)+\ell^1_z(v+w_3)\\
        = &\card{y-y_1} + \card{y_2-y_3} + 2\card{x-x_2} + 2\card{x_1-x_3} + \card{(2x_1+y_1)-(2x_2+y_2)} \\
        \ge & \card{2x-y-2x_3+y_3} = \ell^2_z(v).
    \end{split}\]
\end{proof}

The following lemma lower-bounds $\vol(\fset,\delta)-\opt$ in terms of the loss functions we defined.

\begin{lemma} \label{lem:2d-cost}

\[\vol(\fset,\delta)-\opt\ge\frac{2}{3}\cdot \sum_v (\ell_x(v)+\ell_y(v)+\ell^1_z(v)+\ell^2_z(v)) + \sum_{v:v[x]=0} 2\big(p(v)[x])^+ + \sum_{v:v[x]=1} 2\big(1-p(v)[x]\big)^+.\]
\end{lemma}
\begin{proof}
We will show that
\[\vol(\fset,\delta)-\opt\ge \sum_v (2\ell_x(v)+2\ell_y(v)+2\ell^1_z(v)) + \sum_{v:v[x]=0} 2\big(p(v)[x]\big)^+ + \sum_{v:v[x]=1} 2\big(1-p(v)[x]\big)^+.\]
Note that this (together with \Cref{clm:transfer}) immediately implies the lemma.

From \Cref{obs: proj_ down}, the associated vectors for $p(v)$ and $f(v)$ have the same $x$ coordinates. Therefore, we can replace all $f(v)$ by $p(v)$ in the of \Cref{clm:x}, obtaining 
    \begin{itemize}
        \item $\sum_{i,j} \ell(\mathsf{ab}[i,j]) = \sum_{v:v[x]=0} 2\big(p(v)[x]\big)^+$;
        \item $\sum_{i,j} \ell(\mathsf{de}[i,j]) = \sum_{v:v[x]=1} 2\big(1-p(v)[x]\big)^+$;
        \item $\ell_1(v,d)+\ell_1(v,b) \ge 2\card{p(v)[x]-p(v+w_1)[x]}$;
        \item $\ell'_2(v,d)+\ell'_2(v,b) \ge 2\card{p(v)[x]-p(v+w_2)[x]}$.
    \end{itemize} 
    As $w_1+w_2=w_3+w_4$, from the third and the fourth inequalities above,
    \[
        \sum_v (\ell_1(v,d)+\ell_1(v,b) + \ell'_2(v,d)+\ell'_2(v,b)) \ge \sum_v 2\card{p(v)[x]-p(v+w_1+w_2)[x]} = \sum_v \ell_y(v)\]
    For any $v$, we define the associated vectors as $p(v)=(x,y,0)$, $p(v+w_3) = (x_3,y_3,0)$ and $p(v+w_4) = (x_4,y_4,0)$. From \Cref{obs: dist in TS}, $2(p(v),p(v+w_3))_{\TS} = \card{(2x+y)-(2x_3+y_3)} + \card{y-y_3}$. Thus by \Cref{clm:proj_abc}, \Cref{obs:terminal} and the property that $(f(v),f(v+w_3))_{\TS} \ge (p(v),p(v+w_3))_{\TS}$, 
    \begin{align*}
            \ell'_3(v,d)+\ell_3(v,b)+2\ell_3(v,c) 
        & =   4(f(v),f(v+w_3))_\TS - (f(v),d)_\TS + (f(v),b)_\TS + 2(f(v),c)_\TS \\
            & \quad\quad + (f(v+w_3),d)_\TS - (f(v+w_3),b)_\TS - 2 (f(v+w_3),c)_\TS \\
       & \ge  4(p(v),p(v+w_3))_\TS - (p(v),d)_\TS + (p(v),b)_\TS + 2(p(v),c)_\TS \\
            & \quad\quad + (p(v+w_3),d)_\TS - (p(v+w_3),b)_\TS - 2 (p(v+w_3),c)_\TS \\
        & \ge  2(\card{(2x+y)-(2x_3+y_3)} + \card{y-y_3}) - (2-x) + x + 1 + 2x + 2y \\
            & \quad\quad + (2-x_3) -x_3 - 1 - 2x_3 -2y_3 \\
        & =    2\card{y-y_3} + 2(\card{(2x+y)-(2x_3+y_3)} - ((2x+y)-(2x_3+y_3))) \\
        & \ge  2\card{y-y_3} + 4\bigg(2(x-x_3)+(y-y_3)\bigg)^+ = 2\ell_x(v).
    \end{align*}
    Similarly, we can also show that $\ell_4(v,d)+\ell'_4(v,b)+2\ell_4(v,c) \ge 2\ell_z^1(v)$. The lemma now follows from \Cref{lem:cost} and the fact that $\vol(\fset,\delta)-\opt=\sum_{P\in \pset}\ell(P)$. 
\end{proof}

In the rest of the section, we will prove that the RHS in \Cref{lem:2d-cost} is $\Omega(1)$. For any vertex $v$, define $\ell(v) = \ell_x(v)+\ell_y(v)+\ell^1_z(v)+\ell^2_z(v)$ and for any $0 \le i+j \le L$, define 
$$\ell[i,j] = \bigg(\sum_{q=0}^L \ell((q/L,i/L,j/L))\bigg)+3\big(p((0,i/L,j/L))[x]\big)^++3\big(1-p((1,i/L,j/L))[x]\big)^+.$$
By \Cref{lem:2d-cost}, $\vol(\fset,\delta)-\opt\ge 2/3\cdot\sum_{i,j}  \ell[i,j]$. Let $D=3 \cdot 7^{3N}$ and $\eps=0.001$. The following lemma is crucial for completing the proof of \Cref{thm: lower 6}. Its proof is deferred to \Cref{subsec: plane}.

\begin{lemma} \label{lem:line}
    For any $i,j \le L-D$, there exists some $k \le 6D$ such that $\sum_{s=0}^{k-1} \ell[i+2s,j] \ge \eps k$.
\end{lemma}

We now use \Cref{lem:line} to complete the proof of \Cref{thm: lower 6}.
We will show that, for any $0 \le j \le L-D$, $\sum_{i=0}^{L-j} \ell[i,j] \ge \eps(L-D-j)$. Note that implies that
$\sum_{i,j} \ell[i,j] \ge \sum_{j=0}^{L-D} \eps(L-D-j) \ge \eps L^2/3$.

We define a sequence of pairs $(i_0,k_0),(i_1,k_1),\ldots$ as follows: let $i_0 = 0$, for any $q$, let $k_{q}$ be the $k$ given by \Cref{lem:line} for $(i_{q},j)$. If $i_{q}+k_{q}+j \le L-D$, then let $i_{q+1}= i_{q}+k_{q}$, otherwise the sequence is completed. By \Cref{lem:line}, for any $q$, $\sum_{i=i_{q}}^{i_{q}+k_{q}-1} \ell[i,j] \ge \eps k_{q}$. Suppose $i^*$ is the last index, we have $i_{i^*}+k_{i^*}+j > L-D$. Therefore, $\sum_{i=0}^{L-j} \ell[i,j] \ge \eps (i_{i^*}+k_{i^*}) \ge \eps(L-D-j)$.

\subsection{Analysis step 3. handling the problem in $\mathbb{R}^2$: proof of \Cref{lem:line}}
\label{subsec: plane}

Fix a pair $i,j\le L-D$. We denote $v_{q} = (q/L,i/L,j/L)$ for each $0 \le q \le L$. 

\begin{claim} \label{clm:triangle}
    For any $r,\Delta \ge 0$,
    \begin{align*}
        & \sum_{q=0}^{2\Delta-1} \ell_x(v_{r+q}) + \sum_{q=0}^{\Delta-1} \bigg(\ell_z^2(v_{r}+q(2w_3+w_4)) + \ell_y(v_{r+\Delta} +q(w_3+w_4)) + \ell_z^1(v_{r+2\Delta}+qw_4)\bigg) \\ 
        \ge & \card{2p(v_{r+\Delta})[x]-p(v_{r})[x]-p(v_{r+2\Delta})[x]}.
    \end{align*}
\end{claim}

\begin{proof}
Define associated vectors as $p(v_{r})=(x_1,y_1,0)$, $p(v_{r+\Delta}) = (x_2,y_2,0)$, $p(v_{r+2\Delta}) = (x_3,y_3,0)$ and $p(v_{r}+\Delta(2w_3+w_4)) = (x_4,y_4,0)$. 
By definition and triangle inequality,
\[\sum_{q=0}^{2\Delta-1} \ell_x(v_{r+q}) \ge \sum_{q=0}^{2\Delta-1} \card{p(v_{r+q})[y]-p(v_{r+q+1})[y]} \ge \card{y_1-y_3}.\]
    and similarly, we can show that \[\sum_{q=0}^{\Delta-1} \ell_z^2(v_r+q(2w_3+w_4)) = \card{(2x_1+y_1)-(2x_4+y_4)},$$ $$\sum_{q=0}^{\Delta-1} \ell_y(v_{r+\Delta} +q(w_3+w_4)) \ge 2\card{x_2-x_4}, \text{and}$$ $$\sum_{q=0}^{\Delta-1} \ell_z^1(v_{r+2\Delta}+qw_4)) \ge \card{(2x_3-y_3)-(2x_4-y_4)}.\] 
Therefore, for the inequality in the statement of \Cref{clm:triangle},
\[
\begin{split}
2\cdot \textnormal{LHS} & \ge 2\card{y_1-y_3} + 2\card{(2x_1+y_1)-(2x_4+y_4)} + 4\card{x_2-x_4} + 2\card{(2x_3-y_3)-(2x_4-y_4)}\\
& \ge \card{y_1-y_3} + \card{-(2x_1+y_1)+(2x_4+y_4)} + 4\card{x_2-x_4} + \card{-(2x_3-y_3)+(2x_4-y_4)}\\
 & \ge \card{4x_2-2x_1-2x_3}= 2\cdot \textnormal{RHS}.\\
\end{split}
\]
This finishes the proof of \Cref{clm:triangle}.
\end{proof}

By \Cref{clm:Cx}, for any $v_r$, $\ell_x(v_r) \ge \big(p(v_r)[x]-p(v_{r+1})[x]\big)^+$, so if $\sum_{\ell=0}^{L-1} \big(p(v_r)[x]-p(v_{r+1})[x]\big)^+ > \eps$, then we can set $k=1$ and conclude the proof. Therefore, we assume from now on that it is at most $\eps$, and it follows that, for every pair $r_1<r_2$, $p(v_{r_1})[x] - p(v_{r_2})[x] \le \eps$.
Similarly, we can assume that $p(v_0)[x] \le \eps$ and $p(v_L)[x] \ge 1-\eps$.

We distinguish between the following two cases.

\subsubsection*{Case 1. $p(v_{L-D})[x]-p(v_D)[x] \le 0.1$.} 
In this case, either $p(v_D)[x] \ge 0.45$ or $p(v_{L-D})[x] \le 0.55$ holds. We assume without lose of generality that $p(v_D)[x] \ge 0.45$. 

\paragraph{Case 1.1. $p(v_{D/2})[x]\le 0.3$.} In this case, consider any $0 \le r \le D/2$ and $\Delta=D$. Note that 
\[p(v_r)[x] \le 0.3+\eps, \quad p(v_{r+\Delta})[x] \ge p(v_D)[x] - \eps, \quad p(v_{r+2\Delta})[x] \le p(v_{L-D})[x] + \eps,\] so $\card{2p(v_{r+\Delta})[x]-p(v_r)[x]-p(v_{r+2\Delta})[x]}>0.05-4\eps \ge 10 \eps$. From \Cref{clm:triangle},
    \begin{align*}
        & \sum_{q=0}^{2\Delta-1} \ell_x(v_{r+q}) + \sum_{q=0}^{\Delta-1} \bigg(\ell_z^2(v_r+q(2w_3+w_4)) + \ell_y(v_{r+\Delta} +q(w_3+w_4)) + \ell_z^1(v_{r+2\Delta}+qw_4)\bigg) \\ 
        \ge & \card{2p(v_{r+\Delta})[x]-p(v_r)[x]-p(v_{r+2\Delta})[x]} \ge 10 \eps.
    \end{align*}
If for any such $r$, $\sum_{q=0}^{2\Delta-1} \ell_x(v_{r+q}) \ge \eps$ holds, we set $k=1$ and conclude the proof. Otherwise,
    \begin{align*}
        \sum_{r=0}^{D/2} \sum_{q=0}^{\Delta-1} \bigg(\ell_z^2(v_r+q(2w_3+w_4)) + \ell_y(v_{r+\Delta} +q(w_3+w_4)) + \ell_z^1(v_{r+2\Delta}+qw_4)\bigg) \ge 5\eps D,
    \end{align*}
which means $\sum_{s=0}^D \ell[i+2s,j] \ge 5\eps D$ and the statement of \Cref{lem:line} follows from setting $k=D$.

\paragraph{Case 1.2. $p(v_{D/2})[x]>0.3$.} For each $0 \le s \le 6D$, we say $s$ is \emph{good} if $\ell[i+2s,j] \le 20 \eps = 0.02$, otherwise we say $s$ is \emph{bad}. For any good $s$, if $p(v_D+s(w_3+w_4))[x] < 0.25$, then for any $D/2 \le r \le D$, $p(v_r)[x] - p(v_r+s(w_3+w_4))[x] > 0.05-0.02-\eps >20 \eps$, which means
$$
    \sum_{q=0}^{s-1} \ell_y(v_r+q(w_3+w_4)) > 20 \eps.
$$
This implies that
$\sum_{q=0}^{s} \ell[i+2q,j] > 10D\eps$, or equivalently, \Cref{lem:line} is true when $k=s$.

We now assume that for every good $s$, $p(v_D+s(w_3+w_4))[x] \ge 0.25$. For any $0 \le s < D$, we say $s$ is \emph{perfect} if $s,s+D,\dots,s+5D$ are all good. For any perfect $s$, we define associated vectors as 
\begin{itemize}
\item $p(v_0+s(w_3+w_4))=u_0=(x_0,y_0,0)$;
\item  $p(v_D+(s+D)(w_3+w_4))=u_1=(x_1,y_1,0)$;
\item $p(v_0+(s+2D)(w_3+w_4))=u_2=(x_2,y_2,0)$;
\item 
$p(v_D+(s+3D)(w_3+w_4))=u_3=(x_3,y_3,0)$;
\item
$p(v_0+(s+4D)(w_3+w_4))=u_4=(x_4,y_4,0)$; and
\item $p(v_D+(s+5D)(w_3+w_4))=u_5=(x_5,y_5,0)$. 
\end{itemize}
By definition, $x_0,x_2,x_4 \le 0.02$ and $x_1,x_3,x_5 \ge 0.25$. Therefore, by triangle inequality,
\begin{align*}
& \sum_{q=0}^{\Delta-1}  \bigg(\ell_z^2(u_0+q(2w_3+w_4)) +  \ell_z^1(u_1+q w_4) + \ell_z^2(u_2+q(2w_3+w_4))+ \ell_z^1(u_3+q w_4) + \ell_z^2(u_4+q(2w_3+w_4))\bigg)\\
    & \ge \card{(2x_0-y_0)-(2x_1-y_1)}+\card{(2x_1+y_1)-(2x_2+y_2)}+\card{(2x_2-y_2)-(2x_3-y_3)}\\
 & \quad + \card{(2x_3+y_3)-(2x_4+y_4)}+\card{(2x_4-y_4)-(2x_5-y_5)} \ge -0.46\times 5 +y_5-y_0 \ge 0.3 \ge 30 \eps,
\end{align*}
where the last inequatlity is since $0 \le y_0,y_5 \le 2$.

Assume there are $m$ perfect indices $s$, then $\sum_{s=0}^{6D} \ell[i+2s,j] \ge 30m\eps$. If $m>D/2$ then \Cref{lem:line} is true when $k=6D$; if $m\le D/2$, then there are at least $D/2$ bad indices $s$, which means that $\sum_{s=0}^{6D} \ell[i+2s,j] \ge 10D\eps$, and \Cref{lem:line} is still true when $k=6D$.

\subsubsection*{Case 2. $p(v_{L-D})[x]-p(v_D)[x] > 0.1$.}
Since the solution size is at most $N$, for all $v\in V(G)$, $p(v)[x]$ takes at most $N$ different values. Let $x_1,\dots,x_N$ be these values sorted in the increasing order. We say that a consecutive subsequence $(x_r,\ldots,x_q)$ is a \emph{group}, iff 
\begin{itemize}
\item for each $r\le s\le q-1$, $x_{r+1}-x_r \le \eps/N$; and 
\item $x_{r}-x_{r-1}> \eps/N$; and $x_{q+1}-x_{q}> \eps/N$.
\end{itemize} 
Clearly, $(x_1,\dots,x_N)$ can be partitioned into consecutive subsequences that form groups. We use the following observation.
\begin{observation} \label{obs:small}
    $\sum_{r: \text{ }x_r \textnormal{ and } x_{r+1} \textnormal{ are in the same group}} (x_{r+1}-x_r) \le \eps$.
\end{observation}

\begin{claim} \label{prop:math1}
For any set $A$ of positive integers  such that $\sum_{a\in A} = m$, $\prod_{a\in A} (7a)\le7^m$.
\end{claim}
\begin{proof}
We prove by induction on $m$. When $m=1$, $\prod_{a\in A} (7a)=7$. Assume the claim is true for $m \le k-1$, consider the case $m=k$. Take any element $a'\in A$, then $\prod_{a\in A} (7a) \le 7a' \cdot 7^{m-a'} \le 7^m$.
\end{proof}


\begin{claim} \label{clm:partition}
There exists an integer $k\le 7^{3N}$, such that the interval $[D,L-D]$ can be partitioned into subintervals, $[a_1,a_2],[a_2,a_3],\ldots,[a_{t-1},a_t]$, such that
    \begin{itemize}
    	\item each interval $[a_i,a_{i+1}]$ is either \emph{smooth} or \emph{steep}, and smooth and steep intervals appear interchangeably (that is, if $[a_i,a_{i+1}]$ is smooth, then $[a_{i+1},a_{i+2}]$ is steep, and if $[a_i,a_{i+1}]$ is steep, then $[a_{i+1},a_{i+2}]$ is smooth);
        \item the length $(a_{i+1}-a_i)$ of a smooth segment is at least $5k$, while the length of a steep segment is at most $k$;
        \item within a smooth segment $[a_i,a_{i+1}]$, all vertices $v_r$ with $a_i\le r\le a_{i+1}$ have their value $p(v)[x]$ lying in the same group.
    \end{itemize}
\end{claim}

\begin{proof}
    We start with $k=1$. We make all $[r,r+1]$ as steep segments where $p(v_r)[x]$ and $p(v_{r+1})[x]$ is not in the same group, and all intervals between steep segments are smooth segments. The only condition that is possibly not satisfied is that some smooth segments may have length less than $5k$. Consider all steep segments, for any $r$ such that $p(v_r)[x]$ and $p(v_{r+1})[x]$ are not in the same group and $p(v_r)[x]>p(v_{r+1})[x]$, we have $\ell_x(v_r)\ge\eps/N$ by \Cref{clm:Cx}. Therefore, the number of such indices $r$ is at most $N$ as otherwise $\sum_{r=0}^{L-1} \big(p(v_r)[x]-p(v_{r+1})[x]\big)^+ > \eps$. Since there are at most $N$ groups of $x$ values, the number of steep segments is at most $3N$.

    Now we repeat the following process: whenever there is a smooth segments is less than $5k$, we merge all steep segments such that the smooth segments between them is less than $5k$. After merging, we make $k$ be the maximum length of steep segments.

    When the process terminates, $k$ is the number that satisfies all the condition, so we only need to prove that in the end $k\le 7^{3N}$. In each step, suppose any new steep segments contains at most $\alpha$ original steep segments, then $k$ can be increased by at most $6(\alpha-1)+1$ factor. On the other hand, the number of steep segments  decreases by at least $\alpha-1$. Let $\alpha_1, \alpha_2, \dots$ be the $\alpha$ in each step, then $\sum_r (\alpha_r-1) \le 3N$ and in the end $k\le \prod_r(6(\alpha_r-1)+1)$. By \Cref{prop:math1}, $k \le 7^{3N}$.
\end{proof}

Consider a steep segment $[r_1,r_2]$ given by \Cref{clm:partition}. 

If $r_1>D$, $r_2<L-D$, and $p(v_{r_2})[x] - p(v_{r_1})[x]>0$. We denote $\alpha = \card{p(v_{r_2})[x] - p(v_{r_1})[x]}$, and denote $\beta_1$ and $\beta_2$ be the maximum difference of the groups that contains $p(v_{r_1})[x]$ and $p(v_{r_2})[x]$, respectively.  Let $\gamma = \max_{q,q' \in [r_1-k,r_2+3k]}\{p(v_{q})[x] - p(v_{q'})[x]\}$.  Consider any $r_1 - k \le r \le r_1$ and $\Delta = 2k$. We have $\card{p(v_r)[x]-p(v_{r_1})[x]} < \beta_1$, $\card{p(v_{r+\Delta})[x]-p(v_{r_2})[x]} \le \beta_2$ and $\card{p(v_{r+2\Delta})[x]-p(v_{r_2})[x]} \le \beta_2$. So $\card{2p(v_{r+\Delta})[x]-p(v_r)[x]-p(v_{r+2\Delta})[x]} \le \alpha - \beta_1 - 3\beta_2$. By \Cref{clm:triangle}, these $r$ contribute $\gamma$ to $\ell[i,j]$ and contribute $k(\alpha - \beta_1 - 3\beta_2 - \gamma)$ to $\sum_{s=0}^{k-1}(i+2s,j)$. We say this is the contribution of the segment. 

If $r_1 \le D$ or $r_2 \ge L-D$, by similar argument, the contribution of the segment is $k(\alpha - \beta_1 - 3\beta_2 - \gamma - \eps)$ since otherwise $\ell[i,j]>\eps$. Note that these contribution does not overlap because each smooth segment has length at least $5k$.

We say that a group is covered by a segment $[r_1,r_2]$, iff all its values are between $p(v_{r_1})[x]$ and $p(v_{r_2})[x]$. Consider an arbitrary minimal set of segments that cover all groups. By \Cref{obs:small}, the sum of the $\alpha$ is at least $0.1-\eps \ge 9 \eps$, the sum of $\beta_1+3\beta_2$ is at most $4\eps$. If the sum of the $\gamma$ is at least $\eps$, then $\ell[i,j] \ge \eps$ and \Cref{lem:line} when $k=0$, and otherwise the total contribution of these segments are at least $k(9\eps - 4 \eps -2 \eps - \eps) = 2k \eps$, which again means that \Cref{lem:line} is true.

\subsection{Proof of \Cref{thm: lower 6} for $\zesn_{\textsf{ave}}$}
\label{sec: ave}

In \Cref{subsec: onto plane,subsec: basic,subsec: plane}, we proved \Cref{thm: lower 6} for problem $\zesn$.
In this section, we provide the proof of \Cref{thm: lower 6} for problem $\zesn_{\textsf{ave}}$ (the ``average'' version of problem $\zesn$), by generalizing the hard instances and its analysis in previous subsections.

Recall that in $\zesn$, the input is a weighted graph $G$ and a terminal set $T$, and a feasible solution is a pair $(\fset,\delta)$ with $\delta(F(t),F(t'))\ge \dist_{G}(t,t')$ for all $t,t'$. While in $\zesn_{\textsf{ave}}$, the input is $(G,T,\dset)$, and a feasible solution is a pair $(\fset,\delta)$ with $\sum_{t,t'}\dset(t,t')\cdot\delta(F(t),F(t'))\ge \sum_{t,t'} \dset(t,t')\cdot\dist_{G}(t,t')$.
Intuitively, a solution $(\fset,\delta)$ to $\zesn_{\textsf{ave}}$ can manipulate the distances between terminal clusters, as long as the weighted sum $\sum_{t,t'}\dset(t,t')\cdot\delta(F(t),F(t'))$ does not decrease. Therefore, in order to prove that every $\zesn_{\textsf{ave}}$ solution has a large cost, we first need to ``freeze'' the distances between terminal clusters as in $\dist_{G}(\cdot,\cdot)$, and then follow the analysis for $\zesn$.


Towards this goal, we need to slightly modify the hard instance defined in \Cref{subsec: instance}. The metric $D$ (on terminals $\set{a,b,c,d,e,f}$) stays the same, and we keep all paths defined in $G$. 
In addition, we add, for any three terminals $t,t',t''$ with $D(t,t'')=D(t,t')+D(t',t'')$, a $3$-vertex path $(t,t',t'')$ with weight $L^2$. 
For example, since $D(a,e)+D(e,f)=D(a,f)$, we add a weight-$L^2$ path $(a,e,f)$ to the graph. 
As we are now constructing an instance of $\zesn_{\textsf{ave}}$, we also need to define a demand $\dset$ on terminal pairs. For every pair of terminals, we set the demand between them as the total weight of paths between them. We also add a demand of $\gamma L^2$ between $a$ and $e$ where $\gamma = 10^{-15}$. 

We first ignore the last demand and prove a lower bound for this demand. We will show that there is only one way that the $0$-extension solution has the same cost as the original graph, and then show that adding the last demand will prevent that.

We now start analyzing the cost of any solution $(\fset,\delta)$ to the modified instance.
First, via similar analysis, we can show that
$\opt  :=\sum_{(u,v)\in E(G)}\dist_{\ell}(u,v)\le 200\cdot L^2$.
Denote by $\pset$ the collection of all paths in the modified instance. For each $P\in \pset$, we denote by $t_P,t'_P$ the endpoints of $P$, and define the \emph{loss of path $P$} as $\ell(P)=\vol(P)-\delta(t_P,t'_P)$ (where $\delta(t,t')=\delta(F(t),F(t'))$, abusing the notation).
Therefore, 
$$\vol(\fset,\delta)-\opt= \sum_{P\in \pset}\ell(P),$$ and to complete the proof of \Cref{thm: lower 6}, it suffices to show that $\sum_{P\in \pset}\ell(P)\ge \gamma\cdot L^2$.


Let $\eta = 10^{-9}$. We start with the following immediate observation.

\begin{observation} \label{obs:equ}
For any three terminals $t_1,t_2,t_3$ such that $D(t_1,t_2)+D(t_2,t_3) = D(t_1,t_3)$. If $\delta(t_1,t_2)+\delta(t_2,t_3) \ge \delta(t_1,t_3) + \eta$, then $\sum_{P} \ell(P) \ge \eta\cdot L^2$.
\end{observation}
\begin{proof}
From the construction, there is a path $P=(t_1,t_2,t_3)$ of weight $L^2$. The loss of this path is $\ell(P)\ge \delta(t_1,t_2)+\delta(t_2,t_3)-\delta(t_1,t_3) \ge \eta$, so in total they contribute $\eta\cdot L^2$ to $\vol(\fset,\delta)-\opt$.
\end{proof}

As a corollary, $\vol(\fset,\delta)\ge (1+10^{-11})\cdot \opt$.

From now on we assume that, for all terminals $t_1,t_2,t_3$ such that $D(t_1,t_2)+D(t_2,t_3) = D(t_1,t_3)$, $\delta(t_1,t_2)+\delta(t_2,t_3) \le \delta(t_1,t_3) + \eta$ holds, as otherwise we are done by \Cref{obs:equ}. We say that a solution $(\fset,\delta)$ is \emph{good} if all these requirements are satisfied.

As a next step, we show that, any good solution $(\fset,\delta)$ can be slightly adjusted such that all these ``almost tight inequalities'' actually become equalities, without significantly increasing its cost.

\begin{lemma} \label{lem:adjust}
For any good solution $(\fset,\delta)$, there is another feasible solution $(\fset',\delta')$, such that
    \begin{itemize}
        \item for any terminals $t_1,t_2,t_3$ with $D(t_1,t_2)+D(t_2,t_3) = D(t_1,t_3)$, $\delta'(t_1,t_2)+\delta'(t_2,t_3) = \delta'(t_1,t_3)$;
        \item $|\fset'|\le |\fset|+6$, and $\vol(\fset',\delta')\le (1+30\eta)\cdot \vol(\fset,\delta)$.
    \end{itemize}
\end{lemma}

\begin{proof}
    We first prove some properties of $\delta$. We say two number $x \sim_{\eta} y$ if $\card{x-y}<\eta$.
    \begin{claim} \label{clm:close}
        $\delta(a,c) \sim_{2\eta} \delta(b,c) \sim_{2\eta} \delta(d,f) \sim_{2\eta} \delta(e,f)$, $\delta(a,f) \sim_{2\eta} \delta(b,f) \sim_{2\eta} \delta(c,d) \sim_{2\eta} \delta(c,e)$.
    \end{claim}
    \begin{proof}
        By the assumption due to \Cref{obs:equ}, we have $\delta(b,c)+\delta(c,d) \sim_{\eta} \delta(b,f)+\delta(f,d)$ and $\delta(b,c)+\delta(b,f) \sim_{\eta} \delta(c,d)+\delta(d,f)$, which means $\delta(b,c) \sim_{\eta} \delta(d,f)$ and $\delta(b,f) \sim_{\eta} \delta(c,d)$. Similarly, we have $\delta(a,f) \sim_{\eta} \delta(c,d)$, $\delta(a,c) \sim_{\eta} \delta(d,f)$, $\delta(b,f) \sim_{\eta} \delta(c,e)$ and $\delta(b,c) \sim_{\eta} \delta(e,f)$.
    \end{proof}

We now proceed to construct a solution $(\fset',\delta')$ that satisfies the requirements in \Cref{lem:adjust}. We first define its distance between terminal clusters. 
Let $A = \delta(b,c) - 3\eta$ and $B = \delta(a,e) - 3\eta$. We set
\[\delta'(a,c)=\delta'(b,c)=\delta'(d,f)=\delta'(e,f)=A,\] \[\delta'(a,e)=B,\]
\[\delta'(a,b)=\delta'(d,e)=2A,\] \[\delta'(a,f)=\delta'(b,f)=\delta'(c,d)=\delta'(c,e)=A+B, \text{ and}\] \[\delta'(a,d)=\delta'(b,e)=\delta'(c,f)=\delta'(b,d) = 2A+B.\]
It is easy to verify that $\delta'$ satisfies the first condition of \Cref{lem:adjust}. Moreoever, we have the following immediate obseration from \Cref{clm:close} and fact that $(\fset,\delta)$ is a good solution.
    \begin{observation} \label{obs:D''}
        For any terminal pair $t_1,t_2$, we have $\delta(t_1,t_2) -10\eta \le \delta'(t_1,t_2) \le \delta(t_1,t_2)$.
    \end{observation}


We next construct $\fset$ and generalize the definition of $\delta'$ to all pairs of sets in $\fset'$.
For each terminal $t\in T$, define vector $\bar t=(\bar t_s)_{s\in T}$ where $\bar t_s=\delta'(t,s)$.
Let $F$ be a set in $\fset$. We denote $x=(x_t)_{t\in T}$ where $x_t=\delta(F,F(t))$, and define $\bar x=(\bar x_t)_{t\in T}$ as the vector such that
\begin{equation}
\bar x_t=\norm{x-\bar t}_{\infty}.\label{eqn: new vector}
\end{equation}
In the next claim, we prove that $\bar x$ is a valid vector for metric $\delta'$ and is not far from $x$.

\begin{observation} \label{obs:x'}
    For any terminals $t_1$ and $t_2$, $\bar x_{t_1}+\delta'(t_1,t_2) \ge \bar x_{t_2}$ and $\bar x_{t_1}+\bar x_{t_2}\ge \delta'(t_1,t_2)$ hold. Moreover, for any terminal $t$, there is a terminal $t'$ such that $\bar x_t = x_t - \delta'(t,t')$ and $x_t \le \bar x_t \le x_t + 10\eta$.
\end{observation}

\begin{proof}
    The first and second inequalities are due to the fact that $\delta'(t_1,t_2) = \norm{\bar t_1-\bar t_2}_{\infty}$. For any terminal $t$, suppose $\norm{x-\bar t}_{\infty} = \card{x_{t'} - \delta'(t,t')}$. By definition of $\norm{\cdot}_{\infty}$ we have $\card{x_{t'}-\delta'(t,t')} \le x_t$. On the other hand, we have $x_t \ge \delta(t,t') - x_{t'}$ and by \Cref{obs:D''}, we have $\delta'(t,t') \le \delta(t,t')$, thus $x_t \ge \delta'(t,t') - x_{t'}$. So either $\bar x_t = \norm{x-t}_{\infty} = x_t$, or we have $x_{t'} \ge \delta'(t,t')$ and $\bar x_t = x_{t'} - \delta'(t,t')$. In both case, $\bar x_t = x_t - \delta'(t,t')$. Again by \Cref{obs:D''}, we have $\bar x_t \le x_{t'} - \delta(t,t') + 10\eta \le x_t + 10\eta$. Thus in both cases, we have $\card{\bar x_t-x_t} \le 10 \eta$.
\end{proof}

We define the distance in $\delta'$ between two sets $F,F'$ with vectors $\bar x$ and $\bar y$ respectively as $\norm{\bar x-\bar y}_{\infty}$. 

    \begin{observation} \label{obs:x'y'}
        $\norm{\bar x-\bar y}_{\infty} \le \norm{x-y}_{\infty}$.
    \end{observation}
    \begin{proof}
        For any $t$, without lose of generality assume $\bar x_t\ge \bar y_t$. Assume $\bar x_t = x_{t'}-\delta'(t,t')$ by \Cref{obs:x'}. Then by \Cref{obs:x'},
        $$\bar x_t - \bar y_t = x_{t'} - \delta'(t,t')-\bar y_t \le x_{t'} - \delta'(t,t') - (y_{t'} - \delta'(t,t')) \le x_{t'} - y_{t'}.$$
        Therefore for any terminal $t$, $\card{\bar x_t-\bar y_t} \le x_{t'}-y_{t'} \le \norm{x-y}_{\infty}$.
    \end{proof}


We now complete the construction of $(\fset',\delta')$.
We start from the collection $\fset$. For each terminal $t\in T$, we create a singleton set $\set{t}$ and replace the old set $F(t)\in \fset$ that contains it by $F(t)\setminus \set{t}$. Clearly, the resulting collection, which we denote by $\fset'$, is still a partition of all vertices in $G$, and $|\fset'|\le |\fset|+|T|=|\fset|+6$. For metric $\delta'$, 
\begin{itemize}
\item for every pair $F,F'\in \fset'$ that does not contain any terminals, $\delta'(F,F')=\norm{\bar x-\bar y}_{\infty}$, where vectors $\bar{x},\bar{y}$ are defined according to \ref{eqn: new vector} and their original vectors $x,y$;
\item for every pair $t,t'$ of terminals, the distance $\delta'(\set{t},\set{t'})$ has been defined after \Cref{clm:close}; and
\item for a terminal $t$ and a set $F\in \fset'$, the distance $\delta'(\set{t},F)$ is simply defined as $\bar x_t$, where $\bar{x}$ is defined according to \ref{eqn: new vector}.
\end{itemize}
From \Cref{obs:D''}, \Cref{obs:x'} and \Cref{obs:x'y'}, the only edges whose length are increased are edges between a terminal and a non-terminal. Each path in $G$ contains at most $2$ such edges, which means its cost is increased by at most $20 \eta$ by \Cref{obs:x'}. Thus the total cost is increased by at most a factor of $(1+10\eta)$. Finally, we note that the weighted average distance between the terminals in $\delta'$ might not equal the weighted average distance in $\delta$, as we have decreased the distance between some terminal pairs. Thus we need to normalize the distance between all pairs of points so that the weighted average distance is the same. We denote the new metric as $\delta^*$. By \Cref{obs:D''}, the normalizing factor is again at most $(1+10\eta)$. Altogether, the cost of the final solution $(\fset',\delta^*)$ is at most $(1+10\eta)^2 \le (1+30\eta)$ times the cost of $(\fset,\delta)$.
\end{proof}


Consider now the solution $(\fset',\delta')$ given by \Cref{lem:adjust}. As the metric $\delta'$ satisfies the first condition in \Cref{lem:adjust}, the restriction of $\delta'$ onto $T$ has the same tight span structure as that of $D$ (via similar analysis in \Cref{apd: structure of TS}). Therefore, via similar analysis, we can show that if $\delta'(a,e)>1/2$, then $\vol(\fset',\delta')\ge (1+10^{-6})\cdot \opt$, and consequently
\[\vol(\fset,\delta)\ge (1+30\eta)^{-1}\cdot (1+10^{-6})\cdot \opt\ge (1+10^{-11})\cdot\opt.\]
Together with \Cref{obs:equ}, in this case we get that $\vol(\fset,\delta)\ge (1+10^{-11})\cdot \opt$.

In order to handle the case where $\delta'(a,e) \le 1/2$, we now consider the last demand, which is $\gamma L^2$ unit between $a$ and $e$. Note that $\delta(a,e)=1$.
If $\delta'(a,e) \le 1/2$, then the average length of the paths in $G$ in $\delta$ is increased by a $(1+\gamma/200)$-factor, which means that $\vol(\fset',\delta')\ge (1+\gamma/200)\cdot \opt$, and therefore 
\[\vol(\fset,\delta)\ge (1+\gamma/200)\cdot \opt\ge (1+10^{-18})\cdot\opt.\]

Assume now that $\delta'(a,e) \ge 1/2$.
\begin{itemize}
\item if $\delta'(a,e)>10^3$, note that we have $L^2$ paths between $(a,b)$ and $(b,e)$, and the total cost of these paths is at least $10^3L^2$, which is more than $2\cdot \opt$;
\item if $1/2 \le \delta'(a,e) \le 10^3$, the average path length is at least $(1-10^3\gamma)$ factor as before. By the analysis before, the cost of the solution is at least \[\vol(\fset,\delta)\ge (1-10^3\gamma)\cdot(1+10^{-11})\cdot \opt  \ge (1+10^{-18})\cdot\opt.\]
\end{itemize}

\paragraph{Acknowledgement.} We would like to thank Julia Chuzhoy for introducing this problem to us and for helpful discussions. We also thank Arnold Filtser for the information on previous works on this problem.

\appendix

\section{Comparison between Variants of $0$-Extension Problems}

\label{apd: comparison}

In the classic 0-Extension problem \cite{karzanov1998minimum}, the input consists of 
\begin{itemize}
	\item an edge-capacitated graph $G=(V,E,c)$; 
	\item a subset $T\subseteq V$ of $k$ vertices, that we call \emph{terminals}; and
	\item a metric $D$ on terminals in $T$. 
\end{itemize}
A solution to the 0-Extension problem is a partition $\fset$ of the vertex set $V$ into $|T|$ subsets, such that distinct terminals of $T$ belong to different sets in $\fset$; we call sets in $\fset$ \emph{clusters}, and for each vertex $u\in V$, we denote by $t(u)$ terminal that lies in the same cluster as $u$.
The goal is to minimize the cost \underline{$\vol(\fset)=\sum_{(u,v)\in E(G)}c(u,v)\cdot D(t(u),t(v))$}.

A natural approach towards constructing low-cost solutions is by considering the following semi-metric relaxation LP.
\begin{eqnarray*}
	\mbox{(\textsf{semi-metric LP})}\quad &  \text{minimize}\quad  \sum_{(u,v)\in E}c(u,v)\cdot \delta(u,v)\\
	& s.t. \quad (V,\delta) \text{ is a semi-metric space}\\
	& \delta(t,t')=D(t,t'), \quad\forall t,t'\in T
\end{eqnarray*}

Assume that we have solved the LP relaxation, and now want to round the obtained solution into a feasible solution to the 0-Extension problem. The \emph{rounding problem} that we face can be formulated as follows.
The input consists of 
\begin{itemize}
	\item an edge-capacitated graph $G=(V,E,c)$; 
	\item a subset $T\subseteq V$ of $k$ vertices, that we call \emph{terminals}; and
	\item lengths $\{\ell_e\}_{e\in E}$ of edges of $G$. 
\end{itemize}
Recall that $\dist_{\ell}(\cdot,\cdot)$ is the shortest-path distance metric on $V$ induced by the lengths $\{\ell_e\}_{e\in E}$.
A solution is a partition $\fset$ of the vertex set $V$ into $|T|$ subsets, such that distinct terminals of $T$ belong to different sets in $\fset$.
The goal is to minimize the cost \underline{$\vol(\fset)=\sum_{(u,v)\in E(G)}c(u,v)\cdot \dist_{\ell}(t(u),t(v))$}. In particular, the ratio between $\vol(\fset)$ and the value $\sum_{(u,v)\in E(G)}c(u,v)\cdot\ell(u,v)$ is a central measure to be minimized, which is called the \emph{average stretch}.

Intuitively, initially each edge in $G$ has length $\ell(e)$, which induces a shortest-path distance metric on the terminals. The rounding problem aims to find a way of ``moving'' all non-terminals to terminals, so each edge $(u,v)\in E(G)$, assuming $u$ is moved to $t(u)$ and $v$ is moved to $t(v)$, is stretched to an edge connecting $t(u)$ and $t(v)$, and will therefore have resulting length $\dist_{\ell}(t(u),t(v))$. The cost of such a moving schedule is simply the total resulting length of all edges, while $\sum_{(u,v)\in E(G)}\ell(u,v)$ can be viewed as the cost of the original graph.

$\ $

The most standard ``Steiner node version'' of the classic $0$-Extension problem should be as follows. Its input consists of 
\begin{itemize}
	\item an edge-capacitated graph $G=(V,E,c)$; 
	\item a subset $T\subseteq V$ of $k$ vertices, that we call \emph{terminals}; and
	\item a metric $D$ on terminals in $T$. 
\end{itemize}
A solution consists of
\begin{itemize}
	\item a partition $\fset$ of $V$, such that distinct terminals of $T$ belong to different sets in $\fset$; for each vertex $u\in V$, we denote by $F(u)$ the cluster in $\fset$ that contains it;
	\item a semi-metric $\delta$ on the clusters in $\fset$, such that \underline{for all pairs $t,t'\in T$, $\delta(F(t),F(t'))= D(t,t')$}.
\end{itemize}
The cost of a solution $(\fset,\delta)$ is defined as $\vol(\fset,\delta)=\sum_{(u,v)\in E}c(u,v)\cdot\delta(F(u),F(v))$, and its \emph{size} is defined as $|\fset|$.
The goal is to compute a solution $(\fset,\delta)$ with small size and cost. 

Similarly, we can write down a semi-metric LP relaxation for this generalized version of $0$-Extension, \underline{the $\zesn$ problem defined in \Cref{sec: tech_overview} and \Cref{sec: variants} is the \emph{rounding problem} of this version}.

Comparing their rounding problems, the $\zesn$ problem is essentially generalized from the classic 0-Extension problem in two aspects. First, we do not enforce all non-terminals to be moved to terminals, but allow them to be moved to some other ``Steiner nodes''. Specifically, in $\zesn$ we allow that some sets in $\fset$ do not contain any terminal (and therefore $|\fset|$ can be strictly greater than $|T|$), and so when such a set is contracted, they form a Steiner node in the sparsifier. Second, observe that we can equivalently view a solution to the 0-Extension problem as a pair $(\fset,\delta)$ instead of $\fset$, whose cost is defined as $\vol(\fset,\delta)=\sum_{(u,v)\in E(G)}c(u,v)\cdot \delta(F(u),F(v))$, same as $\zesn$. Only, in $0$-Extension $\delta$ has to be exactly the shortest-path distance metric $\dist_{\ell}(\cdot,\cdot)$ on $T$, while in $\zesn$, we only require that the restriction of $\delta$ onto $T$ dominates $\dist_{\ell}(\cdot,\cdot)$ entry-wise, and in $\zesn_{\textsf{ave}}$, it is required that
some average distance between terminal clusters (clusters that contain a terminal) in $\delta$ is at least the corresponding average distance between terminal in $\dist_{\ell}(\cdot,\cdot)$, an arguably weaker condition.



\section{Missing Proofs}

\subsection{Proof of \Cref{clm: feasible}}
\label{apd: Proof of clm: feasible}

	First we verify that, throughout the algorithm, all inequalities in $\set{\hat x_t+\hat x_{t'}\ge D(t,t')}_{t,t'}$ are satisfied. This is true for the initial $\hat x=x$. Now in each iteration, consider a pair $t,t'\in T$.
	\begin{itemize}
		\item If $t,t'$ are both inactive, then we will not update $\hat x_t$ or $\hat x_{t'}$, and the inequality $\hat x_t+\hat x_{t'}\ge D(t,t')$ continues to hold.
		\item If $t,t'$ are both active, then in this iteration $\Delta\le \Delta_{t,t'}=\frac{1}{2}\cdot (\hat x_t+\hat x_{t'}-D(t,t'))$, so after this iteration the new $\hat x_t$ and $\hat x_{t'}$, which we denote by $\hat x'_t$ and $\hat x'_t$, satisfy that
		\[\hat x'_t+\hat x'_{t'}= \hat x_t+\hat x_{t'}-2\Delta\ge \hat x_t+\hat x_{t'}-2\Delta_{t,t'}=  x_t+\hat x_{t'}-2\cdot \frac{1}{2}\cdot \bigg(\hat x_t+\hat x_{t'}-D(t,t')\bigg)\ge D(t,t').\]
		\item If $t$ is active and $t'$ is inactive (the case where $t'$ is active and $t$ is inactive is symmetric), then (i) in this iteration $\Delta\le \Delta_{t,t'}=(\hat x_t+\hat x_{t'}-D(t,t'))$; and (ii) after this iteration the new $\hat x_t$ and $\hat x_{t'}$, which we denote by $\hat x'_t$ and $\hat x'_t$, satisfy that $\hat x'_t= \hat x_t-\Delta$ and $\hat x'_{t'}= \hat x_{t'}$. Therefore,
		\[\hat x'_t+\hat x'_{t'}=\hat x_t+\hat x_{t'}-\Delta\ge \hat x_t+\hat x_{t'}-\Delta_{t,t'}=  x_t+\hat x_{t'}-\bigg(\hat x_t+\hat x_{t'}-D(t,t')\bigg)\ge D(t,t').\]
	\end{itemize}
	
	Next we show that, for each $t\in T$, some inequality $\set{\hat x_t+\hat x_{t'}\ge D(t,t')}_{t'}$ is tight for $\hat x$ at the end of the algorithm. Consider the iteration when $t$ becomes inactive. From the algorithm, there exists $t'$ such that $\Delta=\Delta_t=\Delta_{t,t'}$. Then:
	
	\begin{itemize}
		\item If $t'$ was active before this iteration, then in this iteration $\Delta= \Delta_{t,t'}=\frac{1}{2}\cdot (\hat x_t+\hat x_{t'}-D(t,t'))$, and so after this iteration, the new $\hat x_t$ and $\hat x_{t'}$, which we denote by $\hat x'_t$ and $\hat x'_t$, satisfy that
		\[\hat x'_t+\hat x'_{t'}= \hat x_t+\hat x_{t'}-2\Delta= \hat x_t+\hat x_{t'}-2\Delta_{t,t'}=  x_t+\hat x_{t'}-2\cdot \frac{1}{2}\cdot \bigg(\hat x_t+\hat x_{t'}-D(t,t')\bigg)= D(t,t').\]
		\item If $t'$ was inactive before this iteration, then $\Delta=\Delta_{t,t'}=(\hat x_t+\hat x_{t'}-D(t,t'))$, and after this iteration the new $\hat x'_t$ and $\hat x'_t$, satisfy that $\hat x'_t= \hat x_t-\Delta$ and $\hat x'_{t'}= \hat x_{t'}$. Therefore,
		\[\hat x'_t+\hat x'_{t'}=\hat x_t+\hat x_{t'}-\Delta= \hat x_t+\hat x_{t'}-\Delta_{t,t'}=  x_t+\hat x_{t'}-\bigg(\hat x_t+\hat x_{t'}-D(t,t')\bigg)= D(t,t').\]
	\end{itemize}
	In both cases, after this iteration, $t$ and $t'$ become inactive, so the coordinates $\hat x_t,\hat x_{t'}$ are no longer updated, implying that the equality $\hat x_t+\hat x_{t'}= D(t,t')$ continues to hold for the following iterations and for the ending $\hat x$.

\subsection{Completing the proof of \Cref{thm: upper 5}}
\label{apd: type 2 and 3}

In this section, we complete the proof of \Cref{thm: upper 5} for the remaining two types of tight spans.

\paragraph{Type 2.} An illustration of a type-$2$ metric is shown in \Cref{fig: type_2}, with its $1$-dimensional sets (line metrics $a$-$a'$,$b$-$b'$,$c$-$c'$,$d$-$d'$,$e$-$e'$) omitted for simplicity.

\begin{figure}[h]
	\centering
	\scalebox{0.1}{\includegraphics{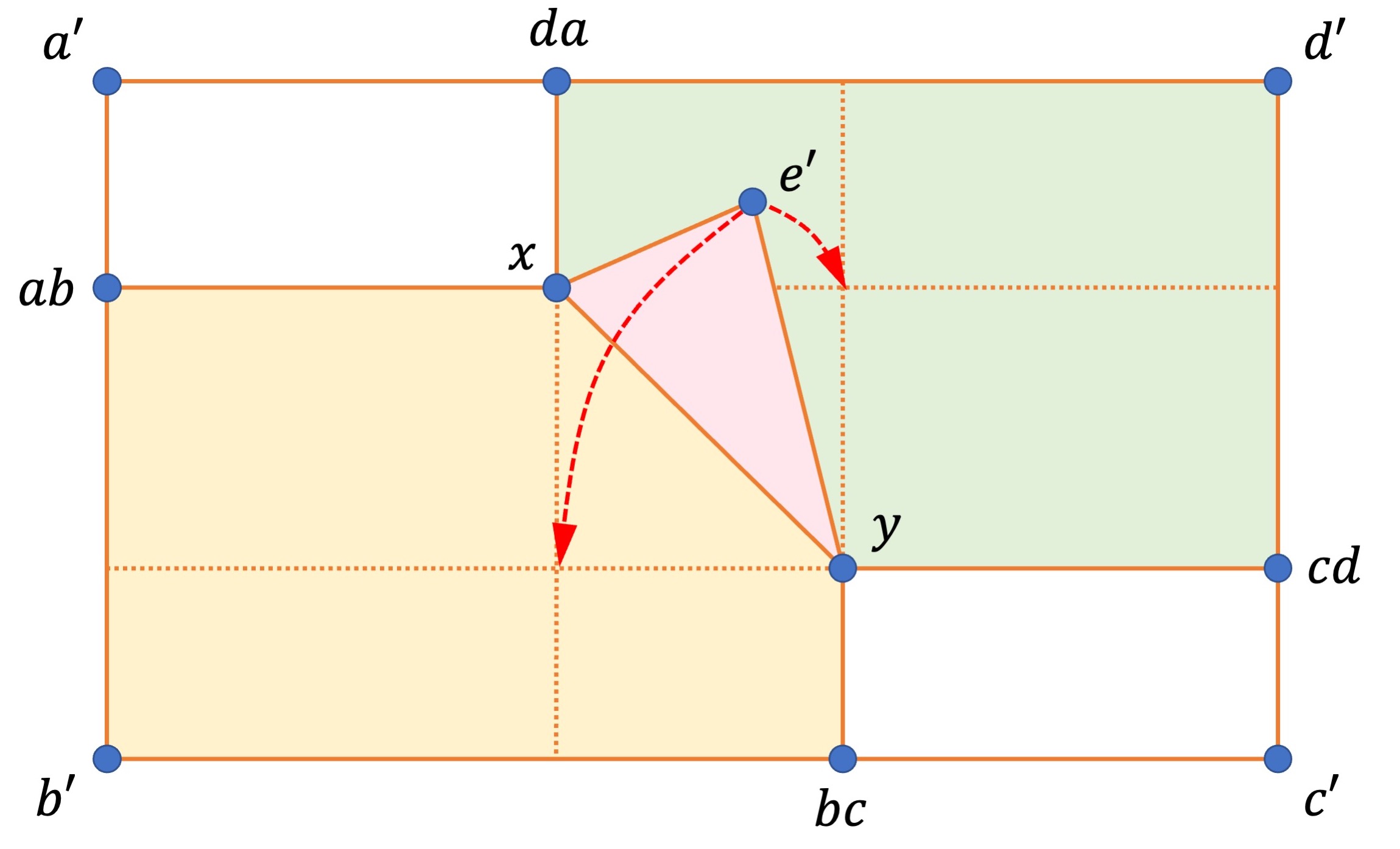}}
	\caption{An illustration of a type-$2$ tight span on $\set{a,b,c,d,e}$ (with pendant line metrics omitted).\label{fig: type_2}}
\end{figure}

The $2$-dimensional sets in this tight span are not as ``non-overlapping'' as they are in type-$1$ tight spans. 
Overall, the shape is the union of a rectangle $\overline{a'b'c'd'}$ and an isosceles right triangle $\overline{xe'y}$, both equipped with $\ell_1$ metric (with $a'b'$ and ${a'd'}$ as axis for the rectangle and $e'x$ and $e'y$ as axis for the triangle).
To better understand its structure, we look at its restriction onto its $4$-point subsets (that is, the tight span of a $4$-point subset of $\set{a,b,c,d,e}$).

\begin{figure}[h]
	\centering
	\subfigure[The restricted tight span on $\set{a,c,d,e}$.]
	{\scalebox{0.096}{\includegraphics{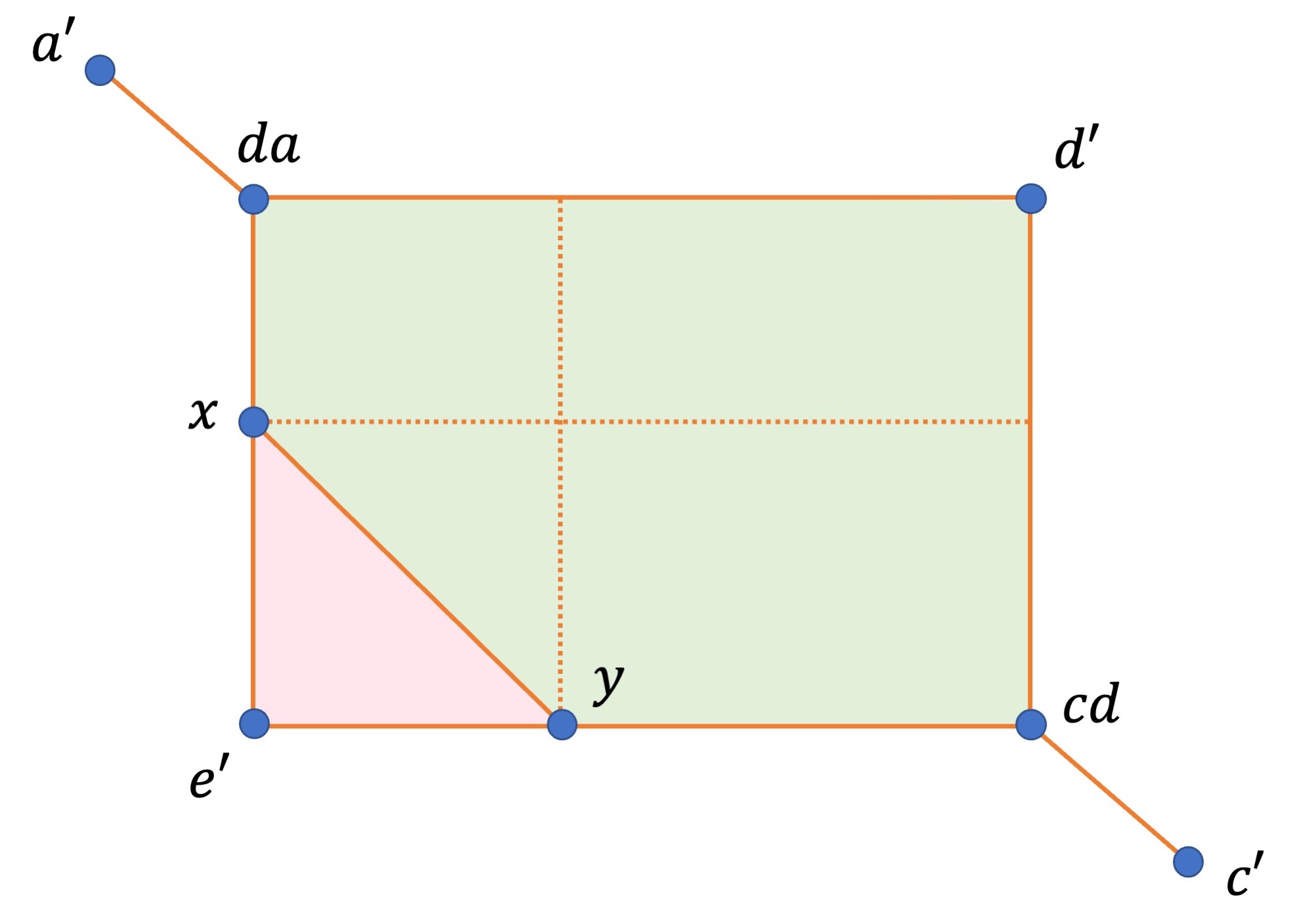}}}
	\hspace{0.5cm}
	\subfigure[The restricted tight span on $\set{a,b,c,e}$.]
	{
		\scalebox{0.1}{\includegraphics{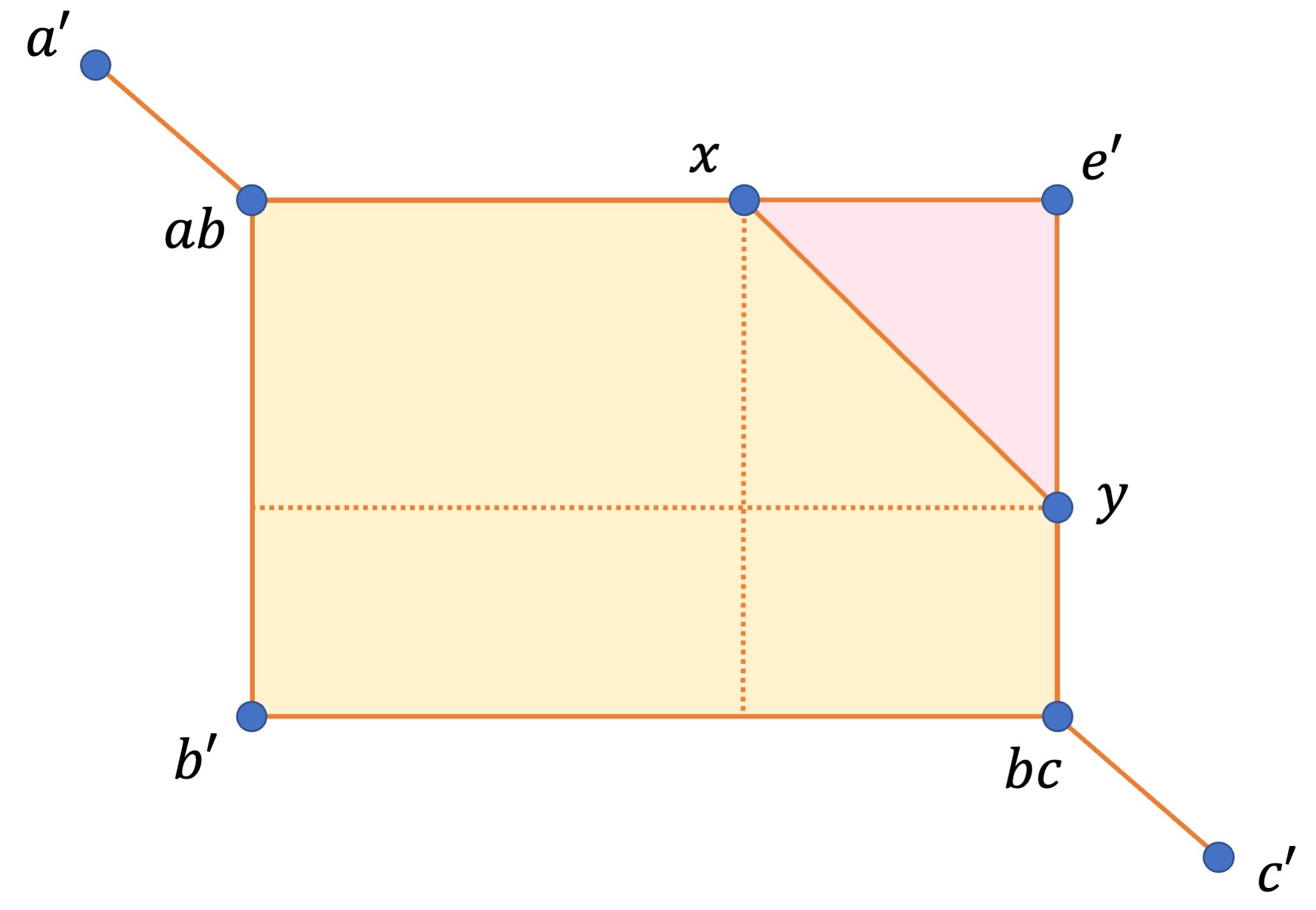}}}
	\caption{Restrictions of a type-$2$ tight span on $\set{a,b,c,d,e}$ onto its $4$-point subsets (pendants omitted).\label{fig: type_2_detail}}
\end{figure}

Its restriction onto $\set{a,b,c,d}$ is the rectangle $\overline{a'b'c'd'}$ (we still omit its $1$-dimensional sets for simplicity). Its restriction onto other $4$-point subsets are less obvious to see, as the triangle $\overline{xe'y}$ need to be ``properly folded into the rectangle'' in order to form the right rectangles.
As an example, the restriction onto $\set{a,c,d,e}$ is the rectangle $da$-$e'$-$cd$-$d'$, where rectangle is completed by folding the triangle to its lower left along line $x$-$y$. Similarly, the restriction onto $\set{a,b,c,e}$ is the rectangle $ab$-$b'$-$bc$-$e'$, where rectangle is completed by folding the triangle to its upper right along line $x$-$y$. See \Cref{fig: type_2_detail} for an illustration. The other two restrictions are similar.

\begin{figure}[h]
	\centering
	\scalebox{0.1}{\includegraphics{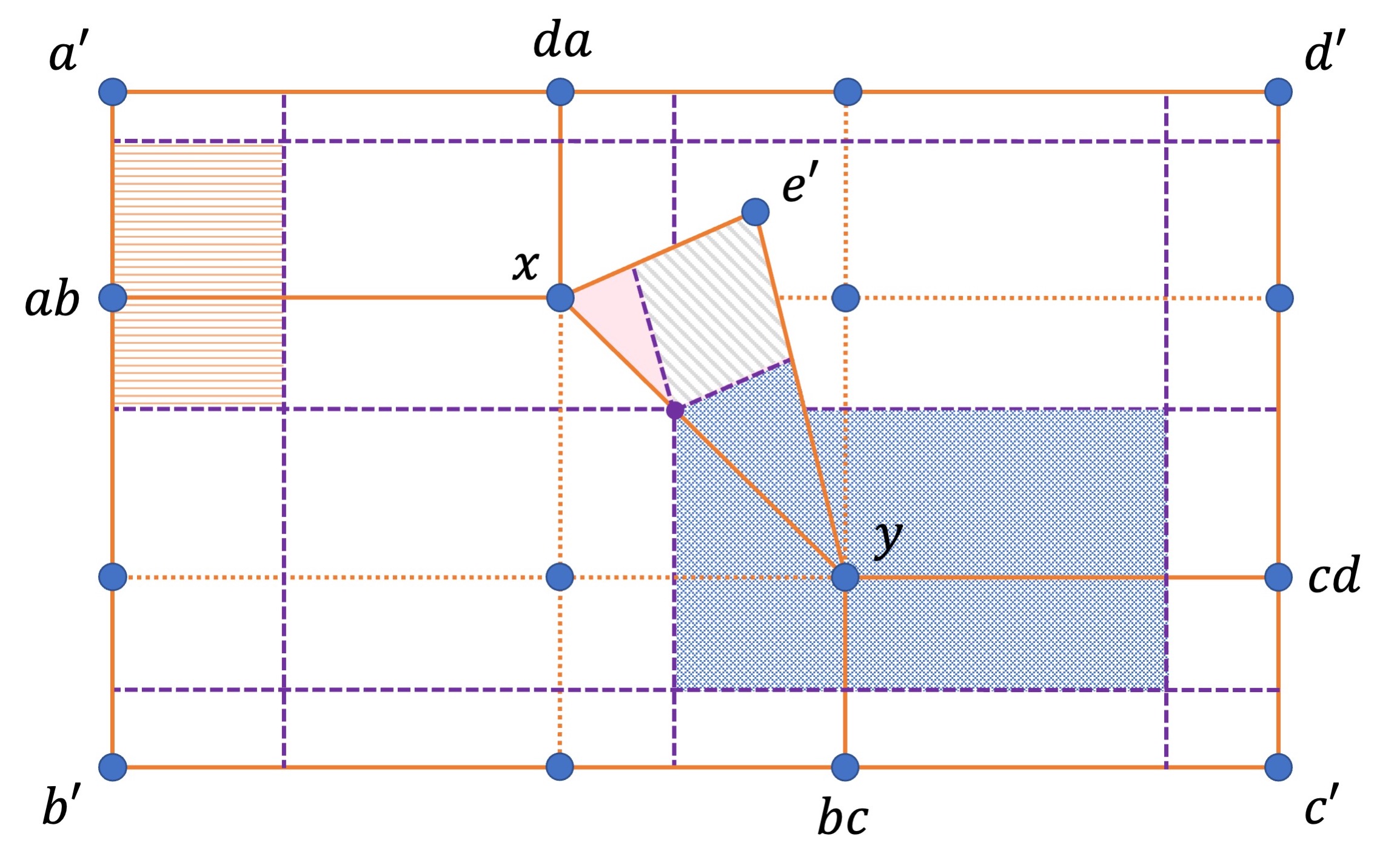}}
	\caption{An illustration of the decomposition of a type-$2$ tight span (and its corresponding partition $\fset$). The regions corresponding to sets $F_{ab},F_{e'},F_{y}$ are shaded.\label{fig: type_2_deco}}
\end{figure}

The construction of solution $(\fset,\delta)$, similar to type-$1$ tight spans, is based on a random decomposition of the tight span. Essentially, we can still prove \Cref{lem: non-expanding} as the metric space $(\TS(D),(\cdot,\cdot)_{\TS})$ is a $\ell_1$ metric.
Here we only describe the decomposition to avoid  redundancy.
For each individual boundary segment, we pick a random point on it and form a partition line with it. For example, the leftmost vertical purple line is determined by a random point on the $a'$-$da$ line, and the bottom horizontal purple line is determined by a random point on the $c'$-$cd$ line.
The horizontal partition line in rectangle $ab$-$cd$ and the vertical partition line in rectangle $da$-$bc$ needs to be ``entangled'', as they both contain line $xy$ from the triangle. So we will sample a point from the $xy$ line (the purple dot in the center), and form all partition lines according to it, including the $ab$-$cd$ horizontal line and the $da$-$bc$ vertical line in rectangle $\overline{a'b'c'd'}$, the $e'y$-parallel line and the $e'x$-parallel line in the triangle.
See \Cref{fig: type_2_deco} for an illustration.
The number of non-terminal sets in the partition $\fset$ is $17$ (all blue nodes), so $|\fset|\le 22$.

\paragraph{Type 3.} An illustration of a type-$3$ metric is shown in \Cref{fig: type_2}, with its $1$-dimensional sets (line metrics $a$-$a'$,$b$-$b'$,$c$-$c'$,$d$-$d'$,$e$-$e'$) omitted for simplicity.

\begin{figure}[h]
	\centering
	\subfigure
	{\scalebox{0.09}{\includegraphics{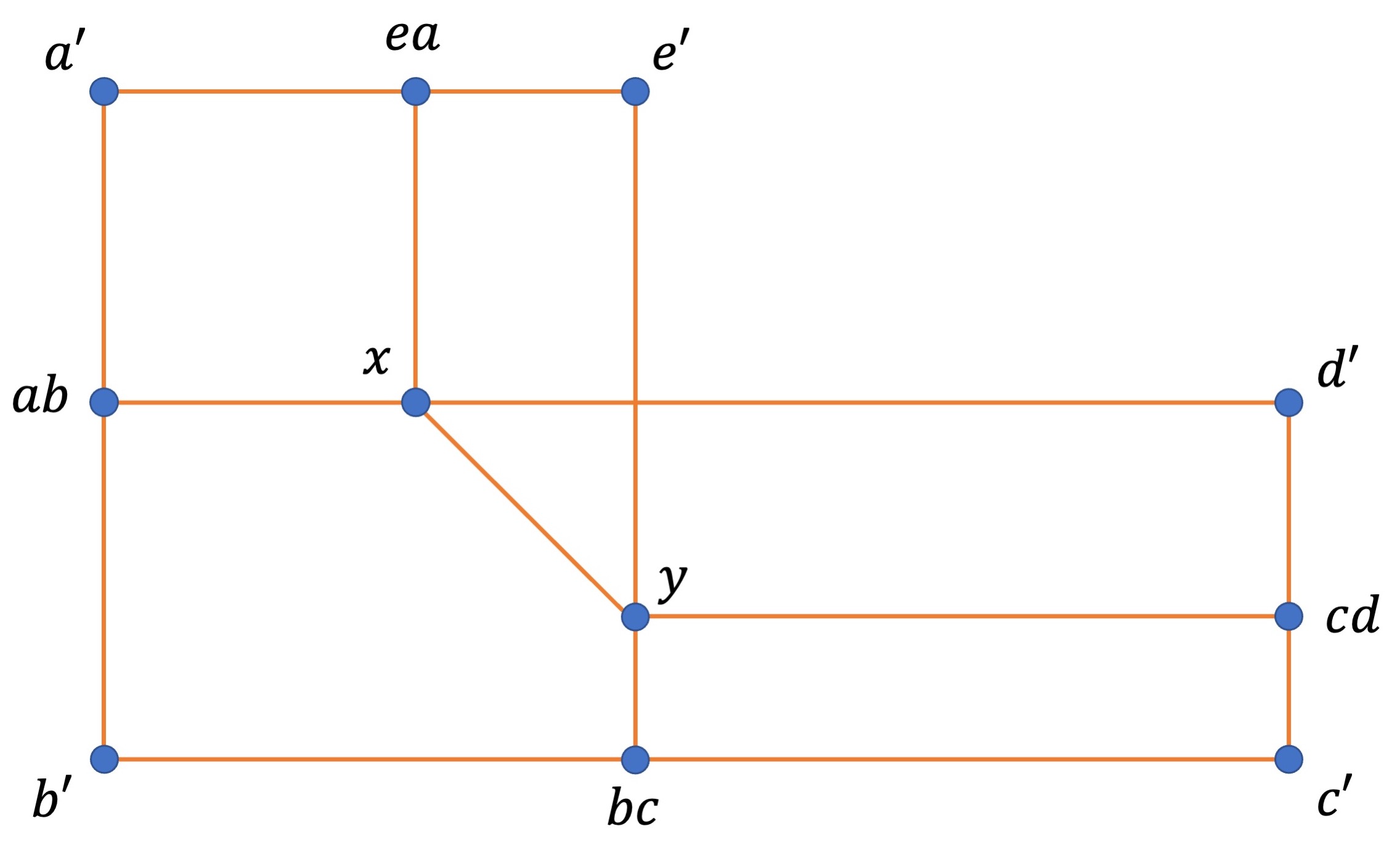}}}
	\hspace{0.5cm}
	\subfigure
	{
		\scalebox{0.09}{\includegraphics{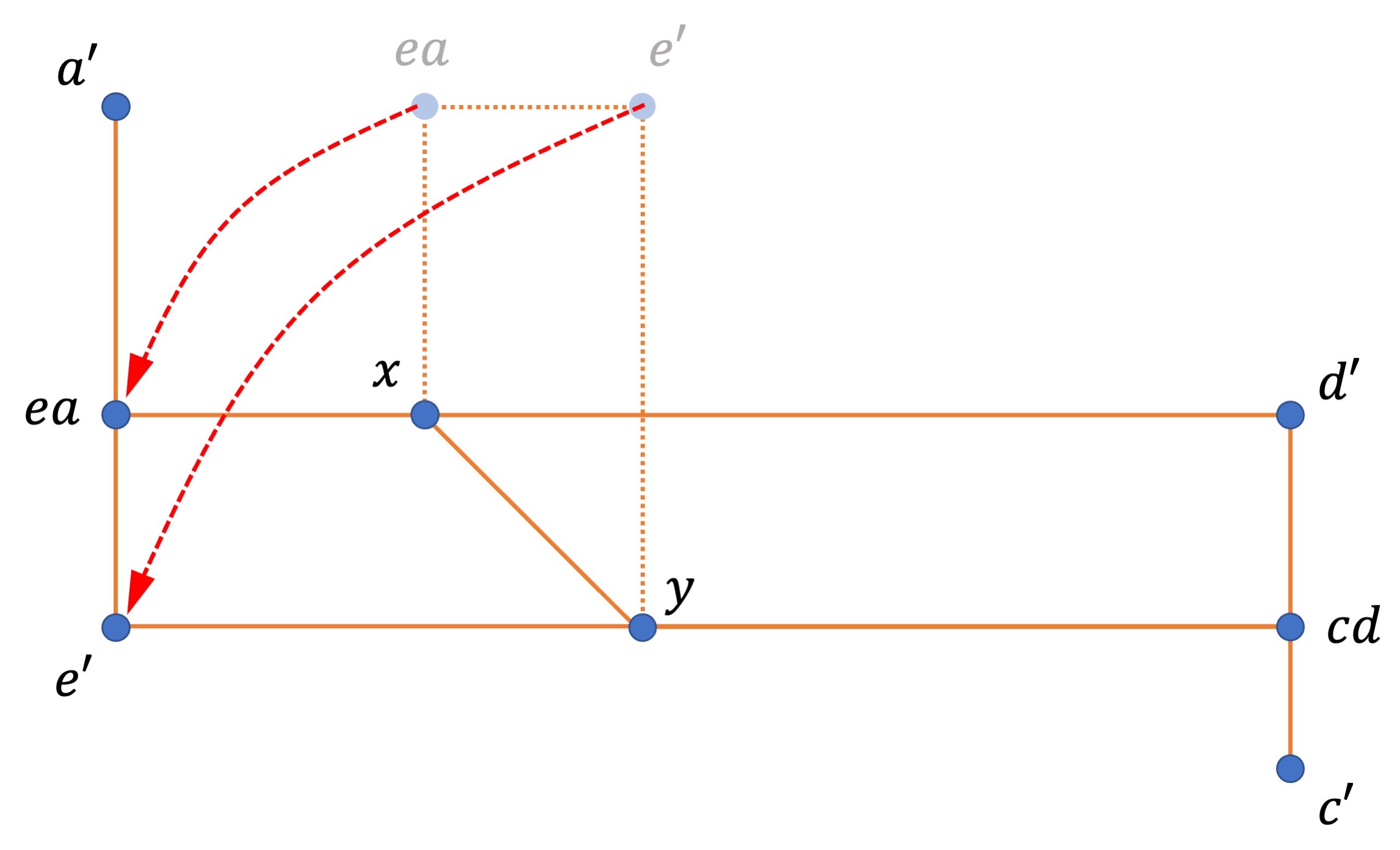}}}
	\caption{An illustration of a type-$3$ tight span on $\set{a,b,c,d,e}$ with pendant line metrics omitted (left), and its restrictions of onto its $4$-point subset $\set{a,c,d,e}$ (right).\label{fig: type_3}}
\end{figure}

Its overall shape is the union of rectangles $a'$-$b'$-$bc$-$e'$ and  $ab$-$b'$-$c'$-$d'$ with their intersection being the $ab$-$x$-$y$-$bc$-$b'$ area,
(that is, the $ea$-$x$-$y$-$e'$ area and the $d'$-$x$-$y$-$cd$ area are disjoint, although they look overlapping in the figure).
Both rectangles are equipped with $\ell_1$ metric.
Its restriction onto 
\begin{itemize}
\item $\set{a,b,c,d}$ is the rectangle $ab$-$b'$-$c'$-$d'$;
\item $\set{a,b,c,e}$ is the rectangle $a'$-$b'$-$bc$-$e'$;
\item $\set{a,b,d,e}$ is the rectangle $a'$-$ab$-$x$-$ea$;
\item $\set{b,c,d,e}$ is the rectangle $y$-$bc$-$c'$-$cd$; and
\item $\set{a,c,d,e}$ is the rectangle $ea$-$e'$-$cd$-$d'$, obtained by folding the area $ea$-$x$-$y$-$e'$ to its lower left along $xy$ line (see \Cref{fig: type_2}).
\end{itemize}

\begin{figure}[h]
	\centering
	\scalebox{0.1}{\includegraphics{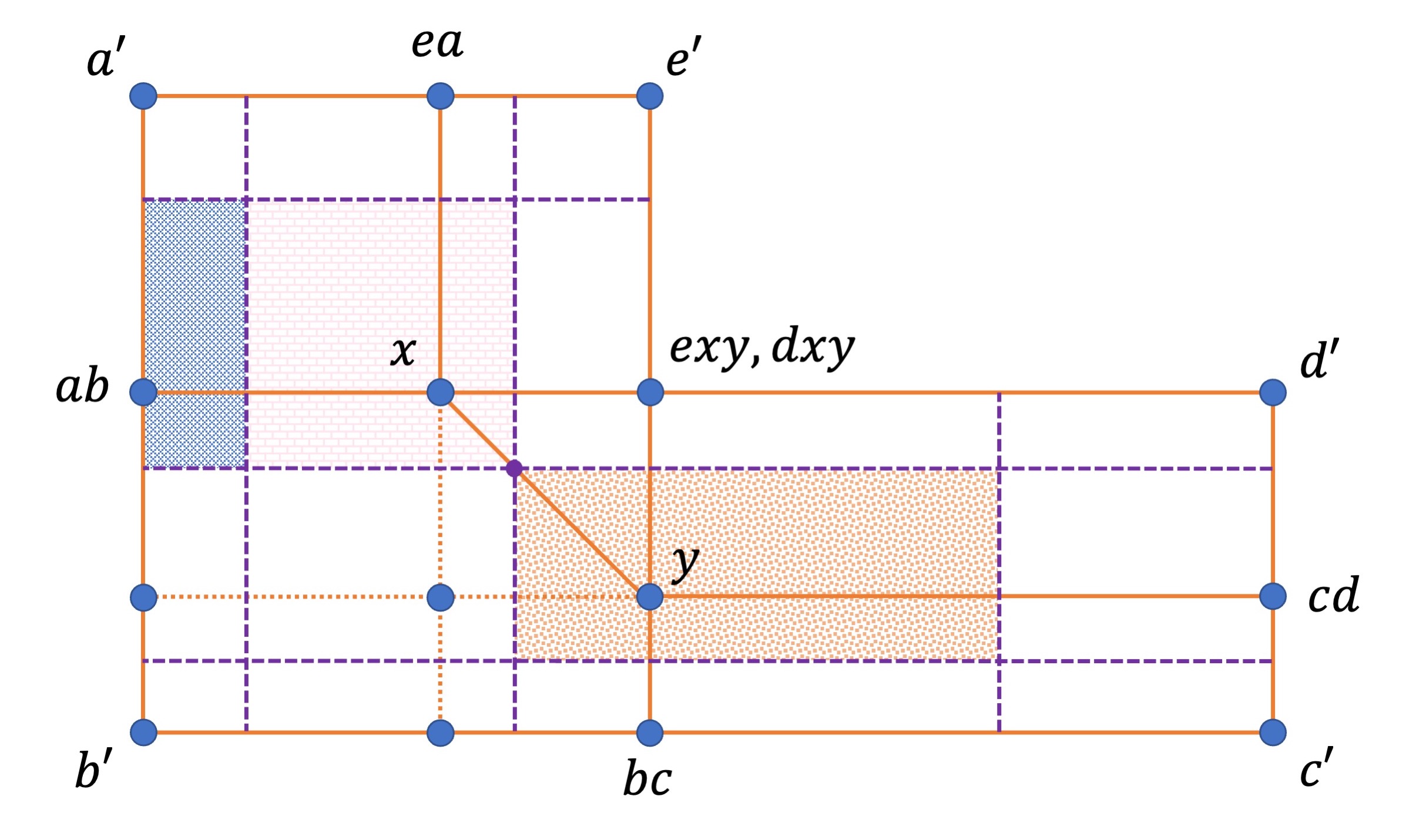}}
	\caption{An illustration of the decomposition of a type-$3$ tight span (and its corresponding partition $\fset$). The regions corresponding to sets $F_{ab},F_x,F_{y}$ are shaded.\label{fig: type_3_deco}}
\end{figure}

The construction of $(\fset,\delta)$ is similar to previous types, as essentially the metric space $(\TS(D),(\cdot,\cdot)_{\TS})$ is a $\ell_1$ metric.
We now describe the decomposition.
We pick a random vertex from $a'$-$ea$ segment to form the leftmost purple line and a random vertex from $c'$-$cd$ segment to form the bottom purple line.
Then we pick a random vertex from $x$-$y$ segment to form the horizontal line in $ab$-$cd$ rectangle and the vertical line in $ea$-$bc$ rectangle. Note that, as the areas $ea$-$x$-$y$-$e'$ and $d'$-$x$-$y$-$cd$ are disjoint, the partition line are also disjoint (although they look overlapping in \Cref{fig: type_3_deco}). A more detailed illustration of the partition $\fset$ is presented in \Cref{fig: type_3_deco_detail}. We can see that number of non-terminal sets in the partition $\fset$ is $16$ (all blue nodes), so $|\fset|\le 21$.

\begin{figure}[h]
	\centering
	\subfigure
	{\scalebox{0.12}{\includegraphics{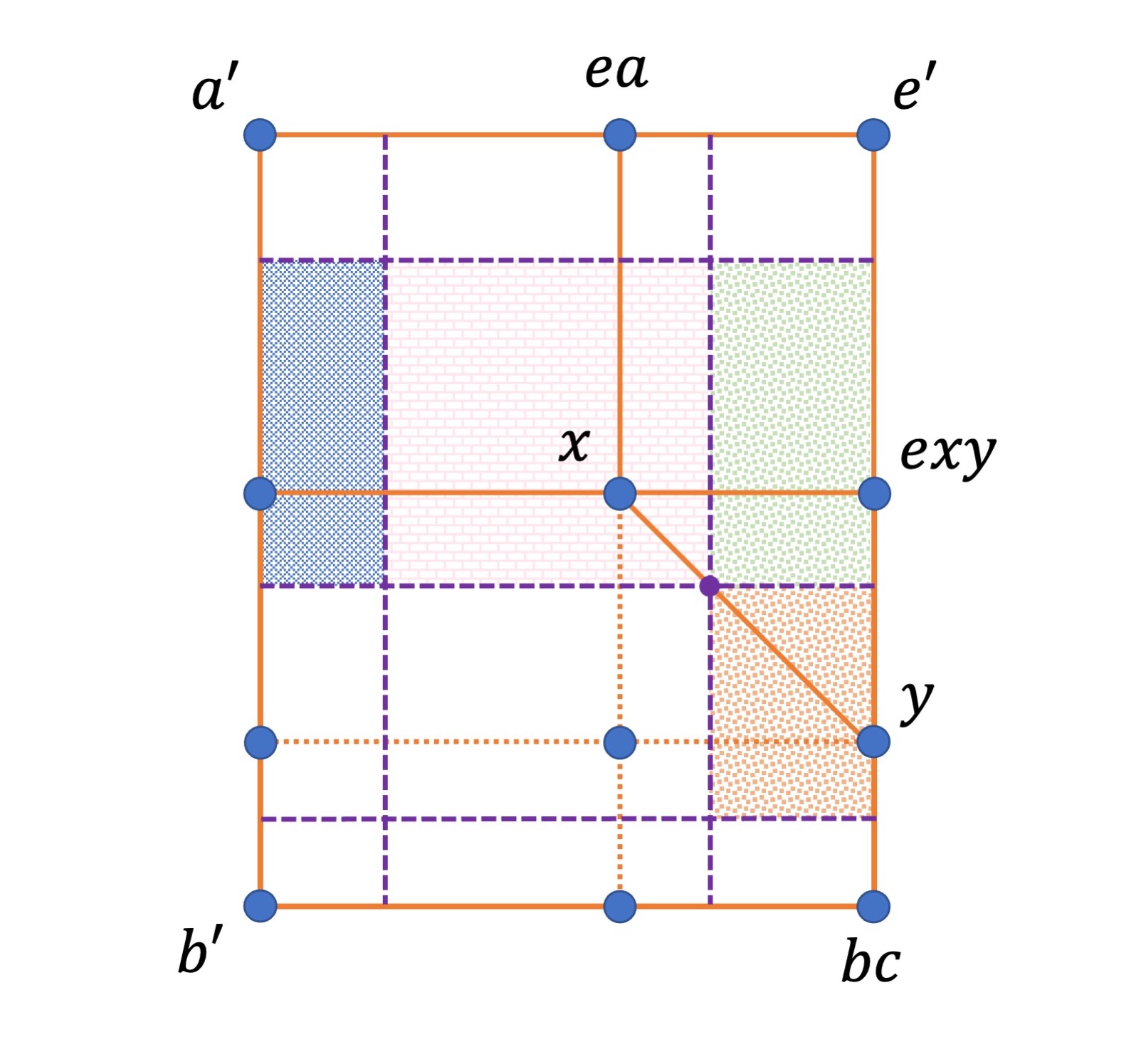}}}
	\hspace{0.5cm}
	\subfigure
	{
		\scalebox{0.1}{\includegraphics{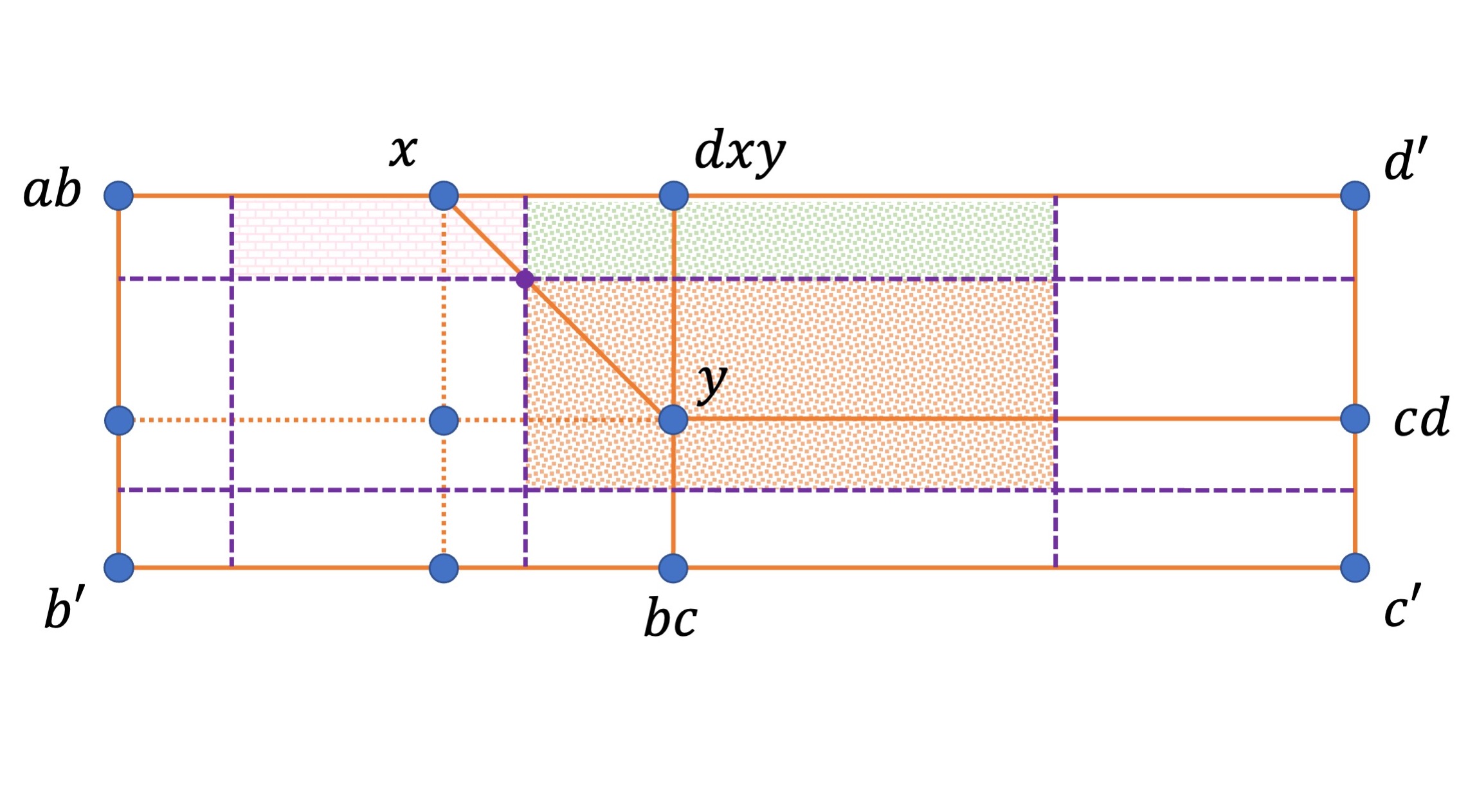}}}
	\caption{The restriction of the decomposition in \Cref{fig: type_3_deco} onto rectangle $a'$-$b'$-$bc$-$e'$ (left) and rectangle $ab$-$b'$-$c'$-$d'$ (right). $F_x$ is the union of two pink regions, $F_y$ is the union of two yellow regions, and $F_{exy}$ and $F_{dxy}$ are shown in green (their regions are disjoint).\label{fig: type_3_deco_detail}}
\end{figure}

\subsection{Determining the structure of $\TS(D)$ in \Cref{subsec: instance}}
\label{apd: structure of TS}

In this section, we explain why the tight span $\TS(D)$ of the metric $D$ is as described \Cref{subsec: instance}.
Recall that the metric $D$ is defined on the terminal set $\set{a,b,c,d,e,f}$, and each point $v\in \TS(D)$ is represented by a $6$-dimensional real vector $v=(v_a,v_b,v_c,v_d,v_e,v_f)$ where $v_t$ indicates the distance from $v$ to terminal $t$.
We start with the following observation.

\begin{observation} \label{obs:tight}
	For any three terminals $t_1,t_2,t_3 \in \set{a,b,c,d,e,f}$ such that $D(t_1,t_2)+D(t_2,t_3)=D(t_1,t_3)$, then for any point $v \in \TS(D)$, if $v_{t_1}+v_{t_2}=D(t_1,t_2)$, then $v_{t_1}+v_{t_3}=D(t_1,t_3)$.
\end{observation}

\begin{proof}
	We have $D(t_1,t_3)-v_{t_1} \le v_{t_3} \le v_{t_2}+D(t_2,t_3) = v_{t_2}+D(t_1,t_3)-D(t_1,t_2) = D(t_1,t_3)-v_{t_1}$, so $v_{t_1}+v_{t_3}=D(t_1,t_3)$.
\end{proof}

Consider now any point $v \in \TS(D)$, by definition there is a terminal $t \in \{a,b,c,d,e,f\}$ such that $v_c+v_t=D(c,t)$. Note that for any $t \in \{a,b,c,d,e,f\}$, we have $D(c,t)+D(t,f)=D(c,f)$, thus by \Cref{obs:tight}, we have $v_c+v_f=D(c,f)$. By the same reason, we have either $v_b+v_e=D(b,e)$ or $v_b+v_d=D(b,d)$, and similarly we have either $v_a+v_d=D(a,d)$ or $v_b+v_d=D(b,d)$. Thus for any $v \in \TS(D)$,  either 
\begin{itemize}
\item equalities $v_c+v_f=D(c,f)$, $v_a+v_d=D(a,d)$ and $v_b+v_e=D(b,e)$ hold, and $v$ is in the triangular prism; or
\item equalities $v_c+v_f=D(c,f)$ and $v_b+v_d=D(b,d)$ hold, and $v$ is in the rectangle. In this case, $v_a = \max\{2-v_b, 3-v_d\}$ and $v_e= \max\{2-v_d,3-v_b\}$.
\end{itemize}

\bibliographystyle{alpha}
\bibliography{REF}

\end{document}